\documentclass[preprint,3p,times]{elsarticle}
\RequirePackage{multirow,booktabs,subfigure,color,array,hhline,makecell}
\usepackage{amssymb}
\usepackage{amsmath}
\usepackage{graphicx}
\usepackage{subfigure}
\usepackage{amsthm}
\usepackage{mathrsfs}
\usepackage{indentfirst}
\usepackage[colorlinks,citecolor=blue,urlcolor=blue]{hyperref}
\allowdisplaybreaks
\usepackage[table]{xcolor}
\usepackage{tikz-network}
\usepackage{pgf}
\usepackage{tikz}
\usepackage{subfigure}
\usepackage{matlab-prettifier}
\usepackage{tabularx}
\usepackage{ragged2e}
\newcolumntype{Y}{>{\RaggedRight\arraybackslash}X}

\usetikzlibrary{arrows, decorations.pathmorphing, backgrounds, positioning, fit, petri, automata}
\usetikzlibrary{shadows,arrows,positioning}
\RequirePackage{algorithm,algpseudocode}
\allowdisplaybreaks
\usepackage{changes}
\usepackage{caption}
\usepackage{bbm}
\usepackage{setspace}
\newtheorem{thm}{Theorem}
\newtheorem{defin}{Definition}
\newtheorem{lem}{Lemma}

\newtheorem{rem}{Remark}

\allowdisplaybreaks[4]
\AtBeginDocument{%
	\providecommand\BibTeX{{%
			\normalfont B\kern-0.5em{\scshape i\kern-0.25em b}\kern-0.8em\TeX}}}
\journal{~}
\begin{document}
\captionsetup[figure]{labelfont={bf},labelformat={default},labelsep=period,name={Fig.}}
\begin{frontmatter}
\title{Finding Core Balanced Modules in Statistically Validated Stock Networks}
\author[label1]{Huan Qing}
\ead{qinghuan@cqut.edu.cn\&qinghuan@u.nus.edu}
\address[label1]{School of Economics and Finance, Chongqing University of Technology, Chongqing, 400054, China}
\author[label2]{Xiaofei Xu\corref{cor1}}
\ead{xiaofeix@whu.edu.cn}
\cortext[cor1]{Corresponding author.}
\address[label2]{School of Mathematics and Statistics, Wuhan University, Wuhan, 430072, China}
\begin{abstract}
Traditional threshold-based stock networks suffer from subjective parameter selection and inherent limitations: they constrain relationships to binary representations, failing to capture both correlation strength and negative dependencies. To address this, we introduce statistically validated correlation networks that retain only statistically significant correlations via a rigorous t-test of Pearson coefficients. We then propose a novel structure termed the largest strong-correlation balanced module (LSCBM), defined as the maximum-size group of stocks with structural balance (i.e., positive edge-sign products for all triplets) and strong pairwise correlations. This balance condition ensures stable relationships, thus facilitating potential hedging opportunities through negative edges. Theoretically, within a random signed graph model, we establish LSCBM's asymptotic existence, size scaling, and multiplicity under various parameter regimes. To detect LSCBM efficiently, we develop MaxBalanceCore, a heuristic algorithm that leverages network sparsity. Simulations validate its efficiency, demonstrating scalability to networks of up to 10,000 nodes within tens of seconds. Empirical analysis demonstrates that LSCBM identifies core market subsystems that dynamically reorganize in response to economic shifts and crises. In the Chinese stock market (2013–2024), LSCBM’s size surges during high-stress periods (e.g., the 2015 crash) and contracts during stable or fragmented regimes, while its composition rotates annually across dominant sectors (e.g., Industrials and Financials). 
\end{abstract}
\end{frontmatter}
\section{Introduction}\label{sec1}
The stock market is a dynamic and complex system shaped by the continuous interactions of countless individual stocks \citep{fama1965behavior,de1985does,Frank2009,sammon2026index}. These interactions are driven by macroeconomic forces (such as interest rates and inflation), sector-specific innovations, collective investor sentiment, political events, and so on \citep{gordon1959dividends,chen1986economic,barsky1993does,antonakakis2013dynamic,engle2013stock,arouri2016economic,paramati2017effects,boungou2022impact,habib2018stock,shah2019stock,jiang2021applications,venturini2022climate,ye2025effect}. For decades, researchers have been interested in understanding these complexities, moving beyond simplistic models that treat stocks as isolated entities or rely solely on pairwise comparisons. Network technique has emerged as a powerful tool to model financial systems as interconnected networks \citep{mantegna1999hierarchical,albert2002statistical,dorogovtsev2002evolution,newman2003structure,BOGINSKI20063171,huang2009network,chi2010network,kwapien2012physical,Acemoglu2015,samitas2022covid,LIU2024122529,MASUDA20251}. In this framework, stocks are represented as nodes, and their relationships—typically measured by price correlations—form the edges of a financial network. This approach has proven invaluable for visualizing market structure, identifying systemic risks, and uncovering hidden dependencies that traditional statistical methods often miss.  

A significant portion of the literature on stock correlation networks relies on the threshold-based method \citep{chi2010network}, which constructs stock networks by linking stocks only when their price correlation exceeds a predefined threshold. This binarization simplifies the network structure for graph-theoretical analysis. Researchers have applied this approach to several interconnected research streams including analyzing market stability under varying conditions \citep{HEIBERGER2014376,NOBI2014135,MAJAPA201635,ESMAEILPOURMOGHADAM2019121800,ZHANG2019748,LI2020101185,VIDALTOMAS2021101981,chen2025systemic}, predicting economic growth using Bayesian classifiers \citep{Heiberger2018}, examining structural transitions and market dynamics \citep{Wang2018,Memon2019,Liang2024}, assessing common factor impacts \citep{EOM20171}, modeling risk diffusion \citep{YANG2022103180, YANG2024101138}, and identifying influential stocks \citep{CHEN2022100836,QU2022111939}. These investigations fundamentally rely on analyzing topological properties such as clustering coefficients (sectoral cohesion), modularity structures (co-moving groups), centrality measures (systemically important assets), and network stability, which provide the analytical framework for interpreting market behavior across these research domains \citep{boccaletti2006complex,LU20161,peng2018influence,tabassum2018social}. Particularly, an interesting topic in stock network analysis lies in detecting communities, where groups of stocks connect together more closely than with the broader market. These communities are interpreted as reflecting shared fundamentals like industry affiliations (e.g., technology stocks) \citep{Li2022,Yan2023,ZHOU2023118944,xing2023community,QING2025112769}. The appeal of this method lies in its simplicity and computational efficiency, enabling researchers to transform correlation matrices into interpretable network graphs.  

Despite its popularity, the threshold-based approach is not without limitations. A critical issue lies in the arbitrary selection of the threshold value, which directly impacts the resulting network structure. Researchers often rely on heuristic criteria or trial-and-error methods to choose this threshold rather than rigorous statistical justification, leading to inconsistent results across studies. For instance, a small change in the threshold can drastically alter the number of connected nodes, the strength of observed relationships, and the identification of communities. This sensitivity undermines the reproducibility of findings and raises questions about the robustness of conclusions drawn from such networks. Moreover, compounding this problem is the binary nature of threshold networks. Relationships are reduced to a simple dichotomy—connected or disconnected—discarding critical information about the strength of correlations \citep{NewmanWeighted}. For example, a correlation of 0.85 and one of 0.55 might both be deemed ``connected" at a threshold of 0.5, despite representing vastly different degrees of co-movement. Such distinctions are crucial for portfolio risk management and diversification strategies. More importantly, the binary framework entirely ignores negative correlations, which are foundational to diversification and hedging. Assets that move inversely during downturns can naturally offset losses in a portfolio, yet traditional threshold networks overlook these relationships by focusing solely on the magnitude of correlations. Market interactions are inherently continuous and directional phenomena; forcing them into a binary, unsigned framework discards economically vital information, leading to an incomplete and potentially misleading picture of market dynamics. By truncating these relationships into a binary framework, existing methods risk oversimplifying the true complexity of the market. These foundational problems critically impact the interpretation of community structures identified within such threshold networks. While community detection algorithms are powerful tools for finding densely connected subgroups in stock networks, their application here faces inherent methodological challenges. First, reliably estimating the optimal number of communities is challenging. Second, different community detection algorithms \citep{SC,RSC,SCORE,CMM,qing2023community,qing2025community,garcia2025improving} applied to the same stock threshold network can yield substantially different groupings. Third, the detected communities are highly sensitive to the chosen correlation threshold – changing the threshold fundamentally reshapes the communities. Consequently, the identified groups may lack clear and consistent financial meaning. This disconnect between algorithmic groupings and tangible economic logic raises serious questions about the practical utility of these communities for applications like portfolio construction or risk management. The inability to map network structures to tangible economic phenomena suggests that current approaches may be missing key elements of market organization.   

The shortcomings of the threshold-based method motivate us to develop an alternative approach that addresses these issues while preserving the interpretability of stock network structures. Some studies construct stock networks based on distance matrices (e.g., \citep{mantegna1999hierarchical,zhang2025assessing,chen2025contagion,zhu2025brown,zhao2025can}), where Pearson correlations are first transformed into distances and then filtered using techniques like minimum spanning trees. While this approach avoids an explicit threshold on the correlation coefficient, it still produces dense, non‑negative matrices that discard the sign of the original relationships, and the subsequent filtering step introduces its own arbitrariness. In contrast, a more promising avenue lies in leveraging statistical validation to filter correlations, ensuring that only those with robust evidence of significance are included in the network. Unlike arbitrary thresholds, statistical validation provides an objective criterion for determining the relevance of a relationship, grounded in hypothesis testing. This method not only mitigates the subjectivity of threshold selection but also retains the full spectrum of correlation strengths, allowing for a more nuanced representation of market interactions. Furthermore, by explicitly accounting for the sign of correlations, such an approach can capture both positive and negative dependencies, enriching the network’s ability to reflect real-world financial dynamics. However, even with statistically validated correlations, constructing a network is just the first step. The true challenge lies in extracting meaningful substructures that align with economic intuition and offer actionable insights. This is where the concept of structural balance theory \citep{heider1946attitudes} becomes particularly relevant. Originating in social psychology to explain the stability of relationships within triads (e.g., ``the friend of my friend is my friend," or ``the enemy of my enemy is my friend"), this theory provides a principled framework for understanding how configurations of positive (friendly) and negative (antagonistic) ties tend towards stable equilibria \citep{heider1946attitudes,cartwright1956structural,facchetti2011computing,zheng2015social,ma2015memetic,wang2016optimizing,cai2022structure,song2022evolutionary,dong2025social}. For instance, \cite{zheng2015social} examined social media signed networks to identify balanced subnetworks and measure social tension. \cite{facchetti2011computing} computed the global level of balance of very large online social networks and claimed that currently available networks are indeed extremely balanced. \cite{dong2025social} proposed a social balance theory-based modeling framework for group-to-empirical decision-making transition with cognitive inertia and trust propagation. See Figure \ref{TOYSBT} for an illustration of the four states in structural balance theory. These principles have been applied to various domains, from political alliances to social media interactions, but their relevance to financial markets remains largely unexplored. Translated into the financial domain, positive correlations represent assets moving in tandem (like allies), while negative correlations represent assets moving oppositely (like adversaries). Given that stock prices can exhibit both positive and negative correlations, and that these relationships can influence hedging strategies and portfolio diversification, structural balance theory provides a natural framework for analyzing the stability and dynamics of financial networks.   
\begin{figure}
\centering
\resizebox{\columnwidth}{!}{
\subfigure[Two balanced states]{\includegraphics[width=0.45\textwidth]{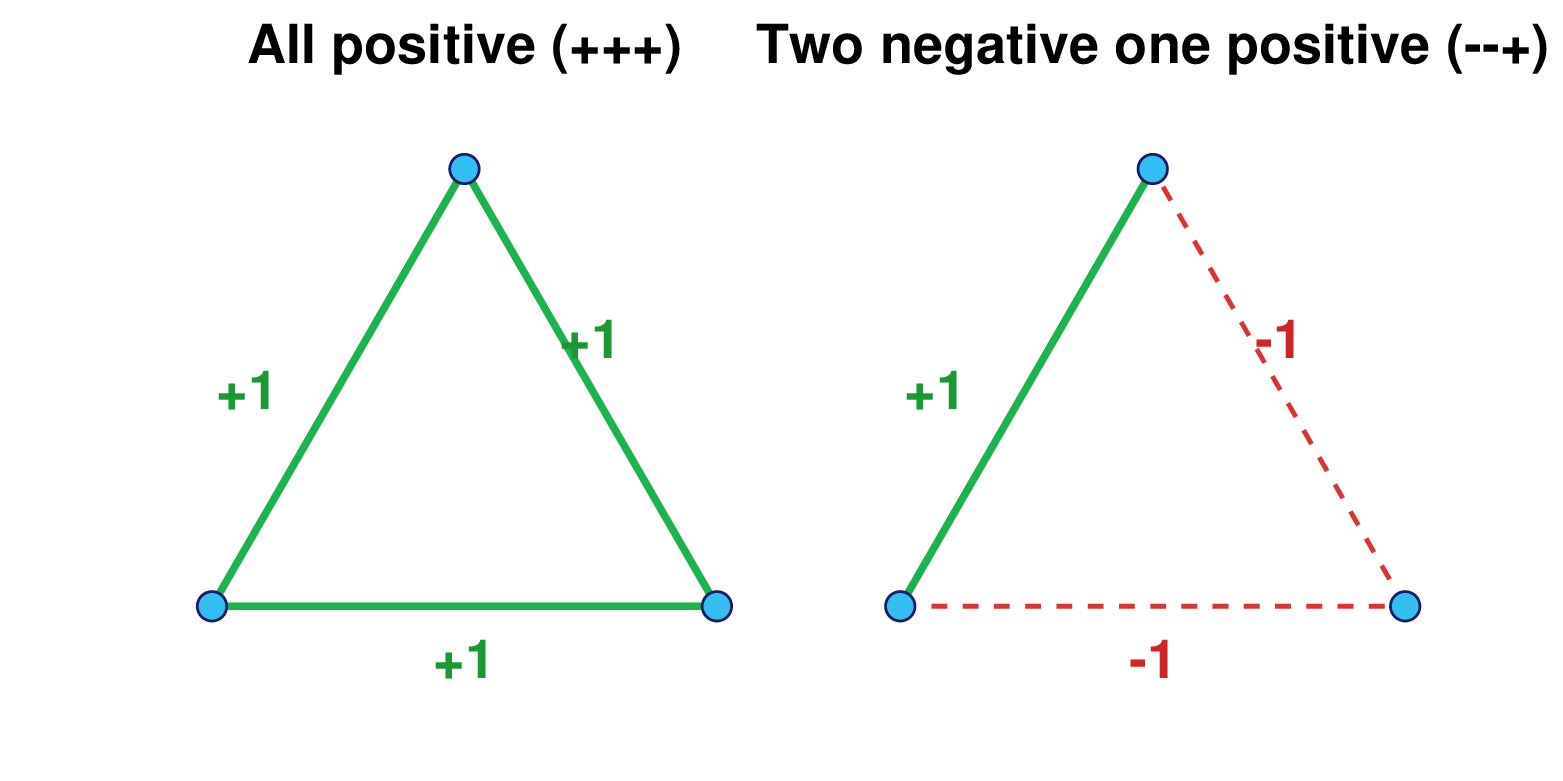}}
\subfigure[Two unbalanced states]{\includegraphics[width=0.45\textwidth]{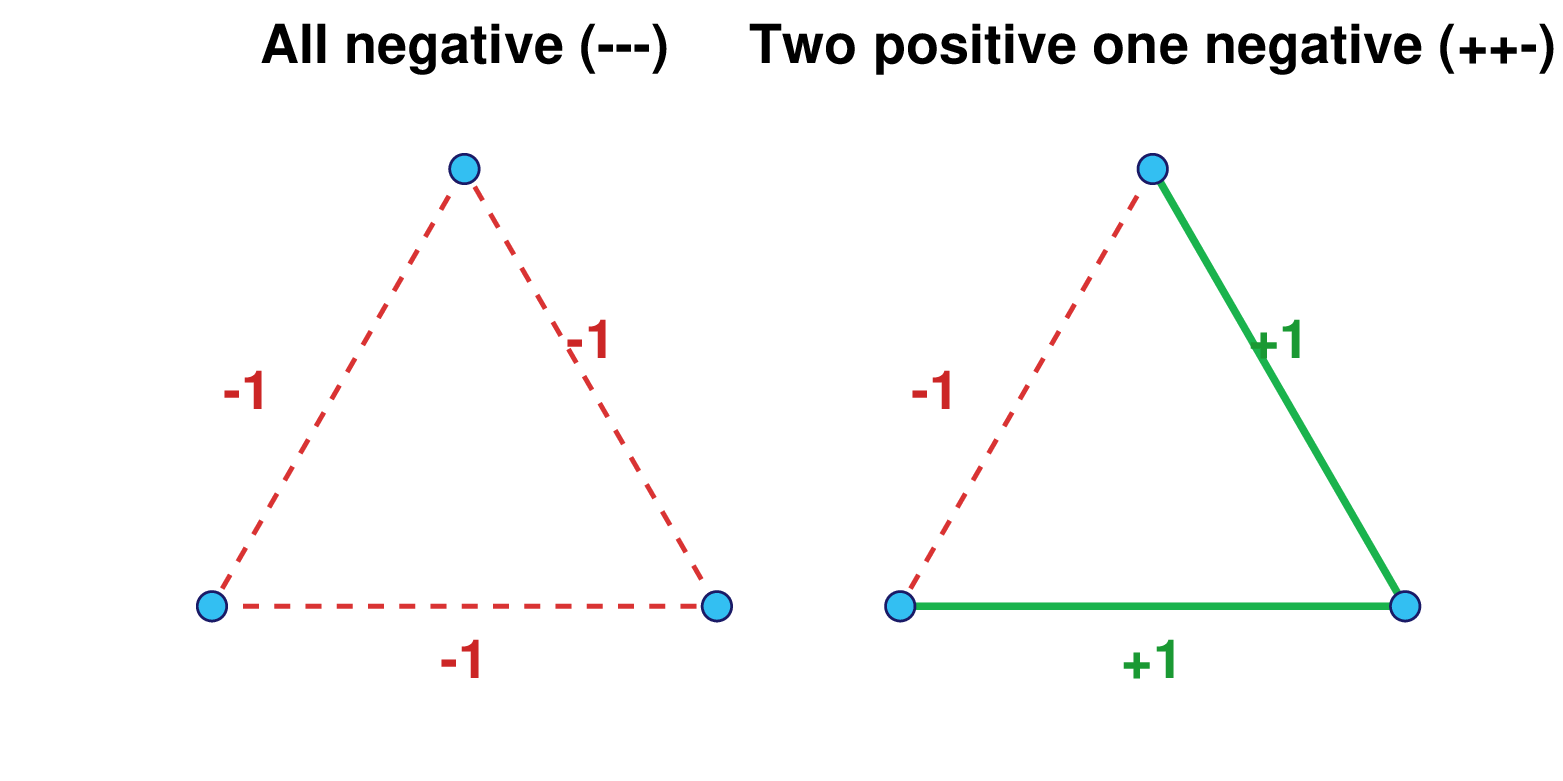}}
}
\caption{Illustration of structural balance theory configurations.}
\label{TOYSBT} 
\end{figure}

To systematically overcome the limitations of the traditional threshold-based stock networks, this article translates the above insights into four concrete contributions:
\begin{itemize}
  \item \textbf{A statistically grounded network construction.} Instead of discarding correlation strength and sign through an arbitrary cutoff, we introduce a statistically validated correlation network that eliminates spurious relationships while preserving the richness of correlation data. In detail, we replace subjective thresholding with rigorous statistical hypothesis testing. By applying t-tests to Pearson correlation coefficients, we construct a network that retains only statistically significant relationships. This method eliminates arbitrary parameters, preserving both the strength and sign (positive/negative) of correlations. By filtering out spurious links, we ensure that the network reflects genuine dependencies, providing a more robust foundation for analysis than traditional threshold networks.  
  \item \textbf{The largest strong‑correlation balanced module (LSCBM).} Rather than applying generic community detection—which often yields unstable clusters with unclear financial interpretation, we define a novel concept: the largest strong-correlation balanced module (LSCBM). LSCBM is a maximal subset of stocks that satisfies two key properties: (1) all pairwise correlations are strong and statistically significant, and (2) the network adheres to the principles of structural balance, ensuring that every triangle of stocks has a positive product of edge signs. This dual requirement ensures that the module is not only densely connected but also structurally stable, making it a candidate for identifying core market subsystems with inherent resilience and hedging potential.  
  \item \textbf{Rigorous asymptotic theory.} Prior empirical studies of stock networks rarely provide formal guarantees for the substructures they analyse. To fill this gap, we rigorously analyze LSCBM within a generative statistical model of random signed networks, where each potential edge is independently positive, negative, or absent with fixed probabilities. We prove that LSCBMs exist with high probability as the number of stocks grows, and we derive the asymptotic scaling of the expected LSCBM size under different connectivity regimes. These results provide a predictive, mathematically rigorous understanding of how market connectivity governs the size of the core stable subsystem—a level of theoretical depth not found in purely descriptive network analyses. This theoretical foundation also confirms that such stable core structures are not anomalies but fundamental features emerging in large markets.
  \item \textbf{An efficient algorithm for large markets.}  Given that the identification of LSCBM is a non-trivial task in large markets, we develop MaxBalanceCore, a heuristic algorithm that leverages network sparsity and correlation strength thresholds. Empirically, we validate MaxBalanceCore’s accuracy and scalability, demonstrating its ability to process networks with over 10,000 nodes in seconds. Simulation studies confirm its accuracy, and an empirical application to the Chinese stock market illustrates its ability to identify economically interpretable and stable market subsystems.  
\end{itemize}

The remainder of this article is organized as
follows. Section \ref{Sec2} details the construction of statistically validated correlation networks. Section \ref{Sec3} defines the largest strong-correlation balanced module, establishes its asymptotic properties within a random signed graph model, and develops the MaxBalanceCore algorithm for efficient detection. Section \ref{Sec4} evaluates the algorithm’s performance via simulations and validates the framework’s utility using real-world stock market data. Section \ref{sec5} concludes and suggests future research directions. All technical proofs and the MATLAB codes for the proposed algorithm MaxBalanceCore are provided in the appendix.

\textbf{\textit{Notations.}} We take the following general notations in this article. Write $[m]:=\{1,2,\ldots,m\}$ for any positive integer $m$. Let $N$ denote the number of stocks (nodes), $T$ the length of the logarithmic return vector (number of time points), $\Delta\tau$ the time interval (set to one day), $P_i(\tau)$ the price of stock $i$ at time $\tau$, $r_i(\tau)$ the logarithmic return calculated as $\log\frac{P_i(\tau)}{P_i(\tau-\Delta\tau)}$, $\mathbf{C}$ the $N \times N$ Pearson correlation matrix with elements $\mathbf{C}_{i,j}$, $\bar{r}_i$ the mean of $r_i$, $\rho$ the threshold for classical network construction, $\widetilde{\mathbf{C}}$ the statistically validated correlation matrix, $\alpha$ the significance level in hypothesis testing or the probability of a positive edge in random signed graphs $\mathcal{G}(N,\alpha,\beta)$, $\beta$ the probability of a negative edge in $\mathcal{G}(N,\alpha,\beta)$, $t_{i,j}$ the test statistic for correlation significance, $\nu = T-2$ the degrees of freedom for the t-distribution, $\sigma$ the minimum correlation strength threshold for strong-correlation modules, $\mathcal{S}$ a subnetwork or module, $\mathcal{S}^*$ the largest strong-correlation balanced module, $|\mathcal{S}|$ the cardinality (size) of $\mathcal{S}$, $A$ and $B$ disjoint sets in structural balance partitions, $Z_s$ the number of strong-correlation balanced modules of size $s$, $\lambda(\alpha,\beta)$ a scaling parameter depending on $\alpha$ and $\beta$, $f(a)$ a function in asymptotic analysis, and $H(a)$ the binary entropy function. Asymptotic notations include $\sim$ for asymptotic equivalence, $\Theta$ for a tight bound, $O$ for the big-O upper bound, and $o$ for a lower-order term. Probability, expectation, and variance are denoted by $\mathbb{P}$, $\mathbb{E}$, and $\text{Var}$, respectively. Table~\ref{tab:notation} summarizes the main symbols used in this paper for quick reference.
\begin{table}[htbp]
\centering
\small 
\setlength{\tabcolsep}{2pt} 
\caption{Summary of main notations.}
\label{tab:notation}
\begin{tabular}{>{\raggedright\arraybackslash}p{2.3cm} >{\raggedright\arraybackslash}p{5cm} >{\raggedright\arraybackslash}p{2.3cm} >{\raggedright\arraybackslash}p{5cm}}
\toprule
Symbol & Description & Symbol & Description \\
\midrule
$[m]$ & Set $\{1,2,\ldots,m\}$ for any positive integer $m$ & $N$ & Number of stocks (nodes) \\
$T$ & Length of the logarithmic return vector & $\Delta\tau$ & Time interval for return calculation\\
$P_i(\tau)$ & Price of stock $i$ at time $\tau$ & $r_i(\tau)$ & Logarithmic return of stock $i$\\
$\bar{r}_i$ & Mean of $r_i$ & $\mathbf{C}$ & $N \times N$ Pearson correlation matrix\\
$\widetilde{\mathbf{C}}$ & Statistically validated correlation matrix& $\rho$ & Threshold for classical threshold-based network construction\\
$\alpha$ & Significance level in hypothesis testing; also probability of a positive edge in $\mathcal{G}(N,\alpha,\beta)$ & $\beta$ & Probability of a negative edge in $\mathcal{G}(N,\alpha,\beta)$ \\
$t_{i,j}$ & Test statistic for correlation significance& $\nu = T-2$ & Degrees of freedom for the $t$-distribution \\
$t_{\nu}(\frac{\alpha}{2})$ & Critical value of the $t$-distribution& $\sigma$ & Minimum correlation strength threshold\\
$\mathcal{S}$ & A subnetwork or module (set of nodes) & $\mathcal{S}^*$ & Largest strong-correlation balanced module (LSCBM) \\
$|\mathcal{S}|$ & Cardinality (size) of module $\mathcal{S}$ & $A$, $B$ & Disjoint sets in a structurally balanced partition\\
$Z_s$ & Number of strong-correlation balanced modules (SCBMs) of size $s$ & $\lambda(\alpha,\beta)$ & Scaling parameter\\
$f(a)$ & Function in asymptotic analysis& $H(a)$ & Binary entropy function\\
$\sim$ & Asymptotic equivalence & $\Theta$ & Tight asymptotic bound (Theta notation) \\
$O$ & Big-O asymptotic upper bound & $o$ & Little-o lower-order term \\
$\mathbb{P}$ & Probability & $\mathbb{E}$ & Expectation \\
$\text{Var}$ & Variance & $\xi_+$ & Proportion of positive elements in $\widetilde{\mathbf{C}}$ \\
$\xi_-$ & Proportion of negative elements in $\widetilde{\mathbf{C}}$& $\mu_+$ & Average value of positive elements in $\widetilde{\mathbf{C}}$ \\
$\mu_-$ & Average value of negative elements in $\widetilde{\mathbf{C}}$ & $\varsigma$ & Proportion of nodes belonging to LSCBM: $\varsigma = |\mathcal{S}^*|/N$ \\
$\mathbf{S}$ & Signed adjacency matrix& $\operatorname{impact}_i$ & Impact (degree) of node $i$ in $\mathbf{S}$\\
$\mathbb{I}[\cdot]$ & Indicator function& $\mathcal{G}(N,\alpha,\beta)$ & Random signed graph model: each edge independently is $+1$ with prob. $\alpha$, $-1$ with prob. $\beta$, $0$ with prob. $1-\alpha-\beta$ \\
\bottomrule
\end{tabular}
\end{table}

\section{Statistically validated correlation network construction}\label{Sec2}
Consider a stock market with $N$ stocks and let $P_{i}(\tau)$ be the stock price of stock $i$ at time $\tau$ for $i\in[N]$. We know that the logarithmic return of the stock $i$ at a time interval $\triangle \tau$ can be calculated as
\begin{align}\label{logreturn}
r_{i}(\tau)=\mathrm{log}\frac{ P_{i}(\tau)}{P_{i}(\tau-\triangle \tau)},
\end{align}
where we set $\triangle \tau$ as one day in this article. Suppose there are $(T+1)$ consecutive trading days. For each stock $i$, its logarithmic return vector is a $1\times T$ vector $r_{i}=[r_{i}(1), r_{i}(2), \ldots, r_{i}(T)]$. To analyze the relationships among stocks, the Pearson correlation coefficient between any two stocks $i$ and $j$ is considered:
\begin{align}\label{Cij}
\mathbf{C}_{i,j}=\frac{\sum_{t}(r_{i}(\tau)-\bar{r}_{i})(r_{j}(\tau)-\bar{r}_{j})}{\sqrt{\sum_{\tau}(r_{i}(\tau)-\bar{r}_{i})^{2}}\sqrt{\sum_{\tau}(r_{j}(\tau)-\bar{r}_{j})^{2}}},
\end{align}
where $\bar{r}_{i}$ (and $\bar{r}_{j}$) represent the mean of $r_{i}$ (and $r_{j}$), and the summations are taken over the period we considered. For the $N$ stocks, we see that the $N\times N$ symmetric correlation matrix $\mathbf{C}$ records all Pearson correlation coefficients among stocks.

It is well-known that the classical threshold networks \citep{chi2010network} for stocks can be constructed in the following way. Let the $N\times N$ symmetric matrix $\mathbf{G}$ be the adjacency matrix of the stock threshold network. For $i\in[N], j\in[N]$, $\mathbf{G}$'s $(i,j)$-th entry is calculated by
\begin{align}\label{DefineGt}
\mathbf{G}_{i,j}=\begin{cases}
1& \mbox{when~} |\mathbf{C}_{i,j}|>\rho,\\
0, & \mbox{otherwise}.
\end{cases}.
\end{align}
where $\rho$ is a threshold value located in $(0,1)$. Though such construction of the threshold unweighted stock networks via Equation (\ref{DefineGt}) is simple, we observe that it has the following limitations:
\begin{itemize}
  \item (a) The selection of the threshold $\rho$ is highly subjective. Different values of $\mathbf{\rho}$ may result in substantially different adjacency matrices $\mathbf{G}$. Consequently, $\mathbf{G}$ fails to accurately capture the connectivity between stocks.
  \item (b) In traditional threshold-based stock networks, the absence or presence of a connection (uncorrelated or correlated) solely by the binary values 0 or 1, obtained by truncating values using a threshold $\rho$, is overly arbitrary. Such binarization completely fails to capture the varying strength of correlations between stocks, and cannot represent negative correlations between them.
\end{itemize}

We find that directly utilizing the correlation matrix $\mathbf{C}$ can overcome the two aforementioned limitations inherent in traditional stock threshold networks. However, we note that a network constructed directly from $\mathbf{C}$ would barely qualify as a network since connections exist between virtually every pair of nodes (stocks), given that $\mathbf{C}_{i,j}$ is rarely exactly zero. More importantly, the correlation coefficient $|\mathbf{C}_{i,j}|$ for some stock pairs can be extremely close to zero (e.g., 0.01). Such weak correlations can often be attributed to random noise or sampling errors, indicating an absence of any true underlying relationship. This naturally leads us to introduce the following statistically validated correlation matrix, upon which we construct the stock correlation network. To systematically distinguish economically meaningful correlations from spurious noise-induced correlations, we implement a statistical hypothesis testing procedure for each pairwise correlation coefficient $\mathbf{C}_{i,j}$. For every pair of stocks $i$ and $j$, we formalize the test as:
\begin{align*}
H_{0}&: \mathbf{C}_{i,j} = 0 \quad \text{(no true linear correlation exists)} \\
H_{1}&: \mathbf{C}_{i,j} \neq 0 \quad \text{(significant linear correlation exists)}
\end{align*}
tested at $\alpha=5\%$ significance level. The test statistic is computed as:
\begin{align}\label{tstatistics}
t_{i,j} = \mathbf{C}_{i,j} \sqrt{\frac{T - 2}{1 - \mathbf{C}_{i,j}^2}},
\end{align}
where $T$ denotes the length of the logarithmic return vector. It is well-known that the test statistic $t_{i,j}$ follows a Student's t-distribution with $\nu=T-2$ degrees of freedom under $H_{0}$ \citep{edgell1984effect,obilor2018test}. When $|t_{i,j}|> t_{\nu}(\frac{\alpha}{2})$ (where $t_{\nu}(\frac{\alpha}{2})$ is the critical value of the t-distribution), we reject $H_{0}$ and conclude that $\mathbf{C}_{i,j}$ is statistically significant. For $i\in[N], j\in[N]$, the statistically validated correlation matrix $\widetilde{\mathbf{C}}$ is then constructed element-wise as:
\begin{align}\label{sscm}
\widetilde{\mathbf{C}}_{i,j} =
\begin{cases}
\mathbf{C}_{i,j} & \text{if } H_0 \text{ is rejected} \\
0 & \text{otherwise}.
\end{cases},
\end{align}
where we set $\widetilde{\mathbf{C}}_{ii}=1$ for convenience. The flowchart of $\widetilde{\mathbf{C}}$'s construction process of the stock market is shown in Figure \ref{Flowchart0}.

\begin{figure}[htbp]
\centering
\includegraphics[height=0.43\textheight,keepaspectratio]{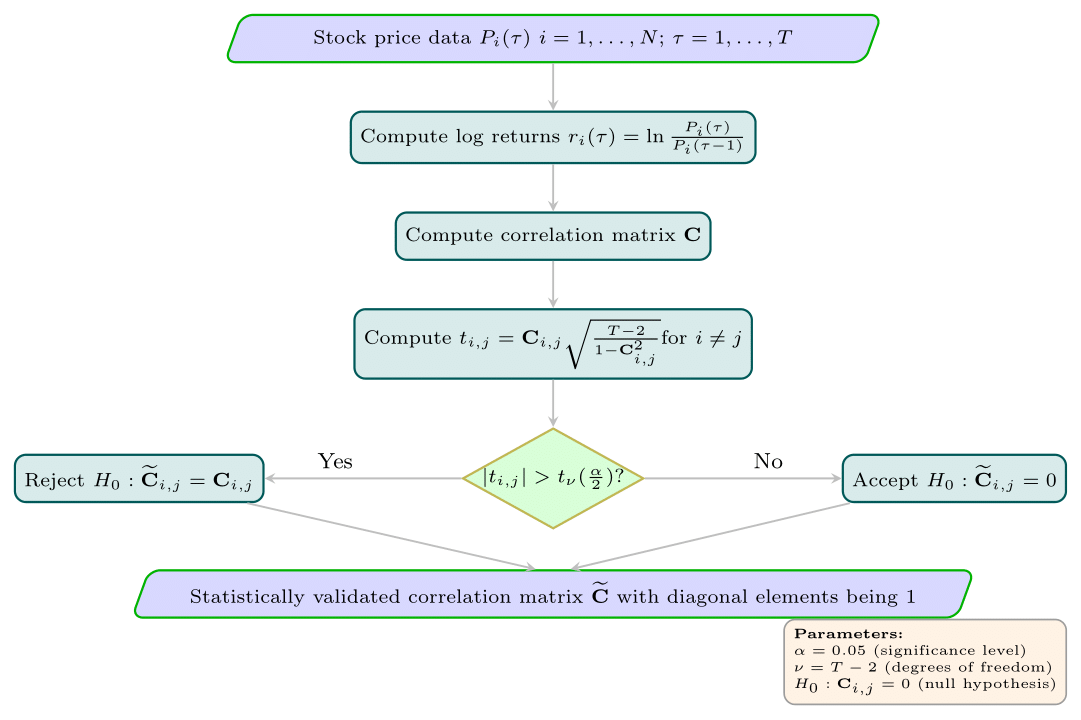}
\caption{Flowchart of the construction process for statistically validated stock correlation networks. }
\label{Flowchart0}
\end{figure}

This validation process effectively filters out correlations attributable to random fluctuations while preserving economically significant relationships. The resulting sparse matrix $\widetilde{\mathbf{C}}$ forms the adjacency matrix of the correlation network for stocks, where non-zero edges represent statistically validated correlations between stocks. It is noteworthy that the symmetric matrix $\widetilde{\mathbf{C}}$ is a weighted adjacency matrix, with all elements located in the interval $[-1, 1]$. Consequently, the resulting correlation network is an undirected, weighted network where:
\begin{itemize}
  \item Edge weights quantify the strength of validated correlations between stocks.
  \item Edge signs naturally represent positive or negative   relationships between stocks.
  \item Sparsity ensures only statistically meaningful connections are retained (i.e., $\widetilde{\mathbf{C}}_{i,j} = 0$ for insignificant correlations).
\end{itemize}

The proposed correlation network offers several advantages over traditional threshold networks in stock market analysis:
\begin{itemize}
  \item The correlation network overcomes the limitations of threshold networks by eliminating the need for subjective threshold selection. This allows for a more objective and data-driven approach to network construction.
  \item Unlike traditional threshold networks, which depict relationships as binary (connected or not), the correlation network quantifies association strengths through edge weights. Meanwhile, the inclusion of signed edges permits the explicit representation of both positive and negative correlations, facilitating a more nuanced investigation into relationships between stocks.
  \item The sparsity of the correlation network ensures that only statistically meaningful connections are retained. This helps to filter out noise and irrelevant information, providing a clearer picture of the true relationships between stocks.
\end{itemize}
\section{Largest strong-correlation balanced module (LSCBM)}\label{Sec3}
In this section, we formalize the definition of largest strong-correlation balanced module by uniting statistically validated correlation strengths with structural balance theory, derive its asymptotic existence, size scaling, and multiplicity under a random signed graph model, and present an efficient algorithm to detect it from large-scale networks.
\subsection{Definition of LSCBM}
To advance our analysis of statistically validated stock correlation networks, we introduce a novel concept: the largest strong-correlation balanced module (LSCBM for short). LSCBM combines structural balance theory with statistically validated correlation networks to identify maximal market subsystems where stocks exhibit economically significant relationships and relational stability. Its definition is provided below.
\begin{defin}\label{LSCBM}
(Largest strong-correlation balanced module, LSCBM)
Let $\widetilde{\mathbf{C}}$ denote the statistically validated correlation matrix defined in Equation (\ref{sscm}) of $N$ stocks. A subnetwork $\mathcal{S}$ is a strong-correlation balanced module (SCBM for short) if:
\begin{itemize}
  \item (1) Strong correlation module: For any pair of nodes $i$ and $j$ in the subnetwork $\mathcal{S}$, they share a strong statistically validated edge:
       \begin{align}\label{SCC}
       |\widetilde{\mathbf{C}}_{i,j}|\geq\sigma,
       \end{align}
       where $\sigma>0$ is a predefined threshold.
  \item (2) Structural balance: For every three distinct nodes $i, j$, and $k$ in $\mathcal{S}$, the product of edge signs for the triangle formed by these three nodes is positive, i.e.,
      \begin{align}\label{balancedtriangle}
        \widetilde{\mathbf{C}}_{i,j}\times \widetilde{\mathbf{C}}_{i,k}\times \widetilde{\mathbf{C}}_{j,k}>0.
      \end{align}
      This permits two configurations: (i) all three correlations among $i, j,$ and $k$ are positive, or (ii) two negative and one positive.
\end{itemize}

The largest strong-correlation balanced module (LSCBM) $\mathcal{S}^{*}$ is the maximal such subgraph regarding node cardinality $|\mathcal{S}|$, i.e.,
\begin{align}\label{LSCBMset}
\mathcal{S}^{*}=\mathrm{argmax}_{\mathcal{S}\subseteq \{1,2,\ldots,N\}}\{|\mathcal{S}|: \mathcal{S}\mathrm{~is~a~SCBM}\}.
\end{align}
\end{defin}

\begin{figure}
\centering
\resizebox{\columnwidth}{!}{
{\includegraphics[width=3\textwidth]{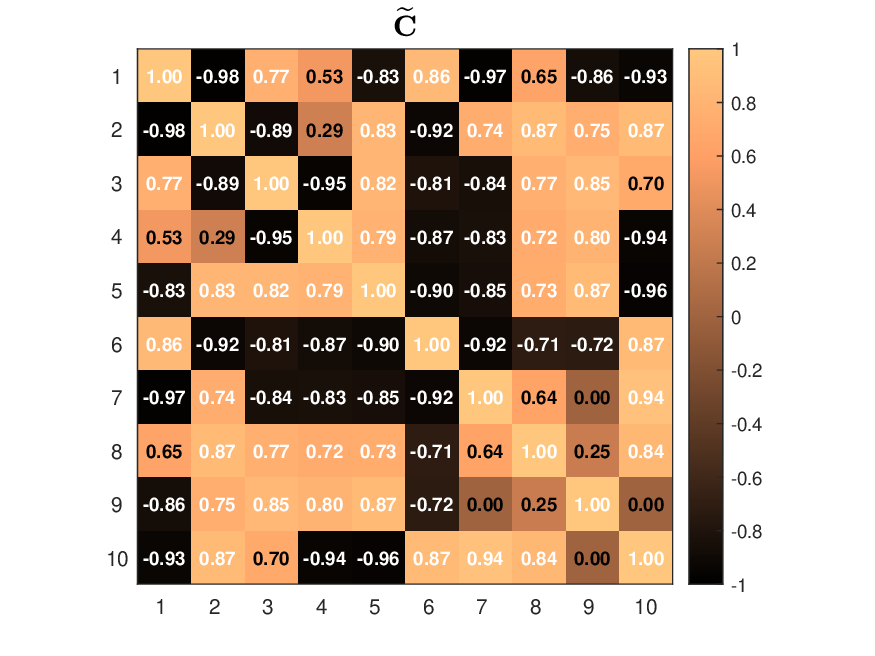}}
{\includegraphics[width=3\textwidth]{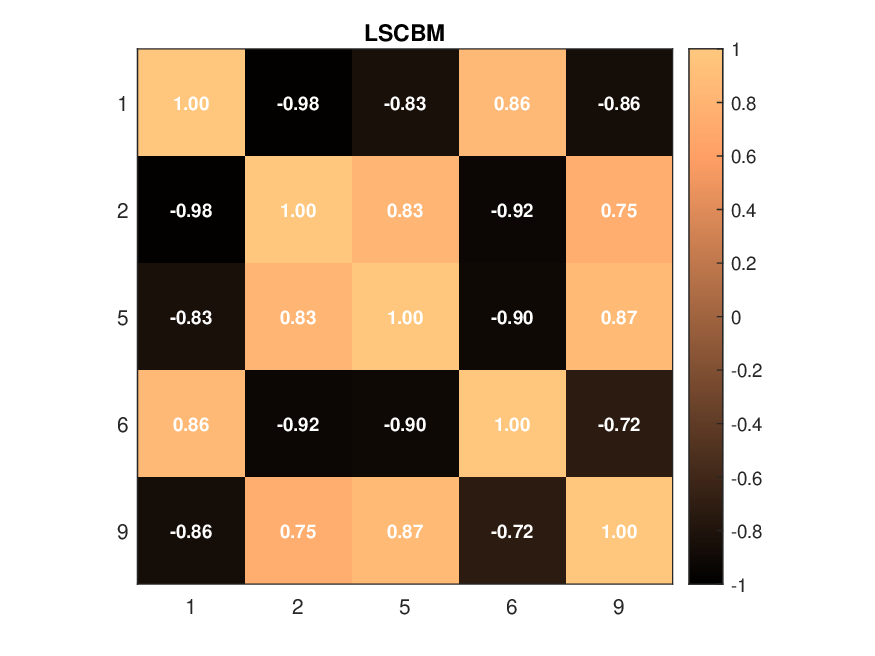}}
}
\caption{An illustrative example of the statistically validated correlation matrix $\widetilde{\mathbf{C}}$ and its LSCBM.}
\label{TOYLSCBM} 
\end{figure}

In Equation (\ref{SCC}), the threshold $\sigma$ defines the minimum correlation strength required for stocks within the LSCBM. A higher $\sigma$ value yields a smaller module size. Unless otherwise specified, we set $\sigma=0.7$ throughout this article. It could be noted that the parameter $ \sigma $ introduced here is conceptually distinct from the thresholds employed in classical threshold network. In our formulation, $ \sigma $ serves as a quantitative criterion reflecting the strength of significant correlations to be selected for specific network construction. In financial markets, the correlations are usually moderate rather than strong, adopting a strength threshold such as 0.7 or 0.8 is generally appropriate. Nevertheless, the specific choice of $ \sigma $ could be carefully calibrated in light of the characteristics of the underlying dataset and the requirements of the target application domain. Figure \ref{TOYLSCBM} illustrates an example of the statistically validated correlation matrix $\widetilde{\mathbf{C}}$ and the LSCBM extracted from it. While rooted in network science, the LSCBM moves beyond pure graph theory to deliver practical insights for analyzing stock market structure and portfolio design. Its importance rests on the following aspects:
\begin{itemize}
  \item By identifying clusters of stocks with strong correlations where the absolute value of the correlation coefficient is no smaller than $\sigma$, LSCBM provides a clear lens to view market segments that move in tandem. This is crucial for understanding sector dynamics and the transmission of market shocks. Such strongly correlated groups often reflect common underlying factors like industry trends, macroeconomic conditions, or shared risk exposures. For instance, tech stocks might form an LSCBM due to their collective sensitivity to interest rate changes or technological innovation cycles.

\item The structural balance aspect of LSCBM offers profound risk management insights. When all triangles within the module are balanced—either through uniform positive correlations or configurations where ``the enemy of my enemy is my friend"—this reveals stable relational patterns. In finance, this stability is valuable for predicting how shocks propagate through the market. A balanced negative triangle, where two negative correlations and one positive correlation exist among three stocks, presents a natural hedging opportunity. This configuration allows investors to construct portfolios where losses in one position are offset by gains in another, providing a built-in risk mitigation mechanism.

\item The LSCBM framework enhances portfolio construction by highlighting both opportunities for concentration and diversification. Strong positive correlation clusters may appeal to investors seeking focused sector exposure, while the inclusion of balanced negative triangles enables the creation of resilient portfolios that withstand various market conditions. By identifying these structurally balanced subnetworks, investors can make more informed decisions about asset allocation, knowing they are backed by statistically validated relationships rather than spurious correlations.
\end{itemize}

In essence, the LSCBM concept bridges network theory with practical financial applications, offering a robust framework for analyzing market structure, designing portfolios, and managing risk in a manner that respects the complex, interdependent nature of financial markets.
\begin{rem}
While introduced here within statistically validated stock correlation networks, the concept of LSCBM offers a fundamental framework broadly applicable to any undirected network, weighted or unweighted. Its core requirements – identifying a maximal group where all pairwise connections meet a meaningful strength threshold and where the overall structure adheres to balance theory (ensuring stable triad configurations) – provide a universal lens for uncovering cohesive and relationally stable subsystems. This allows us to identify critical, resilient cores characterized by strong, well-structured interactions across diverse domains, from social and biological systems to other complex networks.
\end{rem}
\subsection{Theoretical analysis of LSCBM}
The definition of SCBM and LSCBM within statistically validated correlation networks in Definition \ref{LSCBM} captures crucial aspects of financial relationships: the statistical significance of correlations, their strength (via the strength threshold $\sigma$), and their sign. This representation is rich and directly grounded in empirical data, making it highly relevant for practical financial analysis. However, when we shift our focus to theoretical analysis—specifically, to rigorously establish properties such as the asymptotic existence, expected size, and multiplicity of the core concept LSCBM in large-scale markets, we encounter significant challenges inherent to the continuous, data-dependent nature of this construction.

Proving fundamental properties about the LSCBM, especially as the number of stocks \(N\) grows large, necessitates a formal probabilistic model. We require a framework that allows us to control edge generation probabilities and analyze emergent structures precisely. Random graph models provide this foundation. Yet, directly modeling the continuous, statistically validated correlation matrix \(\widetilde{\mathbf{C}}\) within a random graph framework is intractable for deriving sharp asymptotic results. The continuous edge weights and the complex dependence structure arising from the validation process (which itself depends on the underlying return time series) make analytical characterization difficult.

To overcome this barrier and enable rigorous theoretical exploration, we introduce a carefully chosen abstraction: the random signed network. This model, denoted as \(\mathcal{G}(N, \alpha, \beta)\) given in Definition \ref{SignedModel}, simplifies the edge representation while preserving the core structural information essential for defining and analyzing LSCBM in a theoretical context. Crucially, it discards the precise correlation magnitude but retains all key pieces of information derived from the statistical validation and strength filtering process. A signed network, characterized by edge weights in $\{-1, 0, +1\}$, provides the necessary theoretical lens. Here, a non-zero edge ($|\widetilde{\mathbf{C}}_{i,j}|\geq\sigma$) is simply represented by its sign (+1 or -1), and a zero edge ($|\widetilde{\mathbf{C}}_{i,j}|<\sigma$ or statistically insignificant) remains 0. This binarization (+1, -1, 0) captures the essence of the economically meaningful relationships identified in the statistically validated network: which stocks have strong, statistically validated connections and whether those connections are positive or negative. The continuous correlation strength, while important for the initial filtering, is not directly utilized by structural balance theory, which operates solely on the signs of the relationships within triangles. The condition \(|\widetilde{\mathbf{C}}_{i,j}| \geq \sigma\) ensures the edge is economically meaningful; the sign determines its role in structural balance. This abstraction allows us to leverage powerful tools from random graph theory. We can model \(\alpha\) and \(\beta\) as edge formation probabilities, analyze the resulting combinatorial structures, and derive rigorous asymptotic results concerning the LSCBM's properties under different parameter regimes. In essence, the signed network definitions provide the theoretical foundation needed to rigorously analyze the core concepts introduced for empirical financial network analysis. The formal definition of the random signed network model is provided below.
\begin{defin}\label{SignedModel}
(Random signed graph \(\mathcal{G}(N, \alpha, \beta)\) ). Let $\mathcal{N}$ be a set of nodes with cardinality $|\mathcal{N}| =N$. The random signed graph \(\mathcal{G}(N, \alpha, \beta)\) is defined as an undirected graph where every pair of distinct nodes $(i,j)$, $i\neq j$, independently forms:
\begin{itemize}
  \item A positive edge (denoted $+1$) with probability $\alpha$,
  \item A negative edge (denoted $-1$) with probability $\beta$,
  \item No edge (denoted $0$) with probability $1-\alpha-\beta$.
\end{itemize}
Here, $\alpha \in(0,1], \beta\in[0,1)$, and $\alpha+\beta\leq 1$. The model assumes no self-loops, and all edges are mutually independent. This graph is undirected and characterized by its adjacency matrix \(\mathbf{S} \in \{-1, 0, 1\}^{N \times N}\), where \(\mathbf{S}_{i,j} = \mathbf{S}_{j,i}\) for all \(i, j\).
\end{defin}

Within \(\mathcal{G}(N, \alpha, \beta)\), the definitions of SCBM and LSCBM are direct analogs of those in the stock correlation network but adapted to the random graph’s binary edge structure:
\begin{itemize}
  \item SCBM (Strong-correlation balanced module):
   A subgraph \(\mathcal{S}\) is an SCBM if:
   \begin{itemize}
     \item Every pair of distinct nodes \(i, j \in \mathcal{S}\) has a non-zero edge (i.e., \(\mathbf{S}_{i,j} \neq 0\)).
     \item For every triplet of distinct nodes \(i, j, k \in \mathcal{S}\), the product of edge signs satisfies \(\mathbf{S}_{i,j} \times \mathbf{S}_{i,k} \times \mathbf{S}_{j,k} > 0\), implying either (i) all three edges are positive, or (ii) two edges are negative and one is positive.
   \end{itemize}
  \item LSCBM (Largest strong-correlation balanced module):
   The LSCBM \(\mathcal{S}^*\) is the SCBM with maximum cardinality \(|\mathcal{S}|\), formally:
   \[
   \mathcal{S}^* = \underset{\mathcal{S} \subseteq \mathcal{N}}{\arg\max} \left\{ |\mathcal{S}| : \mathcal{S} \text{ is an SCBM} \right\}.
   \]
\end{itemize}

Having defined the random signed graph model \(\mathcal{G}(N,\alpha,\beta)\) and adapted the concept LSCBM to this theoretical framework, a fundamental question arises: does such a core balanced module even exist in large markets? This is not merely a technical concern. In financial applications, the very premise of identifying stable core subsystems hinges on their asymptotic existence as the market grows. We must first establish whether the strict joint conditions of strong correlation (non-zero edges) and structural balance can realistically coexist in large networks. The following lemma addresses this foundational concern, ensuring our concept is theoretically sound.
\begin{lem}\label{lemmain1}
(Non-emptiness)
Consider a random signed graph \(\mathcal{G}(N, \alpha, \beta)\), when $\alpha>0, \beta\geq0$, and $\alpha+\beta\leq1$, the probability that no LSCBM exists vanishes:
\[
\mathbb{P}(\mathcal{S}^* = \emptyset) \to 0\qquad\text{as~}N\rightarrow\infty.
\]
\end{lem}
Lemma \ref{lemmain1} provides the cornerstone for our theoretical analysis: with high probability, at least one LSCBM exists in a large random signed network when \(\alpha > 0\). Knowing that LSCBM exists allows us to confidently explore its scaling behavior and multiplicity property under different network regimes, which is crucial for understanding its potential role in modeling real financial markets.

With existence guaranteed, we investigate the asymptotic scaling of LSCBM’s size within the general regime of \(\mathcal{G}(N,\alpha,\beta)\), where edge probabilities \(\alpha\) and \(\beta\) remain fixed as \(N \to \infty\). Understanding this scaling is crucial—it quantifies how the core market subsystem grows relative to the overall market size and reveals its dependence on the balance between positive and negative relationship densities. The following theorem establishes this universal scaling law and a key property about the multiplicity of LSCBM.
\begin{thm}\label{main1}
(General regime) Consider a random signed graph \(\mathcal{G}(N, \alpha, \beta)\). Define the strong-correlation balanced module (SCBM) $\mathcal{S}$ as a partition \(A \cup B\) (possibly empty parts) such that (1) all edges within \(A\) are positive, (2) all edges within \(B\) are positive, and (3) all edges between \(A\) and \(B\) are negative, i.e., every triangle in a SCBM must obey structural balance (the product of its edge signs is positive) and SCBM has at least three nodes. The largest strong-correlation balanced module (LSCBM) is defined as the SCBM $\mathcal{S}^{*}$ of maximum cardinality in Equation (\ref{LSCBMset}). Then for fixed \(\alpha, \beta > 0\) with \(\alpha + \beta \leq 1\), as \(N \to \infty\), with high probability, we have
\begin{itemize}
  \item $\mathbb{E}[|\mathcal{S}^*|] \sim \frac{\log N}{\lambda(\alpha, \beta)}, \quad \lambda(\alpha, \beta) = \begin{cases}
\frac{1}{2} |\log \alpha| & \alpha \geq \beta \\
\frac{1}{4} (|\log \alpha| + |\log \beta|) & \alpha < \beta
\end{cases}$
\item There exist multiple LSCBMs of size \(|\mathcal{S}^{*}|\). Specifically,
   \[
   \lim_{N \to \infty} \mathbb{P}\left(Z_{|\mathcal{S}^{*}|} \geq 2\right) = 1,
   \]
   where \(Z_s\) denotes the number of SCBM of size \(s\), and \(|\mathcal{S}^{*}|\) is the size of the LSCBM.
\end{itemize}
\end{thm}
Theorem \ref{main1} reveals two key insights for the general regime. First, the size of LSCBM scales logarithmically with \(N\), \(\mathbb{E}[|\mathcal{S}^*|] \sim \log N / \lambda(\alpha,\beta)\), where the scaling constant \(\lambda\) depends critically on the relative magnitudes of \(\alpha\) and \(\beta\). This explicitly links the module’s growth rate to market connectivity patterns. Second, multiple distinct LSCBMs of this maximal size typically coexist in large markets. Real-world markets may exhibit dense interconnectivity. To model this, we consider a dense regime where the probability of a positive edge \(\alpha \approx 1 - b/N\) approaches 1, while negative edges are rare (\(\beta \approx b/N\)). The next theorem characterizes how LSCBM grows under this highly positive connectivity scenario.
\begin{thm}\label{main2}
(Dense regime) Consider a random signed graph \(\mathcal{G}(N, \alpha, \beta)\) with \(\alpha = 1 - \frac{b}{N} + o(1/N)\) and \(\beta = \frac{b}{N} + o(1/N)\), where \(b > 1\) is constant. As \(N \to \infty\), we have
\begin{itemize}
  \item $\mathbb{E}[|\mathcal{S}^{*}|] = \Theta\left(\frac{N \log b}{b}\right)$ and w.h.p. LSCBM is an all-positive module.
\item There exist multiple LSCBMs of size \(|\mathcal{S}^{*}|\) with high probability.
\end{itemize}
\end{thm}
Theorem \ref{main2} shows that in the dense regime, LSCBM's size scales linearly with the market size. This linear growth suggests that highly positive connected markets tend to have proportionally large core balanced subsystems. Furthermore, multiple such large modules coexist with high probability. Conversely, markets may be dominated by adversarial relationships (negative edges). We consider a negative-dominated regime where \(\beta \to 1^{-}\). This regime tests the limits of structural balance under antagonism. The following theorem establishes LSCBM's behavior under such a conflict-dominated case.
\begin{thm}\label{main3}
(Negative-dominated regime) Consider a random signed graph \(\mathcal{G}(N, \alpha, \beta)\) with \(\beta \to 1^{-}\) and \(\alpha \to 0^{+}\) as \(N \to \infty\). We have:
\begin{itemize}
  \item $\mathbb{E}[|\mathcal{S}^{*}|] = O\left(\frac{\log N}{|\log \alpha|}\right).$
  \item If additionally \(|\log \alpha| = o(\sqrt{\log N})\), we have
   \[
   \lim_{N \to \infty} \mathbb{P}\left(Z_{|\mathcal{S}^{*}|} \geq 2\right) = 1.
   \]
\end{itemize}
\end{thm}
Theorem \ref{main3} demonstrates that widespread negativity imposes a significant constraint on the size of stable core modules, limiting the LSCBM size to scale at most logarithmically, \(\mathbb{E}[|\mathcal{S}^{*}|] = O(\log N / |\log \alpha|)\). This contrasts sharply with the linear growth seen in the dense positive regime, highlighting how widespread negative correlations fragment the market's capacity to form large, cohesive cores. However, an important nuance emerges: under the condition that positive edges, though rare, are not vanishingly fast (\(|\log \alpha| = o(\sqrt{\log N})\)), multiple LSCBMs of this maximal size still emerge with high probability. Even in conflict-rich markets, the core stable structures persist, though smaller and more numerous, reflecting market fragmentation.
\subsection{MaxBalanceCore: an efficient algorithm for identifying LSCBM}
Identifying the LSCBM in large financial networks is a computationally tough task (NP-hard). To tackle this, we develop a heuristic algorithm that leans on structural balance theory  
\citep{harary1953notion,cartwright1956structural} and exploits the natural sparsity controlled by the correlation strength threshold $\sigma$ of the statistically validated correlation network. The core idea is smart: focus the search where big modules are most likely to appear and avoid the combinatorial explosion of checking everything. Here's how it works:
\begin{itemize}
  \item First, we build a signed network by using Equation (\ref{SCC}). Only edges where the absolute correlation strength meets or exceeds the threshold $\sigma$ are left, and any other weaker links are discarded. This signed adjacency matrix $S\in\{-1,0,1\}^{N\times N}$ is our starting point.
  \item Second, we kick things off with ``high-impact" nodes, those with lots of strong connections (high degree centrality). These dense hubs are more likely to be part of large modules. We then choose seeds as nodes with the highest degree centrality in the signed adjacency matrix $S$.
  \item Third, for each seed, we divide its strongly connected neighbors into two groups: set $A$ (positive correlations to the seed) and set $B$ (negative correlations). This is where we get strict:
      \begin{itemize}
        \item Inside $A$ (or $B$), every node must have a strong positive link ($S_{u,v} = 1$) to every other node in $A$ (or $B$). If a node lacks even one positive connection within its faction, it must be removed. This enforces the ``strong module" condition internally.
        \item Every node in $A$ must have a strong negative link ($S_{u,v} = -1$) to every node in $B$. Any node showing neutrality or positivity ($S_{u,v}\geq0$) towards someone in the opposite faction is cut. This ensures a strict structural balance between the groups.
      \end{itemize}
  \item Fourth, the surviving nodes in $A\cup B$ now form a valid strong-correlation balanced module (SCBM) candidate centered on the seed.
  \item Fifth, we try to grow this core module. New nodes can only join if they:
      \begin{itemize}
        \item Have a strong correlation ($|C_{ij}|\geq\sigma$) to every node in $A\cup B$.
        \item Show uniformly positive connections to every member of one entire faction and uniformly negative connections to every member of the other faction(maintaining structural balance in Equation (\ref{balancedtriangle})).
      \end{itemize}
  \item Sixth, we run this process for the top 100 seeds (prioritized by impact) and track the biggest valid module found – the LSCBM.
\end{itemize}

\begin{algorithm}
\caption{\textbf{MaxBalanceCore}}
\label{alg:MaxBalanceCore}
\begin{algorithmic}[1]
\Require Statistically validated correlation matrix $\widetilde{\mathbf{C}} \in [-1,1]^{N\times N}$, strength threshold $\sigma > 0$
\Ensure The largest strong-correlation balanced module $LSCBM$
\State Construct signed adjacency matrix $\mathbf{S}$ where $\mathbf{S}_{i,j} = \begin{cases}
      \operatorname{sign}(\widetilde{\mathbf{C}}_{i,j}) & \text{if } |\widetilde{\mathbf{C}}_{i,j}| \geq \sigma \text{ and } i \neq j \\
      0 & \text{otherwise}
   \end{cases}$
\State Compute node impact: $\text{impact}_i \gets \sum_{j=1}^{N} \mathbb{I}[\mathbf{S}_{i,j} \neq 0]$ for $i = 1,\dots,N$
\State $\text{order} \gets \text{sort indices by } \text{impact} \text{ descending}$
\State $\text{best\_module} \gets \emptyset$
\State $\text{best\_size} \gets 0$

\For{$i \gets 1$ \textbf{to} $\min(100, N)$}
    \State $\text{seed} \gets \text{order}[i]$
    \If{$\text{impact}_{\text{seed}} = 0$}
        \State \textbf{continue}
    \EndIf
    \State $\text{neighbors} \gets \{j \mid \mathbf{S}_{\text{seed},j} \neq 0\}$
    \State $A \gets \{\text{seed}\} \cup \{j \in \text{neighbors} \mid \mathbf{S}_{\text{seed},j} > 0\}$
    \State $B \gets \{j \in \text{neighbors} \mid \mathbf{S}_{\text{seed},j} < 0\}$

    \For{$\text{group} \in \{A, B\}$}
        \If{$|\text{group}| \geq 2$}
            \State Remove $u \in \text{group}$ if $\exists v \in \text{group}\ (u \neq v \land \mathbf{S}_{u,v} \neq 1)$
        \EndIf
    \EndFor

    \If{$A \neq \emptyset$ \textbf{and} $B \neq \emptyset$}
        \State Remove $u \in A$ if $\exists v \in B\ (\mathbf{S}_{u,v} \geq 0)$
        \State Remove $v \in B$ if $\exists u \in A\ (\mathbf{S}_{u,v} \geq 0)$
    \EndIf

    \State $\text{module} \gets A \cup B$
    \State $\text{candidates} \gets \{1,\dots,N\} \setminus \text{module}$
    \State $\text{strong\_candidates} \gets \{ \text{node} \in \text{candidates} \mid \forall j \in \text{module},\ |\mathbf{S}_{\text{node},j}| \geq \sigma \}$

    \For{$\text{node} \in \text{strong\_candidates}$}
        \State $\text{canJoinA} \gets (\forall u \in A,\ \mathbf{S}_{\text{node},u} = 1) \land (\forall v \in B,\ \mathbf{S}_{\text{node},v} = -1)$
        \State $\text{canJoinB} \gets (\forall u \in A,\ \mathbf{S}_{\text{node},u} = -1) \land (\forall v \in B,\ \mathbf{S}_{\text{node},v} = 1)$

        \If{$\text{canJoinA}$}
            \State $A \gets A \cup \{\text{node}\}$
            \State $\text{module} \gets \text{module} \cup \{\text{node}\}$
        \ElsIf{$\text{canJoinB}$}
            \State $B \gets B \cup \{\text{node}\}$
            \State $\text{module} \gets \text{module} \cup \{\text{node}\}$
        \EndIf
    \EndFor

    \If{$|\text{module}| > \text{best\_size}$}
        \State $\text{best\_module} \gets \text{module}$
        \State $\text{best\_size} \gets |\text{module}|$
    \EndIf
\EndFor

\State \Return $\text{best\_module}$
\end{algorithmic}
\end{algorithm}
\begin{figure}[htbp]
\centering
\includegraphics[height=0.98\textheight,keepaspectratio]{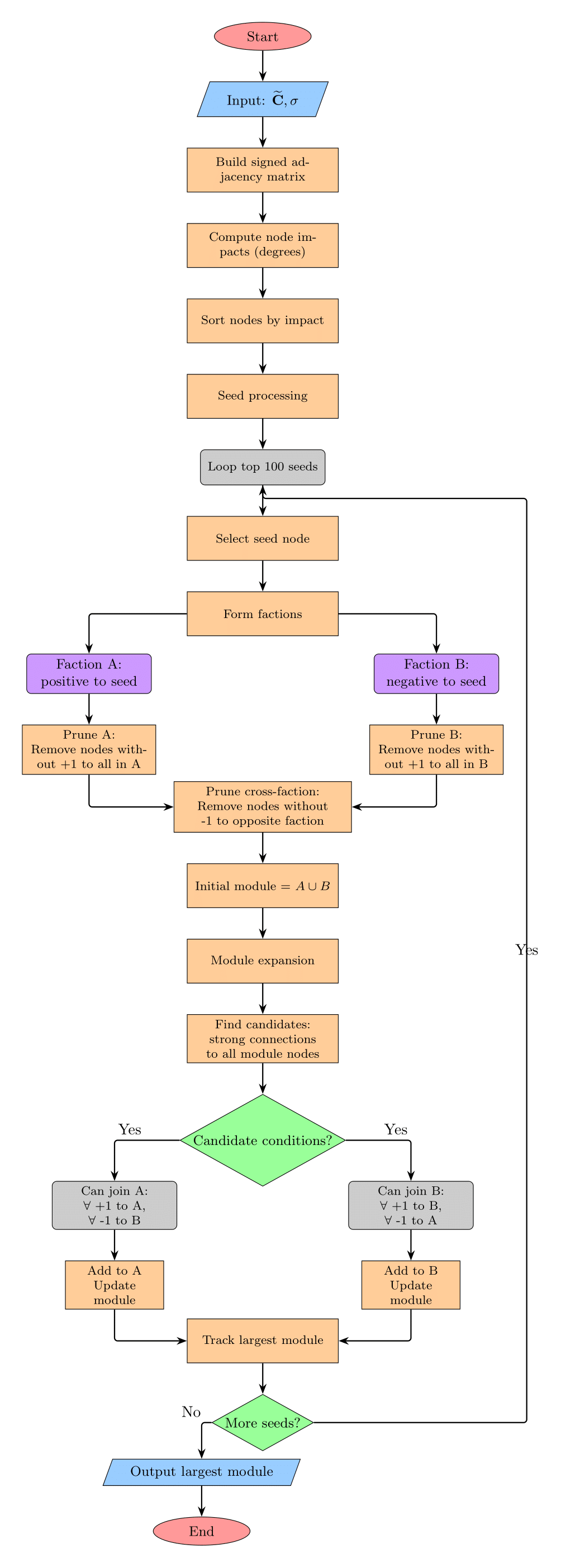}
\caption{Flowchart of our MaxBalanceCore algorithm.}
\label{Flowchart}
\end{figure}

The steps above are summarized in Algorithm \ref{alg:MaxBalanceCore}, where we name our algorithm as MaxBalanceCore because it specifically seeks the maximum-sized module while enforcing structural balance conditions. Figure \ref{Flowchart} displays the flowchart of our MaxBalanceCore algorithm.

\textbf{Complexity analysis.} The time complexity of our proposed MaxBalanceCore algorithm is dominated by the construction of the signed adjacency matrix $S$ ($O(N^{2})$) and the iterative processing of high-impact seeds (up to 100 seeds). For each seed, pruning incompatible nodes involves checking pairwise relationships within subsets $A$ and $B$, which scales as $O(N^{2})$ in the worst case. Module expansion further requires validating candidate nodes against all existing module members, contributing $O(N^{2})$ per seed. Thus, the overall time complexity is $O(N^{2})$. The space complexity is primarily determined by storing the signed adjacency matrix $S$ and auxiliary data structures (e.g., node impact scores, module candidates), resulting in $O(N^{2})$ space due to the dense matrix representation. While sparsity (controlled by $\sigma$) reduces practical computational load, the worst-case bounds remain quadratic in both time and space. This approach stays manageable for huge stock correlation networks (thousands of stocks) because of the following three key choices:
\begin{itemize}
  \item The algorithm initiates searches exclusively from high-degree nodes (prioritized by impact scores $\text{impact}_i$). These hub nodes exhibit a higher statistical likelihood of anchoring large modules. This strategic restriction reduces the number of starting points while maximizing the potential for identifying large solutions early in the search process.
  \item Before module expansion, the algorithm rigorously prunes incompatible nodes using structural balance theory \citep{harary1953notion,cartwright1956structural}. After partitioning a seed's neighbors into faction $A$ and faction $B$, nodes violating strict intra-faction harmony (all $A$-$A$ and $B$-$B$ ties must be $+1$) or inter-faction antagonism (all $A$-$B$ ties must be $-1$) are eliminated. This step drastically reduces the candidate set before computationally intensive expansion, thereby limiting combinatorial growth.
  \item The module expansion phase leverages the inherent sparsity of the statistically validated correlation matrix $\widetilde{\mathbf{C}}$ and the correlation strength threshold $\sigma$. Candidate nodes must satisfy two conditions:
      \begin{itemize}
        \item (i) A strong correlation ($|\widetilde{\mathbf{C}}_{i,j}| \geq \sigma$) exists with every current module member.
        \item (ii) Uniform sign alignment across entire factions (e.g., all ties to $A$ are $+1$ and all ties to $B$ are $-1$, or vice versa).
      \end{itemize}
  These conditions are highly selective in sparse networks, ensuring very few candidates qualify for evaluation. Consequently, the per-iteration computational burden remains manageable.
\end{itemize}

Although the MaxBalanceCore algorithm cannot guarantee identification of the exact LSCBM, discovering a large SCBM holds significant value from both algorithmic and practical perspectives. From an algorithmic perspective, identifying the exact LSCBM is an NP-hard problem, implying that the computational complexity of finding an exact solution would grow exponentially with the scale of the network. Our MaxBalanceCore algorithm employs efficient heuristic methods, leveraging structural balance theory and correlation strength thresholds to efficiently search for large SCBM, thereby avoiding combinatorial explosions and providing practical and scalable solutions for large-scale financial networks within a reasonable timeframe. Such trade-offs are necessary when dealing with complex networks, as exact solutions are often impractical in reality.

In terms of practical applications, the core objective of stock market analysis is to identify groups of stocks that exhibit strong correlations and stable relationships. A large SCBM can offer crucial insights into market structure by revealing which stocks exhibit economically significant and stable relationships. This stability is particularly vital for portfolio design and risk management. For instance, investors can use the stocks identified within an SCBM to construct portfolios with inherent hedging mechanisms or to focus on specific industry groups or market trends. Therefore, while the MaxBalanceCore algorithm may not guarantee identification of the absolute largest SCBM, the large SCBM it identifies is sufficient to meet the needs of financial analysis and investment decision-making.
\section{Experimental evaluation}\label{Sec4}
In this section, we present comprehensive experimental evaluations to validate the MaxBalanceCore algorithm and the proposed LSCBM framework. We first conduct simulation studies to assess the accuracy and efficiency of the algorithm, and verify the theoretical scaling laws for LSCBM's size under different network regimes. We then perform empirical analysis using Chinese stock market data to demonstrate the framework's utility in capturing dynamic market reorganizations during economic events.
\subsection{Simulation studies}
\subsubsection{Performance evaluation of MaxBalanceCore}
In this part, to validate the accuracy and efficiency of our MaxBalanceCore algorithm, we construct synthetic statistically validated correlation networks where the true LSCBMs are known as follows: Suppose there are $N$ nodes and the threshold is $\sigma=0.7$. We first randomly partition $(N_{A}+N_{B})$ nodes into two disjoint sets: set $A$ with $N_{A}$ nodes and set B with $N_{B}$ nodes, ensuring all intra-set connections within $A$ or $B$ are strongly positive (edge weight = +1), while all inter-set connections between A and B are strongly negative (edge weight = -1). The remaining $(N-N_{A}-N_{B})$ nodes (set $R$) are weakly correlated with all other nodes in the network, where any pairwise connection involving set $R$ has absolute correlation strength strictly below the threshold $\sigma$. Specifically, these edges are absent (weight = 0) with probability 0.3, weakly positive (uniformly sampled from $(0,\sigma)$) with probability 0.35, or weakly negative (uniformly sampled from $(-\sigma,0)$) with probability 0.35. This configuration guarantees that the ground-truth largest strong-correlation balanced module is precisely the union of sets A and B, satisfying both the minimum correlation strength ($|\widetilde{\mathbf{C}}_{i,j}| \geq \sigma$) and structural balance conditions. For each simulation study, $N, N_{A}$, and $N_{B}$ are set independently. We say that our MaxBalanceCore correctly recovers LSCBM if the nodes of the output of MaxBalanceCore are exactly the same as those in LSCBM. To measure MaxBalanceCore's accuracy, we use the Accuracy rate defined as the ratio of correctly estimating LSCBM to the total number of independent trials. For each simulation setting, we consider 100 independent replicates in this article. Finally, we should emphasize that since LSCBM is a new concept proposed in this work, no prior algorithms exist to detect it. Our MaxBalanceCore algorithm is the first specialized solution designed for identifying LSCBM in statistically validated correlation networks. Consequently, we are unable to include direct algorithmic comparisons in our numerical experiments.
\begin{figure}[H]
\centering
{\includegraphics[width=0.3\textwidth]{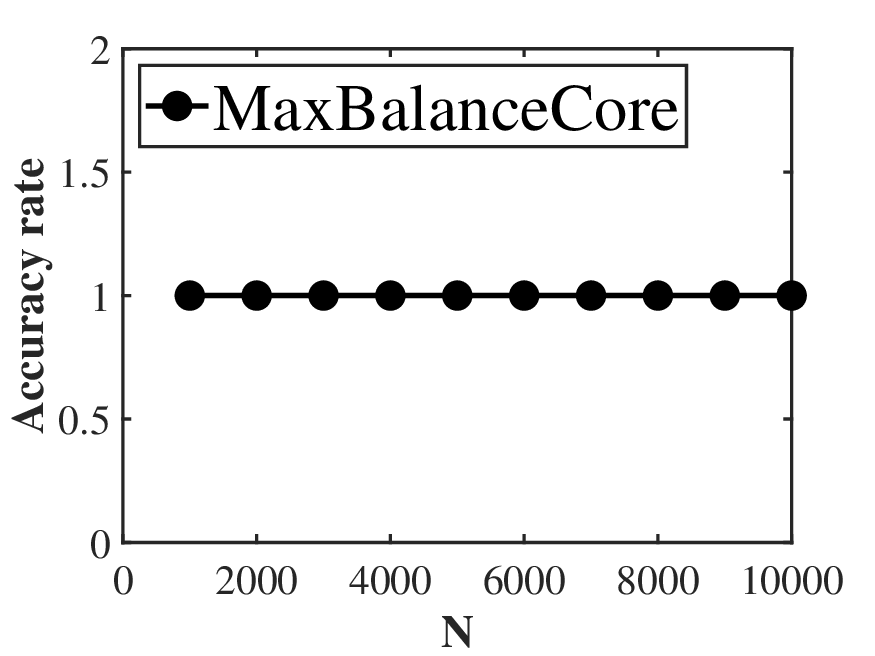}}
{\includegraphics[width=0.3\textwidth]{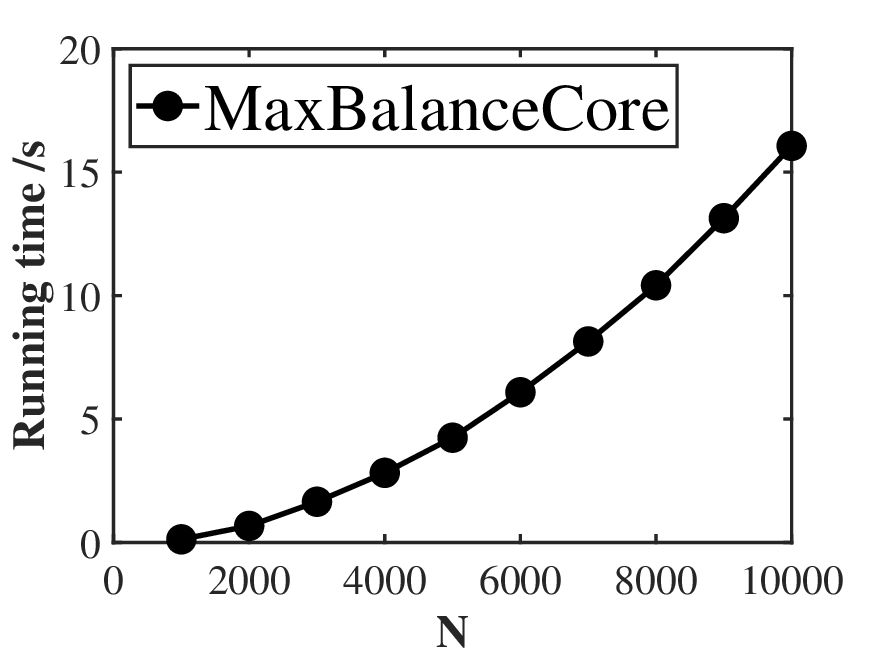}}
\caption{Left: Accuracy rate against $N$. Right: Running time against $N$.}
\label{Sim1} 
\end{figure}
\begin{figure}[H]
\centering
{\includegraphics[width=0.3\textwidth]{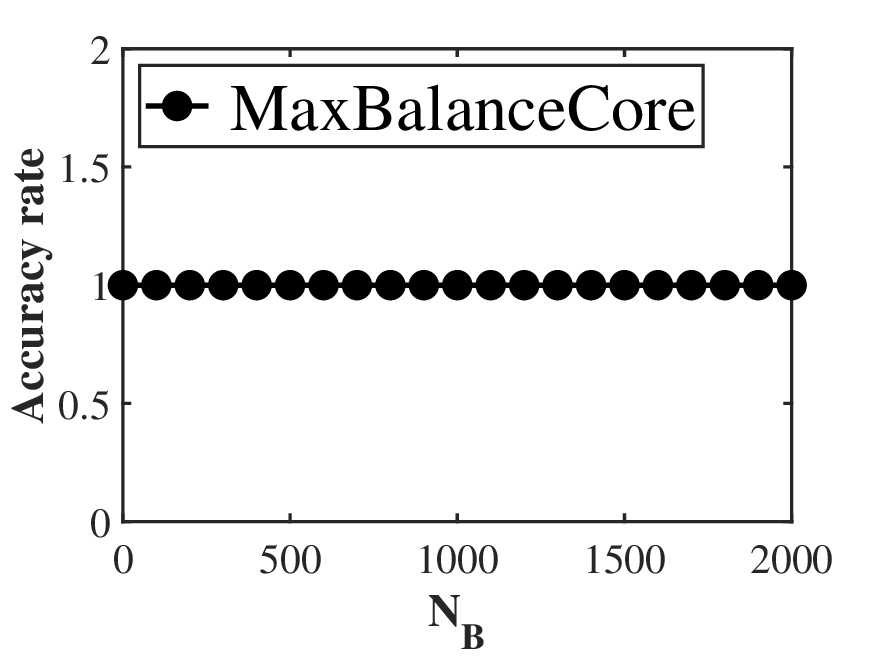}}
{\includegraphics[width=0.3\textwidth]{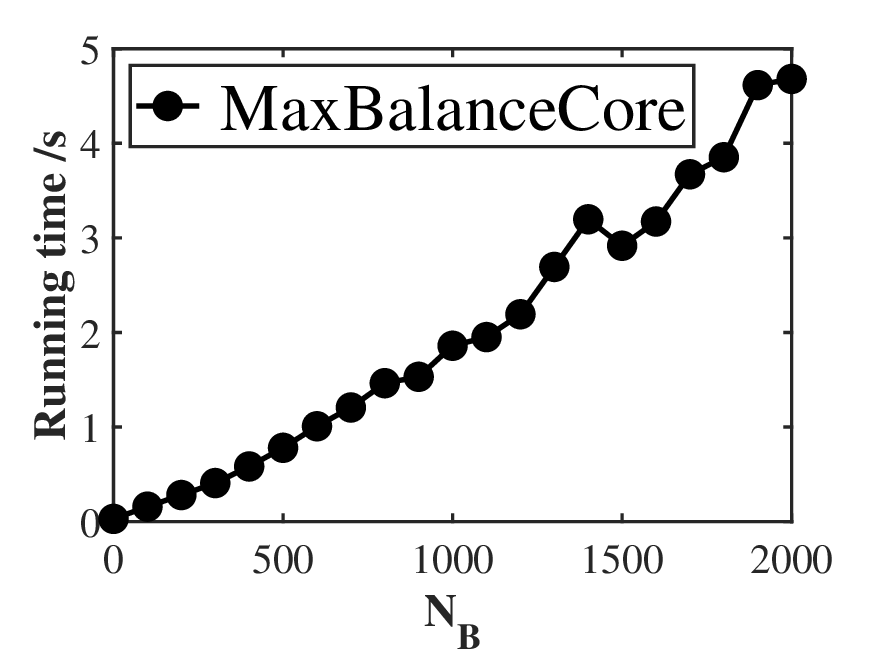}}
\caption{Left: Accuracy rate against $N_{B}$. Right: Running time against $N_{B}$.}
\label{Sim2} 
\end{figure}
\texttt{Simulation study 1:changing $N$.} For this simulation, we set $N_{A}=N/10, N_{B}=N/5$, and vary $N$ in $\{1000,2000,\ldots,10000\}$. The results are shown in Figure \ref{Sim1}. Our MaxBalanceCore algorithm consistently identifies the true LSCBM across all tested network sizes. While runtime scales with N, the algorithm efficiently processes networks of up to 10000 nodes within 20 seconds.

\texttt{Simulation study 2:changing $N_{B}$.} For this simulation, we set $N=3000, N_{A}=20$, and vary $N_{B}$ in $\{0, 100,200,\ldots,2000\}$. Figure \ref{Sim2} presents the results. Our MaxBalanceCore algorithm always recovers the true LSCBM exactly, even in cases of highly asymmetric modules where the size of set $B$ significantly exceeds that of set $A$ (or $B$ is empty). The right panel of the figure shows that as the size of the LSCBM increases, the running time also increases, but it remains feasible for practical applications.
\subsubsection{Verification of theoretical scaling}
To validate Theorems \ref{main1}-\ref{main3}, we conduct simulations where the signed graphs are generated from the model $\mathcal{G}(N,\alpha,\beta)$ with node counts \(N\) ranging in $\{10,20,\ldots,200\}$ or $\{300,600,\ldots,6000\}$, using fixed parameters \(\alpha = 0.6\), \(\beta = 0.3\) for Theorem \ref{main1}, parameterized settings \(\alpha = 1 - b/N\), \(\beta = b/N\) with \(b = 2\) for Theorem \ref{main2}, and \(\alpha = 1/\sqrt{N}\), \(\beta = 1 -1/\sqrt{N}\) for Theorem \ref{main3}. For each \(N\), we generate graphs, compute the size of LSCBM returned by our MaxBalanceCore, and record the ratio of the observed size to its theoretical prediction (i.e., \(\log N / \lambda\) for Theorem \ref{main1}, \(N \log b / b\) for Theorem \ref{main2}, and \(\log N / |\log \alpha|\) for Theorem \ref{main3}) over 100 independent trials. The mean ratios across these trials are then normalized by their collective average over all \(N\), and these normalized ratios are plotted against \(N\) to verify asymptotic convergence to unity. The numerical results presented in Figure \ref{Sim3}   strongly validate the theoretical scaling predictions for LSCBM's size across different random graph regimes. In the general regime of Theorem \ref{main1}, the detected LSCBM size shows remarkable convergence toward the predicted \(\log N / \lambda\) scaling, with minimal deviation across increasing $N$. For the dense regime of Theorem 2, the results demonstrate the \(N \log b / b\) scaling, with observed sizes tightly aligning with theoretical expectations. In the negative-dominated regime of Theorem 3, despite greater sparsity constraints, the numerical results still closely follow the predicted \(\log N / |\log \alpha|\) scaling law.  Across all configurations, the results consistently support the theoretical framework's accuracy in predicting the size of LSCBM.
\begin{figure}[H]
\centering
{\includegraphics[width=0.3\textwidth]{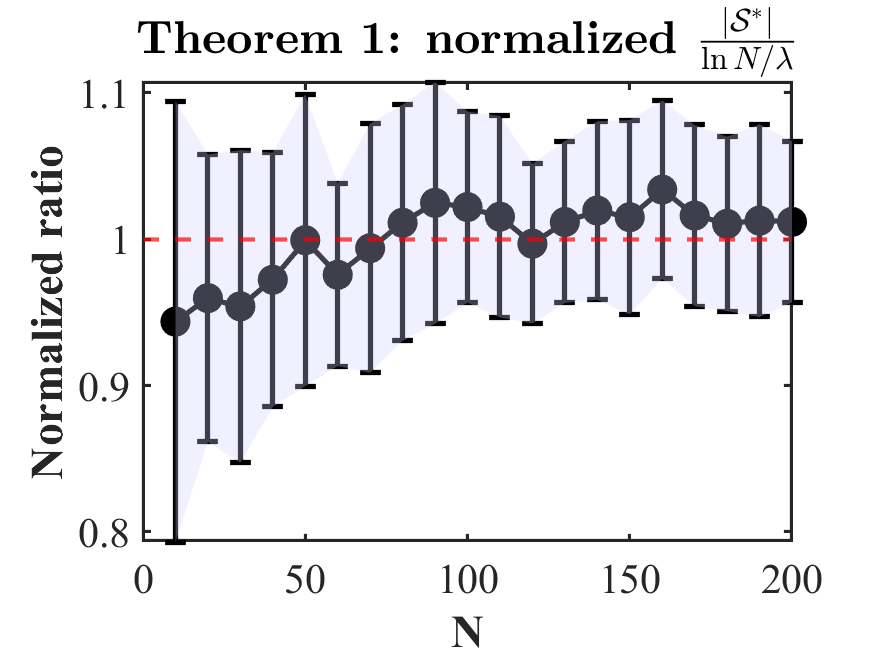}}
{\includegraphics[width=0.3\textwidth]{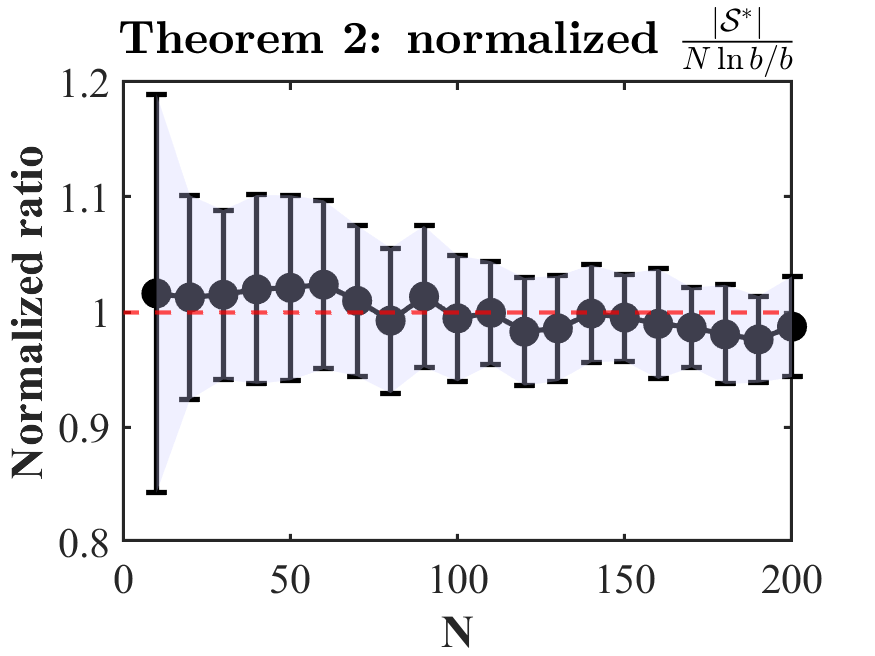}}
{\includegraphics[width=0.3\textwidth]{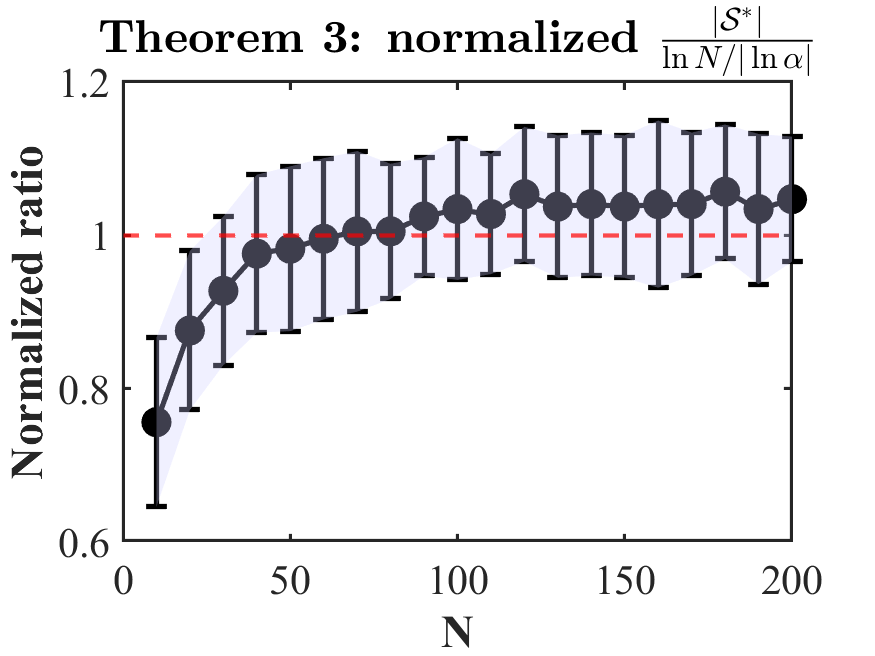}}
{\includegraphics[width=0.3\textwidth]{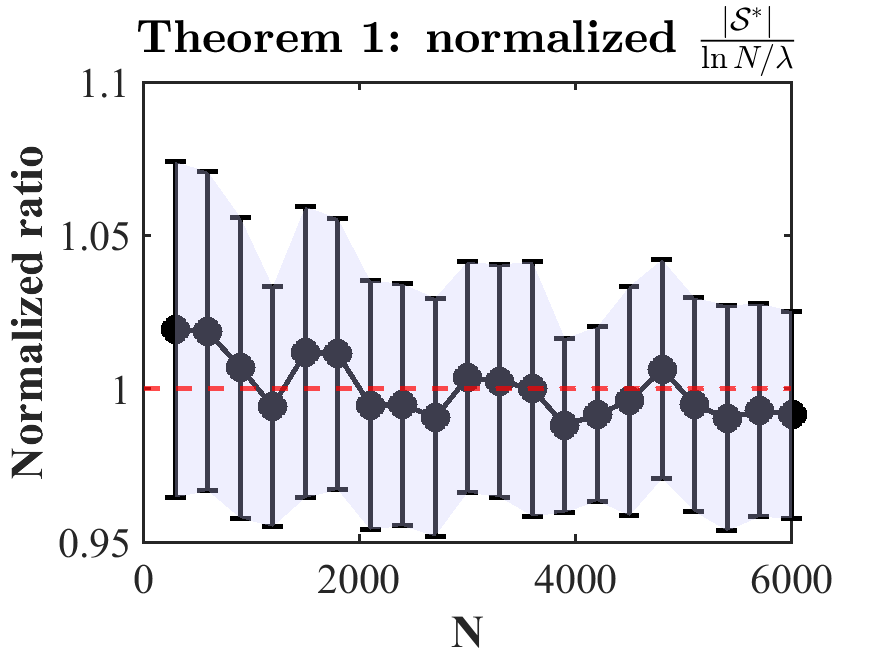}}
{\includegraphics[width=0.3\textwidth]{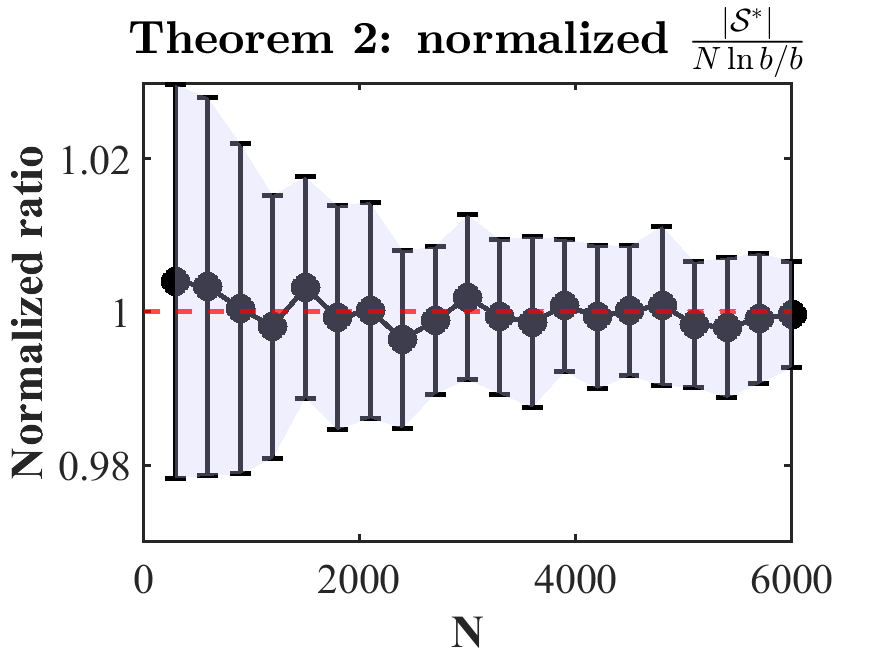}}
{\includegraphics[width=0.3\textwidth]{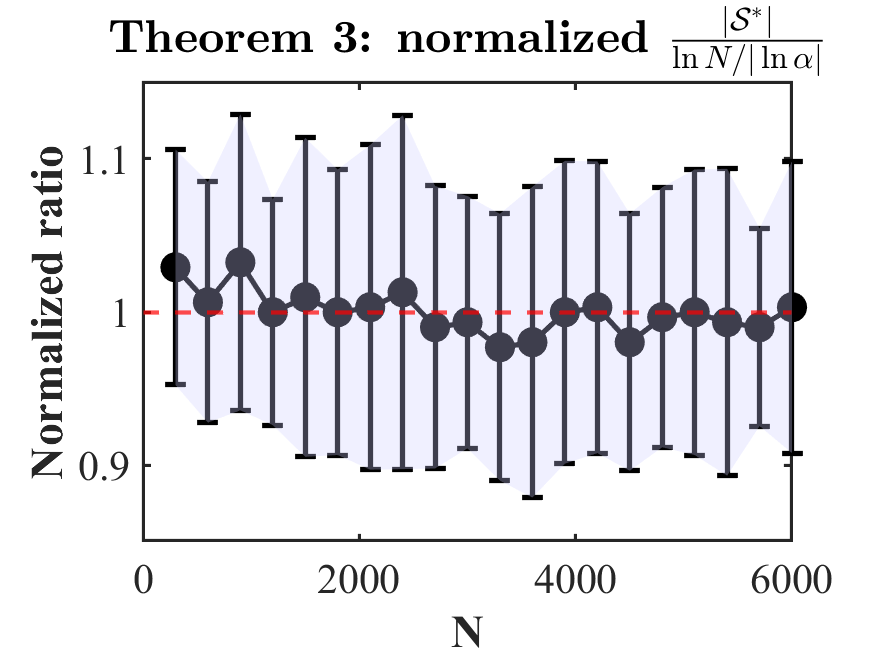}}
\caption{Normalized ratios against $N$.}
\label{Sim3} 
\end{figure}
\subsection{Empirical analysis}
To empirically validate the proposed LSCBM framework and explore its practical implications in real financial markets, we leverage stock data from the RESSET \footnote{\url{www.resset.com}}  database, a primary source for Chinese financial data. Specifically, we collect daily closing price data for all listed stocks on major Chinese exchanges (Shanghai and Shenzhen) across twelve distinct annual periods from 2013 to 2024. This longitudinal design intentionally spans diverse market regimes, including extreme volatility induced by the 2015 Chinese stock market crash (characterized by leveraged sell-offs, circuit breakers, and systemic contagion), periods of relative stability (e.g., 2016–2017), and heightened uncertainty during global events like the COVID-19 pandemic (2020). This dynamic, year-by-year approach allows us to move beyond static snapshots and instead capture the time-varying evolution of market structure under both endogenous shocks (e.g., the 2015 crash) and exogenous crises. To ensure robustness, we rigorously preprocess the data by deleting stocks with missing data. For each year \(m \in \{2013, 2014, \ldots, 2024\}\), we compute the statistically validated correlation matrix \(\widetilde{\mathbf{C}}_m\) using the rigorous t-test procedure outlined in Section \ref{Sec2}.

To characterize the basic properties of the statistically validated stock correlation networks, we define the following metrics for each annual network \(\widetilde{\mathbf{C}}_{m}\) (\(m \in \{2013, 2014, \ldots, 2024\}\)):
\begin{itemize}
  \item Let \(\xi_{+}\) and \(\xi_{-}\) denote the proportions of positive and negative elements in \(\widetilde{\mathbf{C}}_{m}\), respectively, excluding diagonal elements.
  \item Let \(\mu_{+}\) and \(\mu_{-}\) represent the average values of the positive and negative elements in \(\widetilde{\mathbf{C}}_{m}\), respectively, excluding diagonal elements.
  \item Define \(\varsigma:= \frac{|\mathcal{S}^{*}|}{N}\) as the proportion of nodes belonging to the LSCBM \(\mathcal{S}^{*}\) detected by the MaxBalanceCore algorithm relative to the total number of stocks \(N\).
\end{itemize}

\begin{table}[h!]
\centering
\caption{Basic properties of the statistically validated stock networks considered in this article.}
\label{realdataBasic}
\footnotesize  
\begin{tabular}{*{9}{c}}  
\hline
\(\widetilde{\mathbf{C}}\)&$N$&$T$&$\xi_{+}$&$\xi_{-}$&$\mu_{+}$&$\mu_{-}$&$|\mathcal{S}^{*}|$&$\varsigma$\\
\hline
\(\widetilde{\mathbf{C}}_{2013}\)&1462&237&0.9295&0.000077716&0.3241&-0.1592&13&0.0089\\
\(\widetilde{\mathbf{C}}_{2014}\)&1101&244&0.8920&0.00031542&0.2919&-0.1537&15&0.0136\\
\(\widetilde{\mathbf{C}}_{2015}\)&609&243&0.9939&0.0000054014&0.5574&-0.1541&55&0.0903\\
\(\widetilde{\mathbf{C}}_{2016}\)&1364&243&0.9761&0.000017212&0.4762&-0.1531&87&0.0638\\
\(\widetilde{\mathbf{C}}_{2017}\)&1841&243&0.7783&0.0075&0.2926&-0.1733&14&0.0076\\
\(\widetilde{\mathbf{C}}_{2018}\)&2566&242&0.9694&0.000020055&0.3830&-0.1565&35&0.0135\\
\(\widetilde{\mathbf{C}}_{2019}\)&3155&243&0.9625&0.000018893&0.3428&-0.1478&27&0.0086\\
\(\widetilde{\mathbf{C}}_{2020}\)&3248&242&0.9183&0.00010620&0.3037&-0.1483&24&0.0074\\
\(\widetilde{\mathbf{C}}_{2021}\)&3537&242&0.4902&0.0035&0.2102&-0.1542&7&0.0020\\
\(\widetilde{\mathbf{C}}_{2022}\)&3943&241&0.8886&0.0003246&0.2986&-0.1577&22&0.0056\\
\(\widetilde{\mathbf{C}}_{2023}\)&4316&241&0.5848&0.0027&0.2354&-0.1583&31&0.0072\\
\(\widetilde{\mathbf{C}}_{2024}\)&4515&241&0.9702&0.00058595&0.4174&-0.1586&113&0.0250\\
\hline
\end{tabular}
\end{table}

The longitudinal analysis of statistically validated stock correlation networks for Chinese stock markets, as presented in Table \ref{realdataBasic}, reveals profound insights into market structural dynamics, particularly when combined within major economic events. The network properties exhibit significant year-to-year variations, directly reflecting shifts in market regimes driven by both endogenous shocks and exogenous crises. Critically, the proportion of statistically significant positive correlations $\xi_{+}$ dominates throughout the period, consistently exceeding 49.02\% and reaching peaks such as 99.39\% in 2015. This overwhelming prevalence underscores the strong co-movement tendencies inherent in emerging markets, especially during periods of stress. Conversely, the proportion of statistically significant negative correlations $\xi_{-}$ remains extremely low ($\leq0.75\%$), highlighting the scarcity of robust hedging opportunities within the market structure. The average strength of positive correlations $\mu_{+}$ and negative correlations $\mu_{-}$ also fluctuates, with $\mu_{+}$ ranging from 0.2102 to 0.5574 and $\mu_{-}$ consistently between -0.1478 and -0.1733, indicating that validated negative relationships, while rare, exhibit economically meaningful strength when present. Notably, $\mu_{+}$ peaks in 2015 and remains the second-highest in 2016, indicating that the correlation networks formed during these stock-market crises exhibit exceptionally strong positive linkages. This aligns with the findings observed in \citep{XIA2018222,HE2022121732}.

The size of LSCBM $|\mathcal{S}^{*}|$ and its proportion relative to the total stocks $\varsigma$ serve as crucial indicators of market stability and integration. The year 2015 stands out dramatically: $|\mathcal{S}^{*}|$ surges to 55 ($\varsigma= 9.03\%$), coinciding precisely with the Chinese stock market crash. This event, characterized by a leveraged bubble burst, panic selling, and systemic contagion, forced extreme synchronization across stocks. The statistically validated network captures this: $\xi_{+}$ reaches 99.39\%, $\mu_{+}$ jumps to 0.5574, and the large LSCBM size signifies a market-wide collapse into a highly correlated, unstable state where diversification benefits largely vanish. The structural balance within this large module, while adhering to theory, reflects a fragile cohesion driven by uniform panic rather than fundamental alignment. The following years (2016-2017) show a partial normalization, with $|\mathcal{S}^{*}|$ decreasing to 87 ($\varsigma= 6.38\%$) in 2016 and further to 14 ($\varsigma= 0.76\%$) in 2017. This reduction aligns with the post-crisis stabilization phase, circuit breaker implementations, and regulatory interventions, allowing for some return of special stock behavior and reduced market-wide coupling, evidenced by the decline in $\mu_{+}$ to 0.2926 in 2017.

The period encompassing the COVID-19 pandemic (2020-2021) reveals a distinct two-phase pattern. In 2020, the initial global shock leads to another surge in co-movement, reflected in $\xi_{+} = 91.83\%$ and $|\mathcal{S}^{*}|=24$ ($\varsigma= 0.74\%$). While significant, the LSCBM size is notably smaller than during the 2015 crash, suggesting a slightly less uniform panic. However, 2021 exhibits a stark reversal: $\xi_{+}$ plummets to 49.02\%, its lowest value in the dataset, and $|\mathcal{S}^{*}|$ collapses to a minimal 7 ($\varsigma= 0.20\%$). This fragmentation coincides with the divergent recovery paths of sectors and companies evolving during pandemic waves, supply chain disruptions, and heterogeneous policy responses. The market transitioned from a synchronized crash to a phase where company-specific fundamentals and sectoral exposures regained prominence, hindering the formation of large, strongly correlated, and structurally balanced modules. The years 2022-2024 show a gradual resurgence of connectivity. $|\mathcal{S}^{*}|$ increases to 22 ($\varsigma= 0.56\%$) in 2022 and 31 ($\varsigma= 0.72\%$) in 2023, potentially linked to ongoing global macroeconomic uncertainty (inflation, rate hikes) and domestic concerns like the property sector crisis, which may have induced broader risk-off sentiments. Notably, 2024 exhibits a significant jump to $|\mathcal{S}^{*}|=113$ ($\varsigma= 2.50\%$), the second-largest module observed. This could reflect responses to major policy shifts. One possible reason is, amid the protracted downturn in China’s real estate sector since 2022, the recalibration and escalation of U.S. tariff measures on Chinese exports in 2024 have further compounded existing structural headwinds. China’s macroeconomic environment during this period suffers from the heightened policy uncertainty, weak domestic demand, and a broadly adverse economic outlook.

Meanwhile, we find that within all identified LSCBMs across the twelve annual periods of the Chinese stock market (2013-2024), every statistically validated pairwise correlation is positive. This absence of negative correlations within these core modules signifies a critical lack of inherent hedging opportunities in the Chinese stock market. This finding aligns logically with the remarkably low proportion of statistically significant negative correlations ($\xi_{-}\leq0.0075$) shown in Table \ref{realdataBasic}. The theoretical foundation, particularly proofs of Theorem \ref{main2}, provides a lens for understanding this phenomenon: when the probability of positive edges $\alpha$ (i.e., large $\xi_{+}$)  is significantly larger than the probability of positive edges $\beta$ (i.e., small or even close to zero $\xi_{-}$), the emergent LSCBMs are overwhelmingly composed of positively correlated stocks. This theoretical prediction appears clearly in the Chinese stock market data.

This consistent positivity within the LSCBMs underscores a fundamental characteristic of the core structure in the Chinese stock market: strong, stable co-movement dominates. While structural balance theory theoretically accommodates ``enemy of my enemy" configurations (two negatives and one positive) as stable, the empirical scarcity of robust, statistically significant negative correlations meeting the strength threshold $\sigma=0.7$ makes such balanced negative triangles exceedingly rare within the highly interconnected core modules identified by LSCBM. Consequently, the potential for natural hedging within these specific, densely connected, and statistically robust modules is practically absent. This observation resonates with the behavior of the Chinese stock market, often characterized by high synchronization, especially during stress events like the 2015 crash (where $\xi_{+}$ reached 99.39\%). Factors such as strong common risk factor exposures (e.g., policy shifts, macroeconomic trends), prevalent herding behavior among the large retail investor base, and sectoral interdependence likely contribute to this prevalence of positive dependencies in the core, limiting the formation of stable, strongly negatively correlated pairs suitable for hedging within these tightly knit groups. The LSCBM framework, by design, filters out weak or spurious relationships, thus revealing that the strongest and most stable interdependencies at the China stock market's core are uniformly positive, reflecting a market structure where diversification benefits derived from offsetting negative correlations within its core subsystems are minimal during these years. Furthermore, this uniform positive correlation structure within LSCBMs carries significant practical utility for portfolio construction. Since all validated pairwise relationships exhibit positive co-movement, each LSCBM effectively functions as a unified macro-exposure unit representing a distinct systemic risk factor or economic sectoral theme (e.g., the 2024 large-scale module). Consequently, investors can strategically reduce unintended concentration risk by limiting overexposure to multiple stocks within LSCBM since such holdings provide minimal diversification benefits within the module. Instead, portfolio risk management should emphasize: (i) exposure calibration across different, non-correlated LSCBM to harness true diversification, and (ii) complementary cross-asset hedging strategies to offset systemic risks emanating from these cohesive industry groups. This framework transforms LSCBM from a mere statistical concept into actionable ``risk allocation units'' for disciplined equity allocation in the A-share market.
\begin{figure}[H]
\centering
\resizebox{\columnwidth}{!}{
{\includegraphics[width=0.3\textwidth]{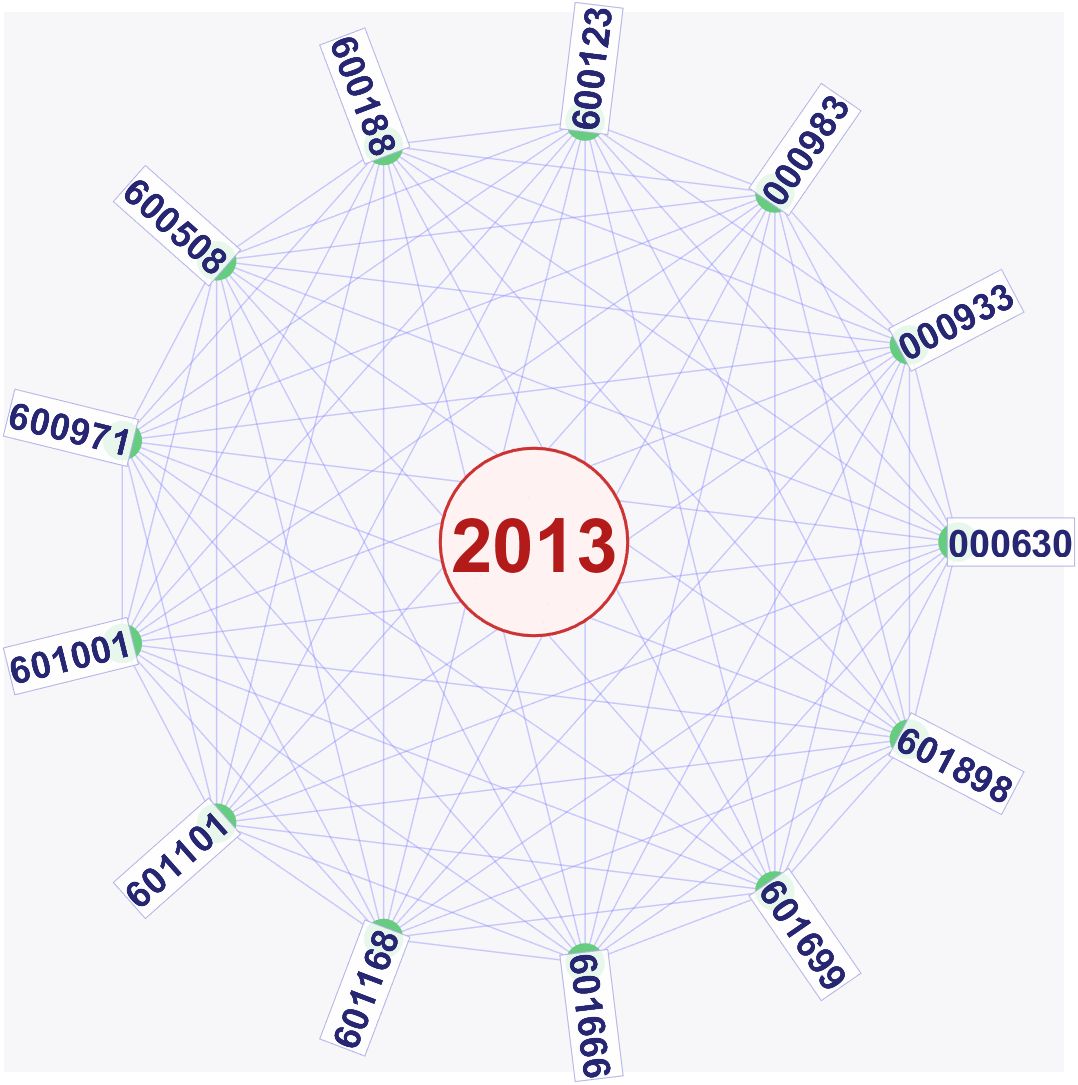}}
{\includegraphics[width=0.3\textwidth]{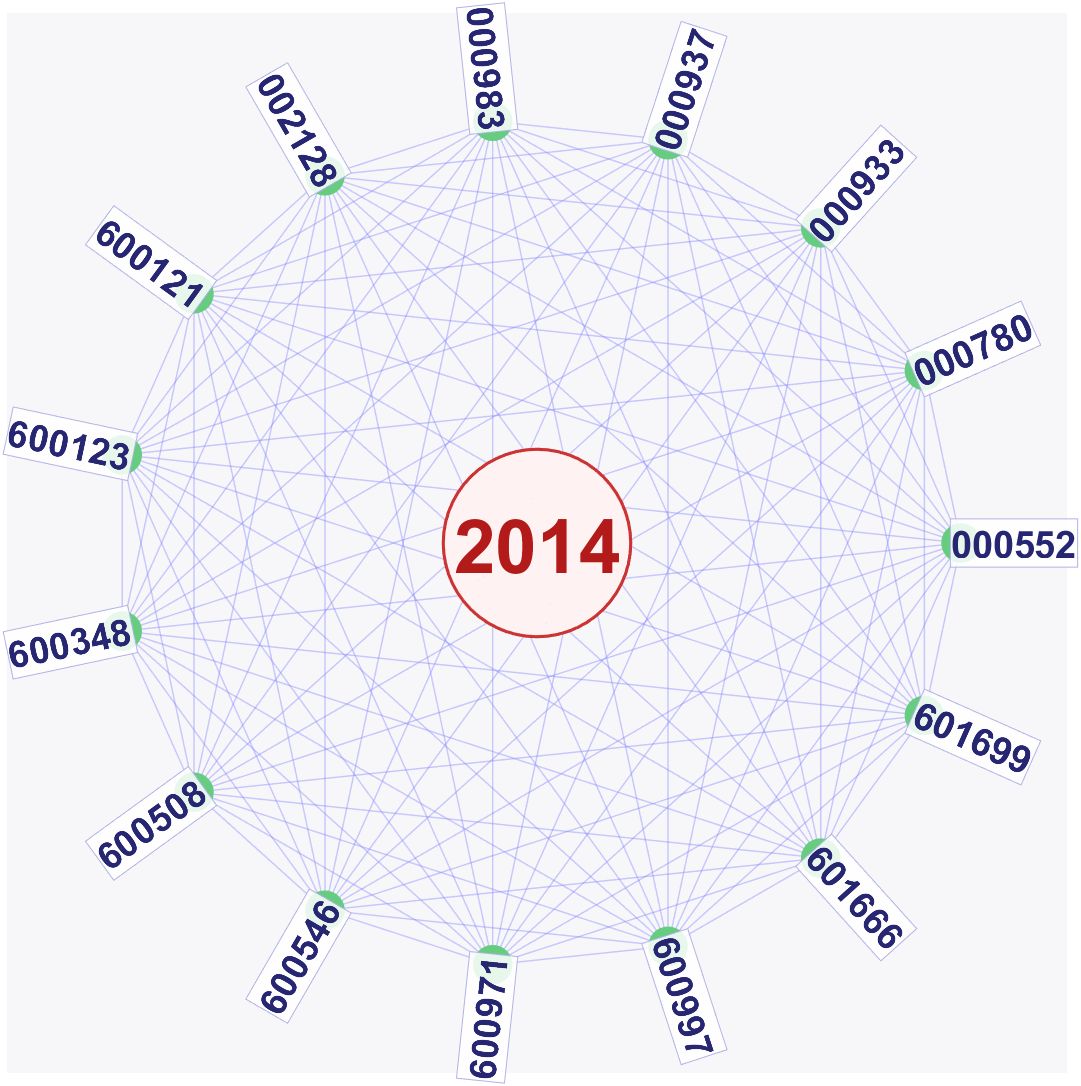}}
{\includegraphics[width=0.3\textwidth]{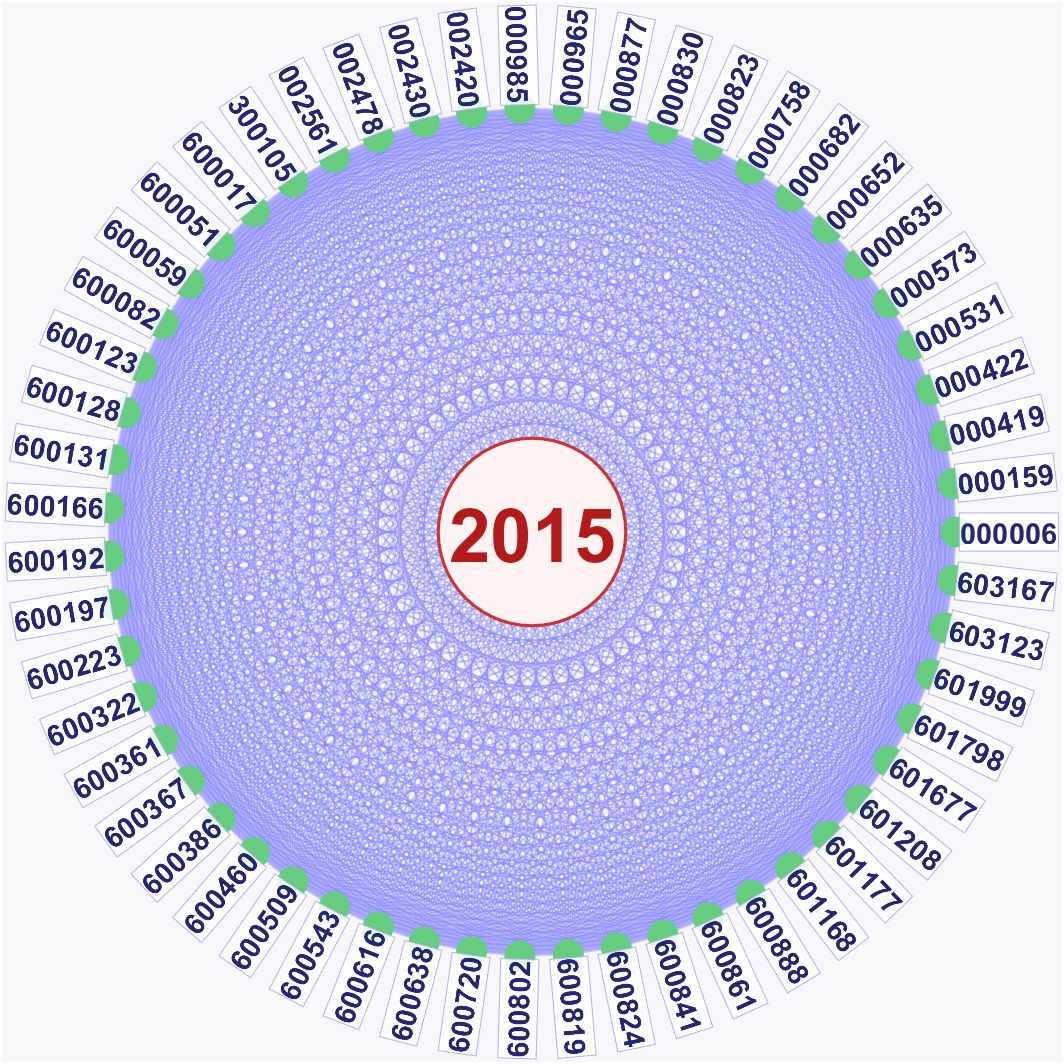}}
{\includegraphics[width=0.3\textwidth]{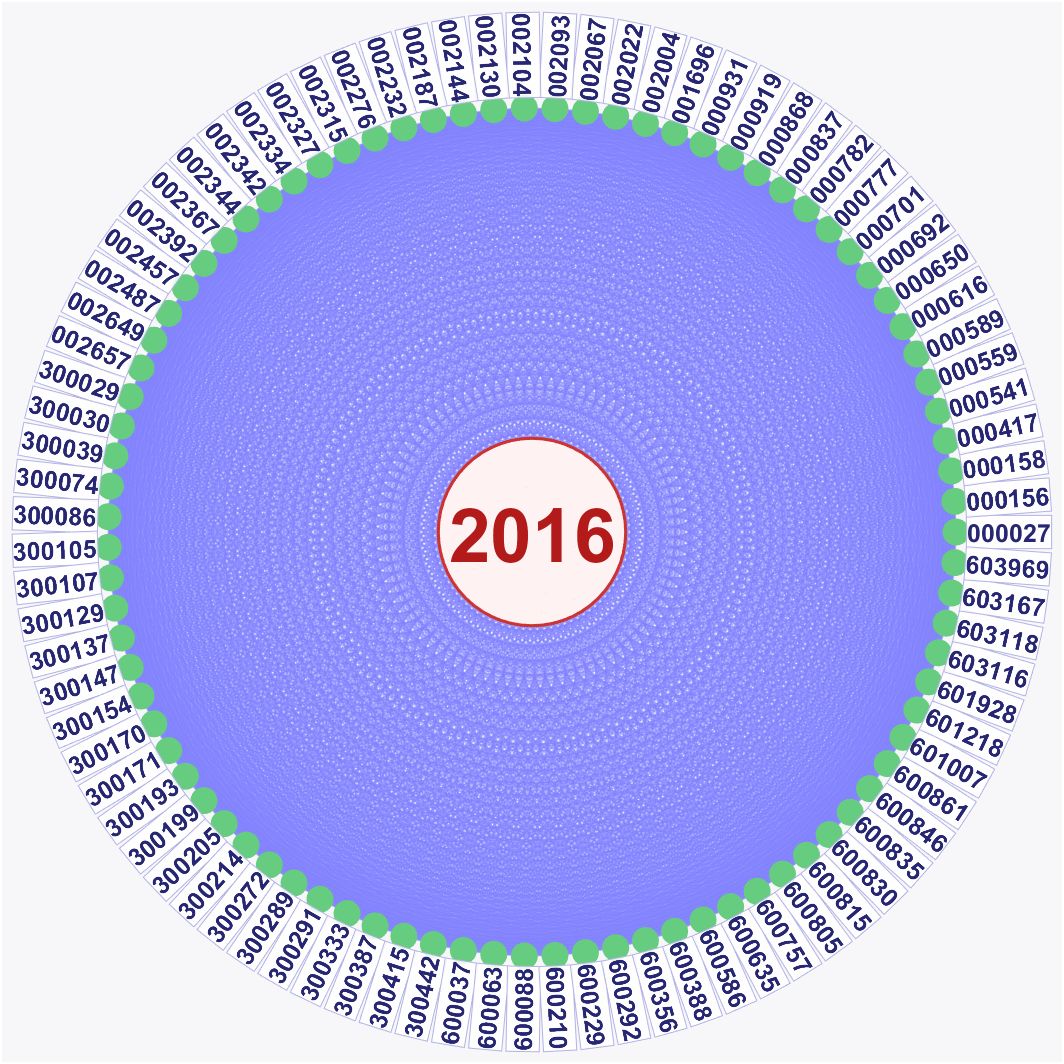}}
{\includegraphics[width=0.3\textwidth]{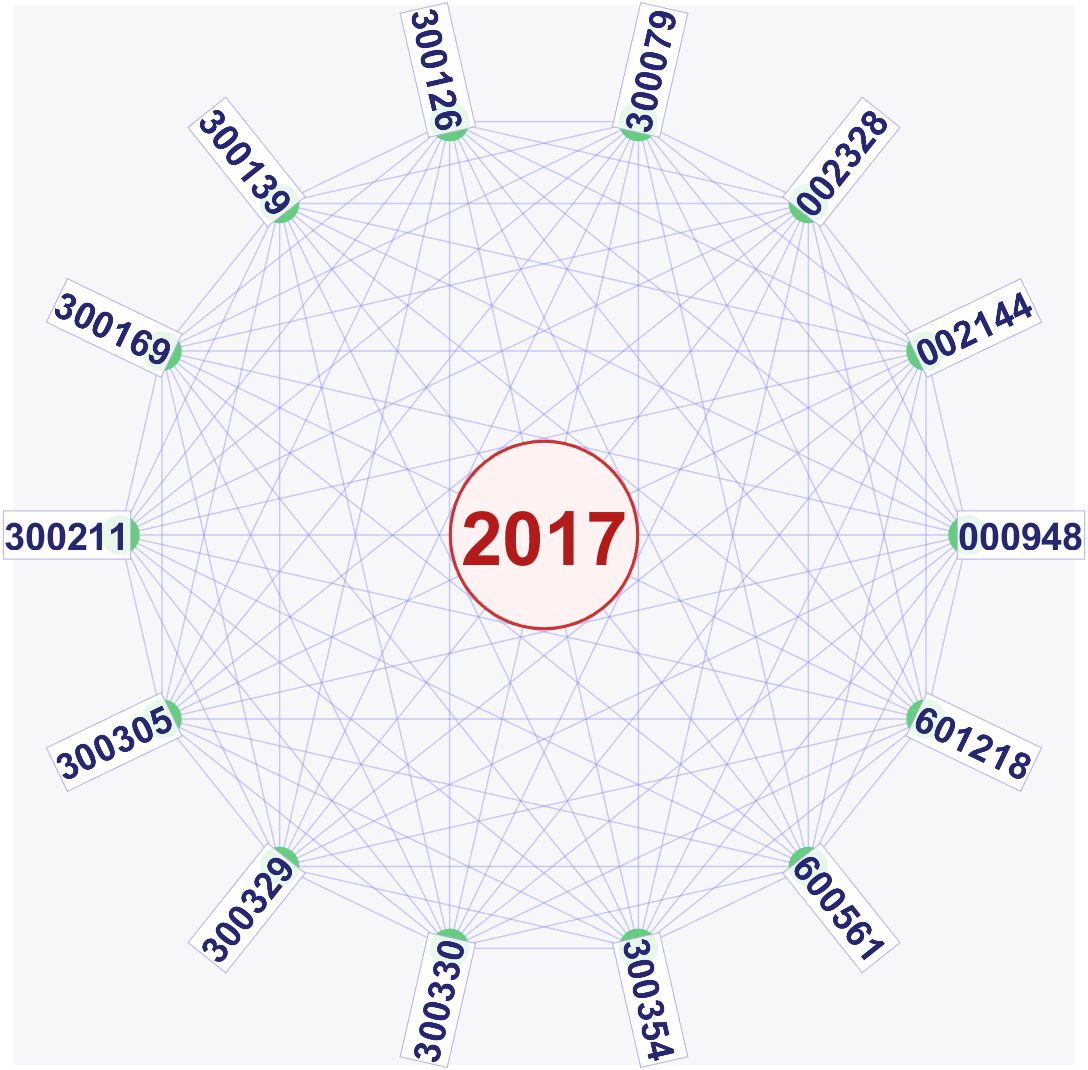}}
{\includegraphics[width=0.3\textwidth]{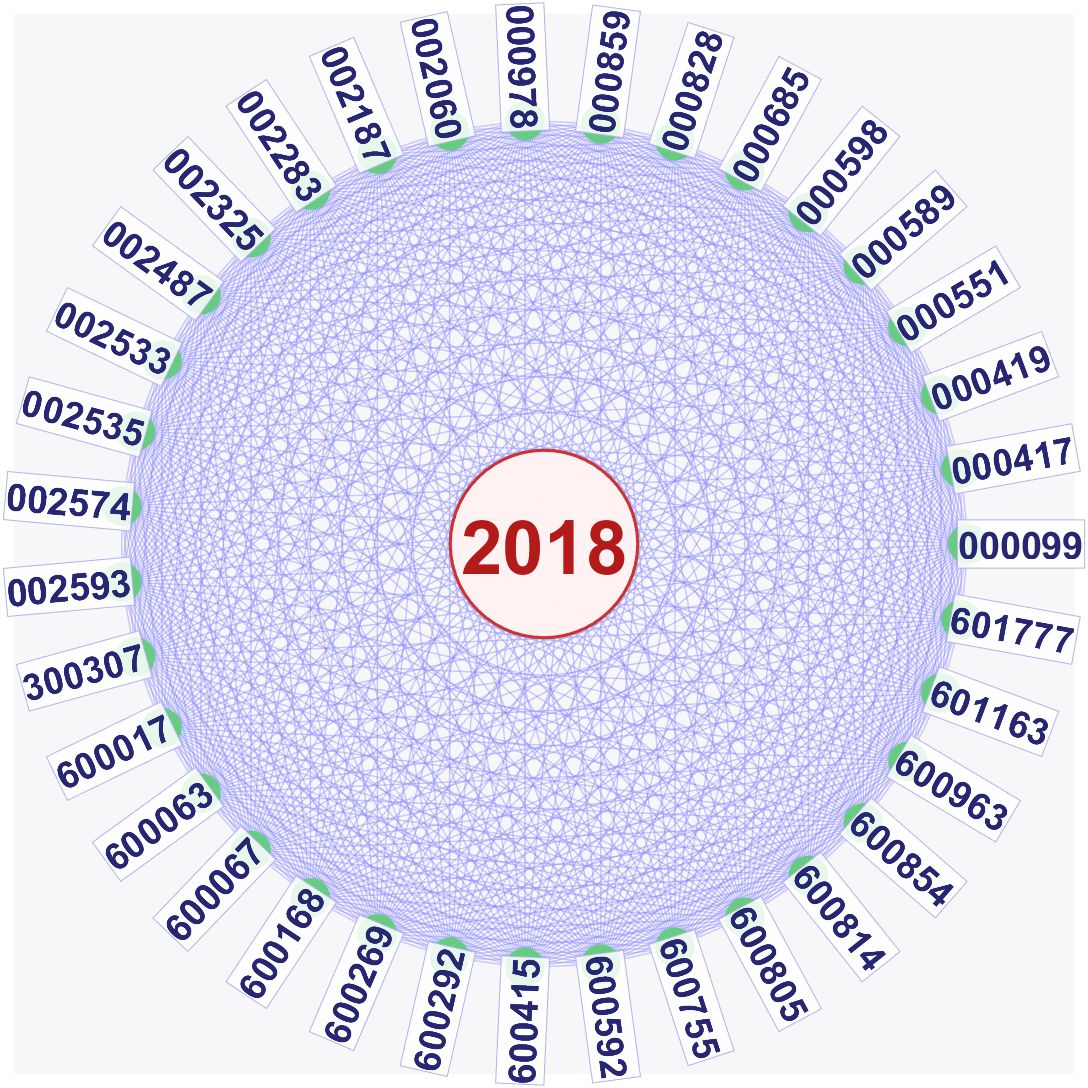}}
}
\resizebox{\columnwidth}{!}{
{\includegraphics[width=0.3\textwidth]{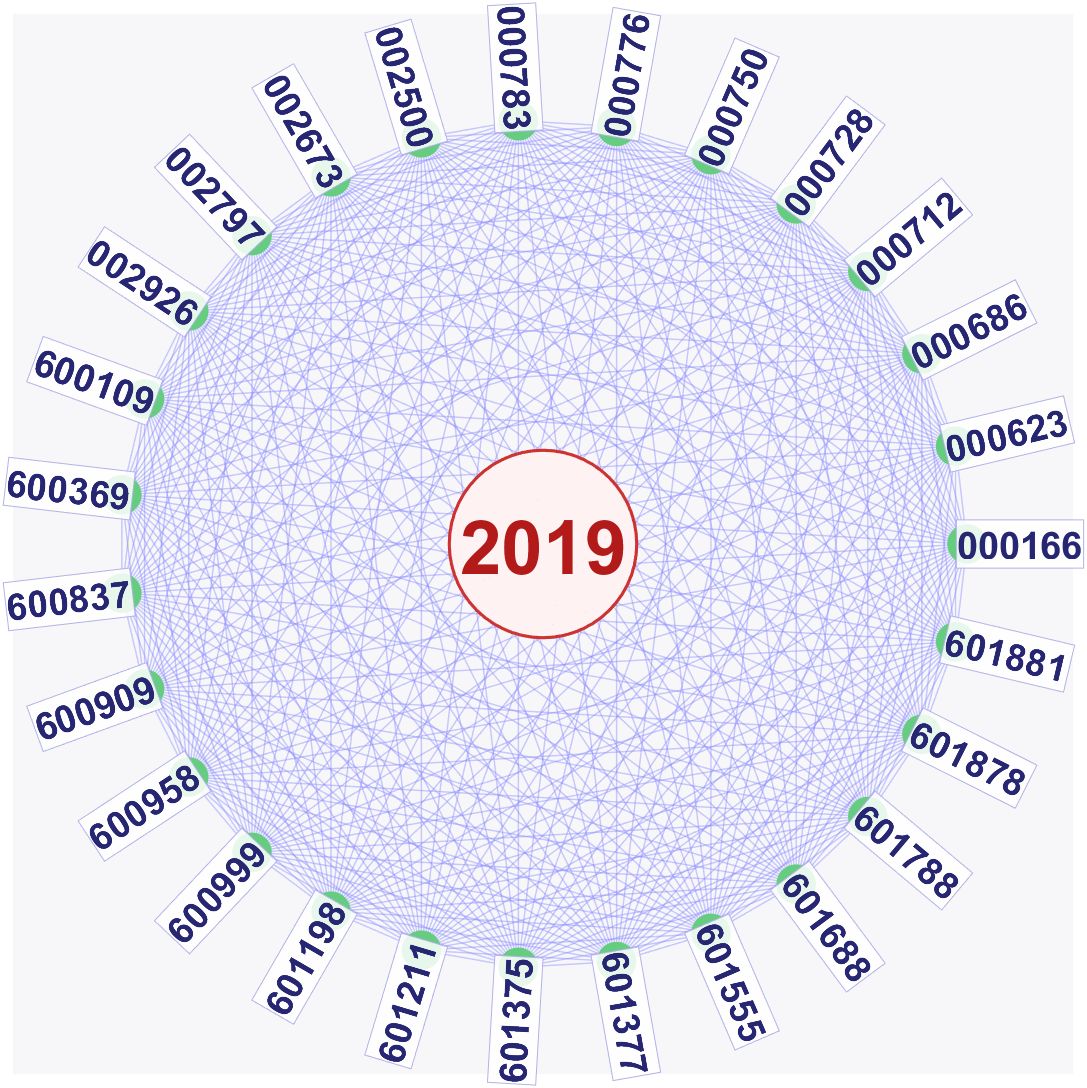}}
{\includegraphics[width=0.3\textwidth]{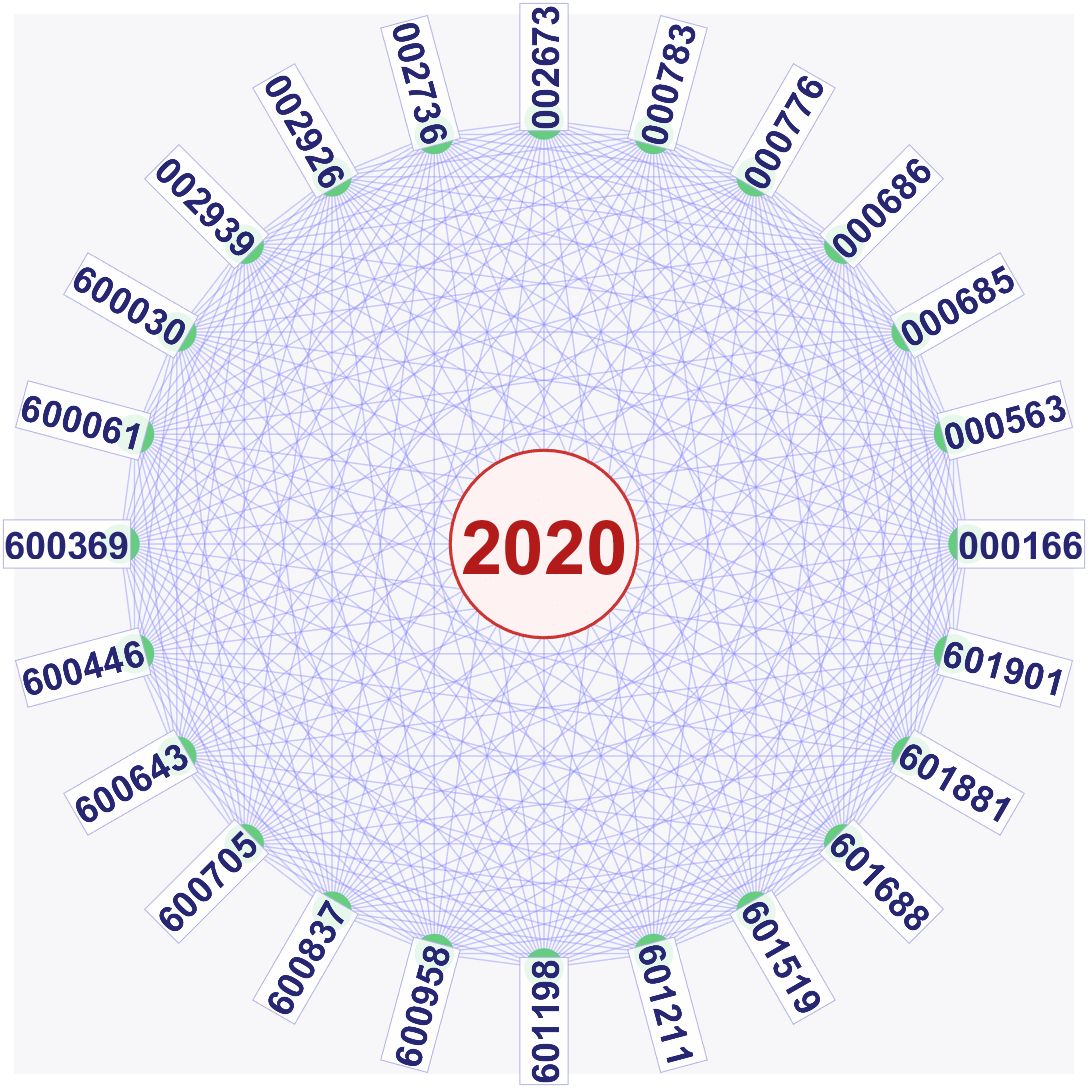}}
{\includegraphics[width=0.3\textwidth]{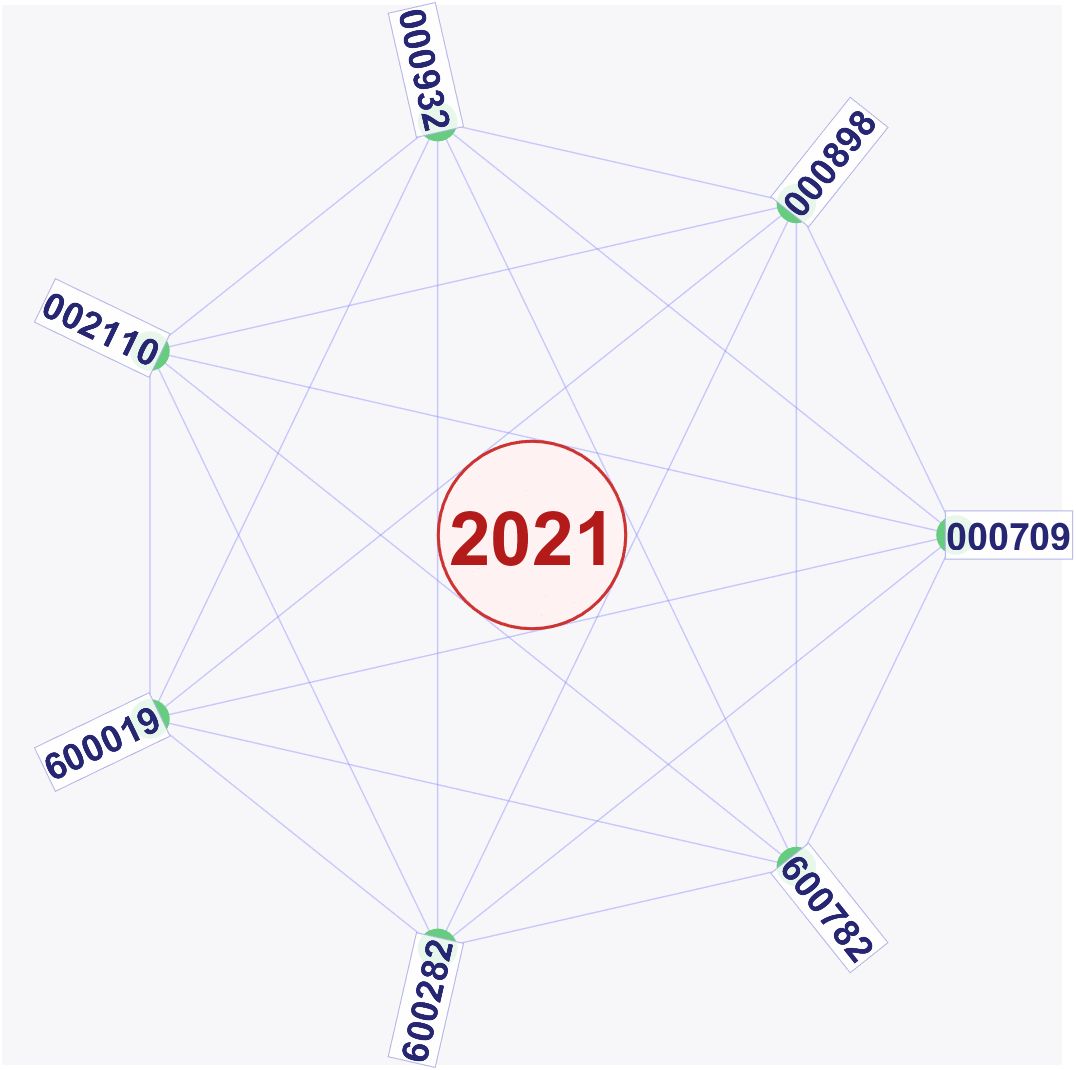}}
{\includegraphics[width=0.3\textwidth]{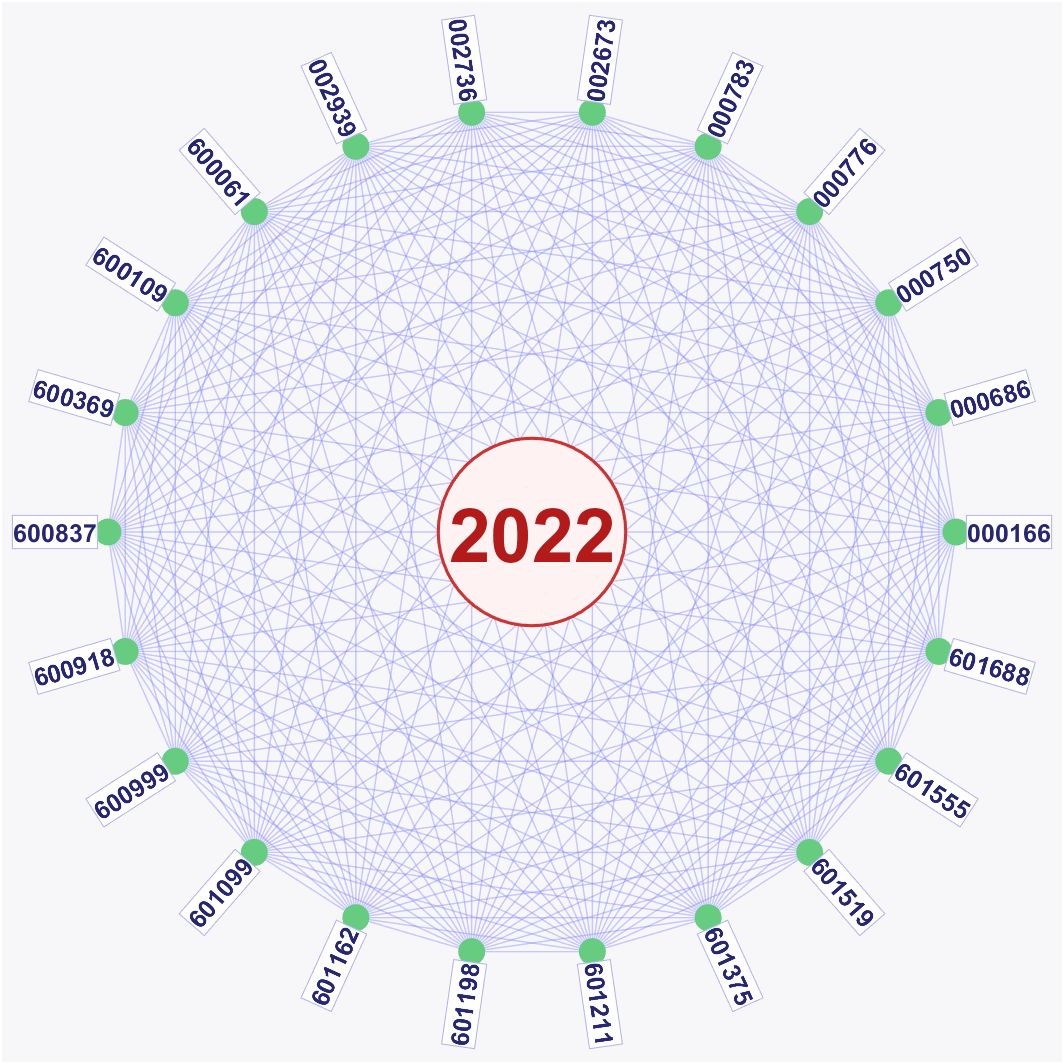}}
{\includegraphics[width=0.3\textwidth]{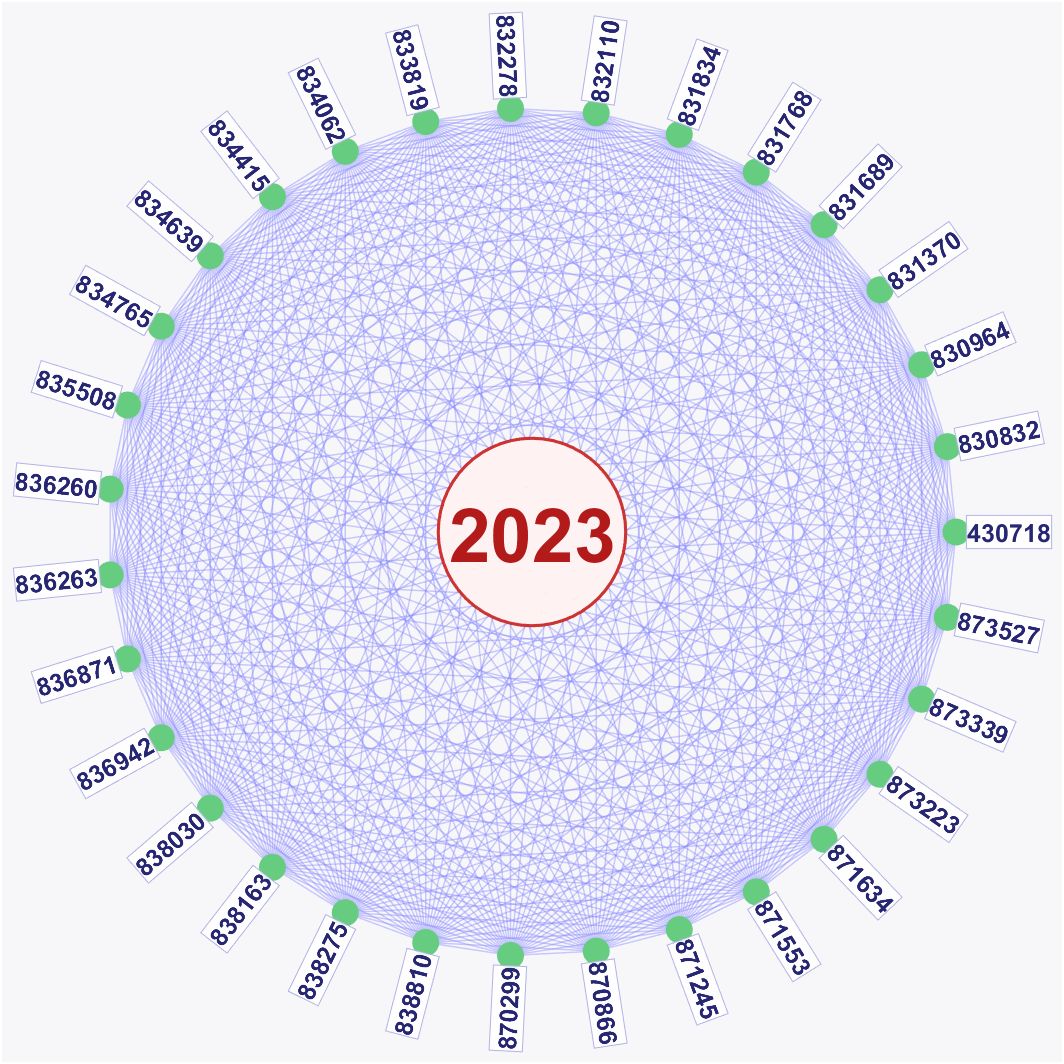}}
{\includegraphics[width=0.3\textwidth]{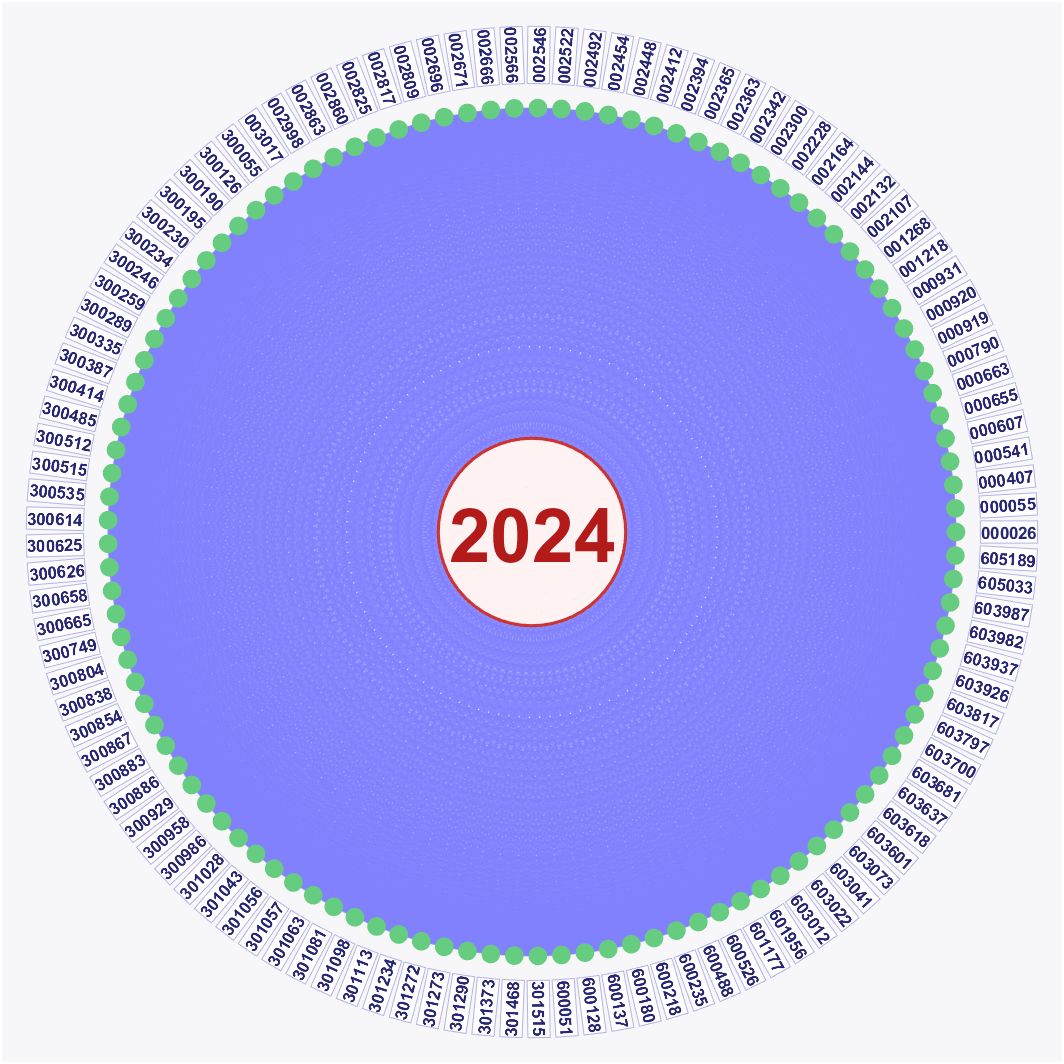}}
}
\caption{Correlation networks of stocks in LSCBMs for the twelve consecutive years 2013-2024, where we omit edge weights for visual clarity.}
\label{stockLSCBM} 
\end{figure}
In Figure \ref{stockLSCBM}, we plot the stocks within the LSCBM for each year from 2013 to 2024. We observe that a clear pattern emerges: the composition of these core modules shows almost no stability across consecutive years. For example, there is no common stock shared between the 2024 and 2023 modules, nor between 2023 and 2022, and the same disconnect holds for 2022 and 2021. This year-to-year turnover highlights how the LSCBM captures shifting market dynamics—during high-stress periods like the 2015 crash, a large, tightly coupled module forms as stocks move in lockstep, but in calmer or fragmented times (e.g., 2021), the module shrinks and reconstitutes around different stocks, reflecting new sectoral influences or risk factors. Essentially, the lack of overlap underscores that the ``core" of the market isn't fixed; it dynamically reorganizes annually, driven by evolving economic conditions and crises, which the LSCBM framework effectively reveals.

\begin{table}[!htbp]
\centering
\caption{Industry Distribution of Stocks in LSCBM for the twelve consecutive years 2013-2024.}
\label{tab:industry}
\scriptsize
\begin{tabularx}{\textwidth}{@{}l c X@{}}
\toprule
Year & Dominant Industry (GICS Sub-Industry) & Representative Stocks \\ 
\midrule
2013 & Energy (Coal \& Consumable Fuels) & 
600123: Shanxi Lanhua Sci-Tech Venture Co., Ltd.\newline
601001: Jinneng Holding Shanxi Coal Industry Co., Ltd.\newline
600188: Yanzhou Coal Mining Co., Ltd. \\ 
\addlinespace

2014 & Energy (Coal \& Consumable Fuels) & 
600348: Shanxi Huayang Group New Energy Co., Ltd.\newline
600546: Shanxi Coal International Energy Group Co., Ltd.\newline
600997: Kailuan Energy Chemical Co., Ltd. \\ 
\addlinespace

2015 & Industrials (Industrial Machinery)& 
600166: Beiqi Foton Motor Co., Ltd.\newline
600192: Great Wall Electrical Co., Ltd.\newline
601798: Harbin Electric Co., Ltd. \\ 
\addlinespace

2016 & Information Technology (Application Software) & 
300074: Huatest Testing Technology Co., Ltd.\newline
300442: Runhe Software Development Co., Ltd.\newline
300415: Yizumi Holdings Co., Ltd.\\ 
\addlinespace

2017 & Information Technology (Internet Services \& Infrastructure)& 
300079: Beijing Sumavision Technologies Co., Ltd.\newline
300354: DongHua Testing Technology Co. Ltd.\newline
000948: Yunnan Nantian Electronics Information Co., Ltd. \\ 
\addlinespace

2018 & Industrials (Industrial Machinery) & 
300307: Ningbo Cixing Co., Ltd.\newline
600592: Fujian Longxi Bearing (Group) Co., Ltd.\newline
601777: Chongqing Qianli Technology Co., Ltd.\\ 
\addlinespace

2019 & Financials (Investment Banking \& Brokerage) & 
601688: Huatai Securities Co., Ltd.\newline
600958: Orient Securities Co., Ltd.\newline
000166: Shenwan Hongyuan Group Co., Ltd. \\ 
\addlinespace

2020 & Financials (Investment Banking \& Brokerage) & 
600837: Haitong Securities Co., Ltd.\newline
601211: Guotai Junan Securities Co., Ltd.\newline
601519: Shanghai DZH Ltd. \\ 
\addlinespace

2021 & Materials (Steel) & 
600019: Baoshan Iron \& Steel Co., Ltd.\newline
000709: HBIS Co., Ltd.\newline
600782: Xinyu Iron \& Steel Co., Ltd. \\ 
\addlinespace

2022 & Financials (Investment Banking \& Brokerage) & 
601198: Dongxing Securities Co., Ltd.\newline
601375: Zhongyuan Securities Co., Ltd.\newline
600061: SDIC Capital Co., Ltd.\\ 
\addlinespace

2023 & Information Technology (Application Software) & 
830964: Aisino Corporation\newline
831370: Newange Ambient Intelligence Technical Service Co.Ltd \newline
834062: Kerun Control Engineering Co., Ltd. \\ 
\addlinespace

2024 & Industrials (Building Products \& Industrial Machinery)& 
301028: Sinoma Science \& Technology Co., Ltd.\newline
301113: Zhejiang Yayi Metal Technology Co., Ltd.\newline
300126: KEN Holding Co., Ltd.\\ 
\bottomrule
\end{tabularx}
\end{table}
In Table \ref{tab:industry}, we highlight the dominant industries within the LSCBM for the Chinese stock market from 2013 to 2024, illustrating how these core clusters adapt to key economic sectors in response to major market events. For instance, the LSCBM was dominated by industrial sectors during the high-stress 2015 market crash, transitioned to financials amid pandemic-driven uncertainty in 2020, and shifted back to industrials by 2024, likely reflecting policy-driven adjustments. This annual rotation across energy, information technology, materials, and financials underscores the LSCBM framework’s ability to capture evolving market themes. By capturing these patterns in real-world dynamics, LSCBM provides a robust lens for identifying economically meaningful and structurally stable subsystems within the market.
\begin{figure}[H]
\centering
\resizebox{\columnwidth}{!}{
{\includegraphics[width=0.3\textwidth]{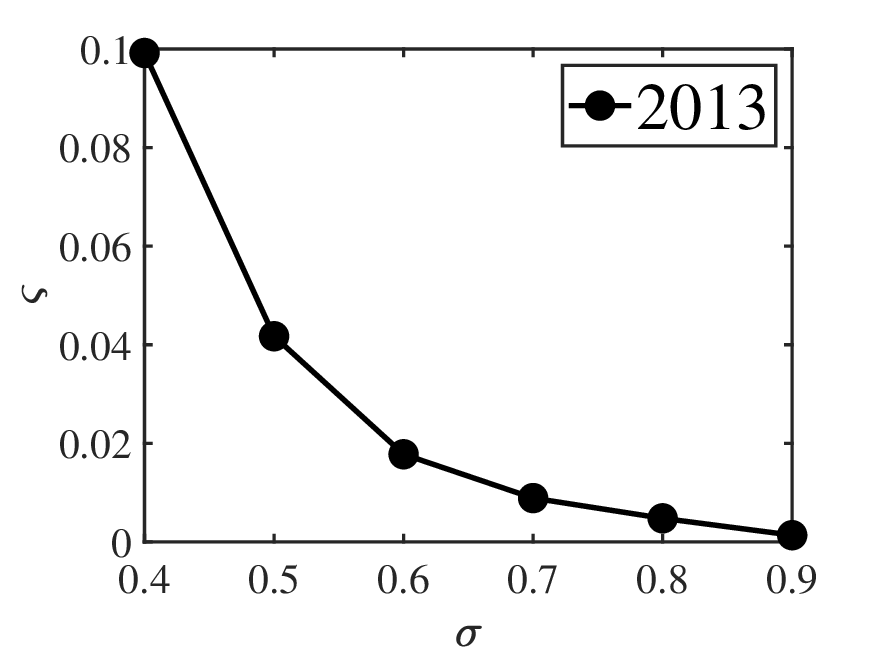}}
{\includegraphics[width=0.3\textwidth]{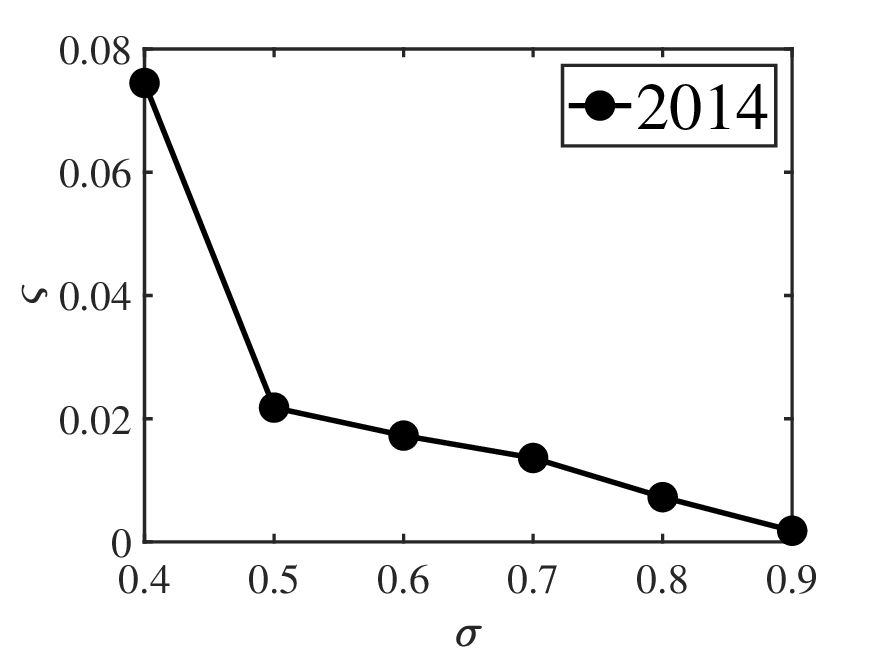}}
{\includegraphics[width=0.3\textwidth]{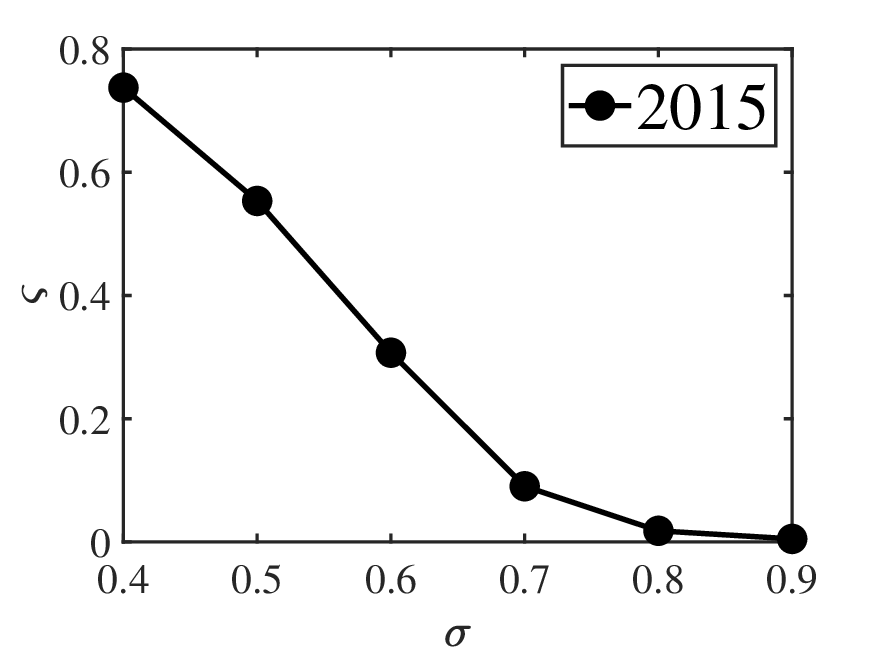}}
{\includegraphics[width=0.3\textwidth]{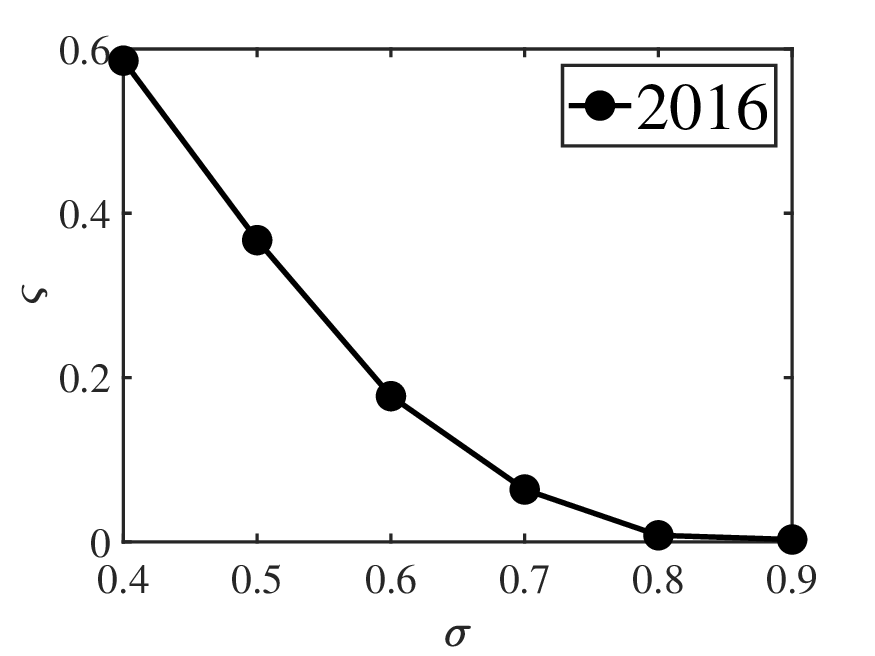}}
{\includegraphics[width=0.3\textwidth]{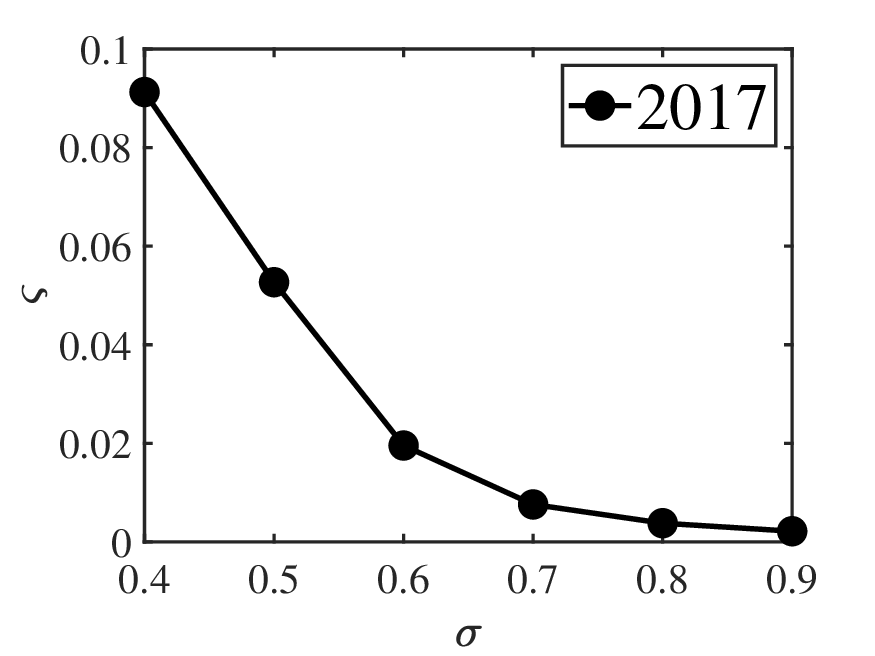}}
{\includegraphics[width=0.3\textwidth]{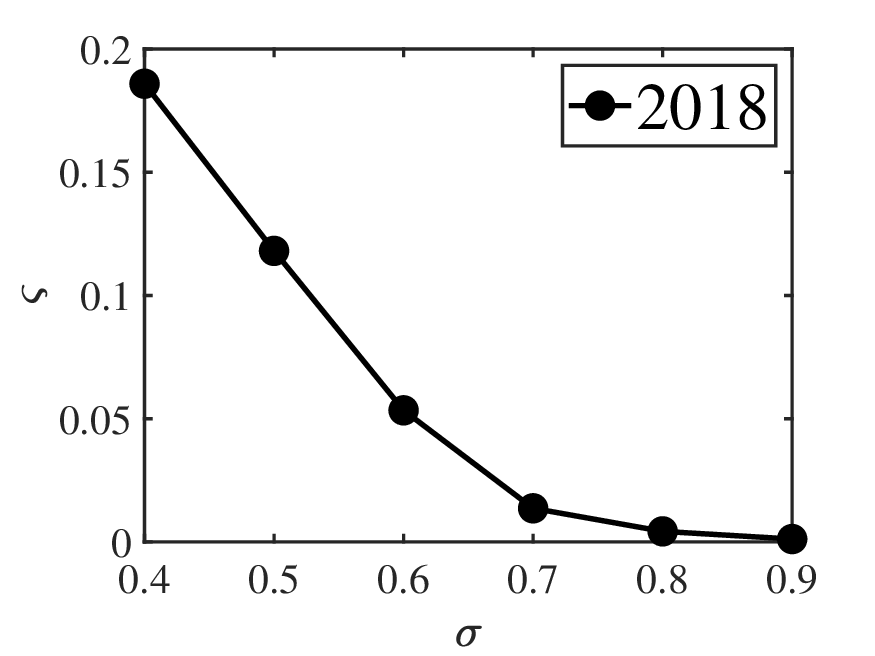}}
}
\resizebox{\columnwidth}{!}{
{\includegraphics[width=0.3\textwidth]{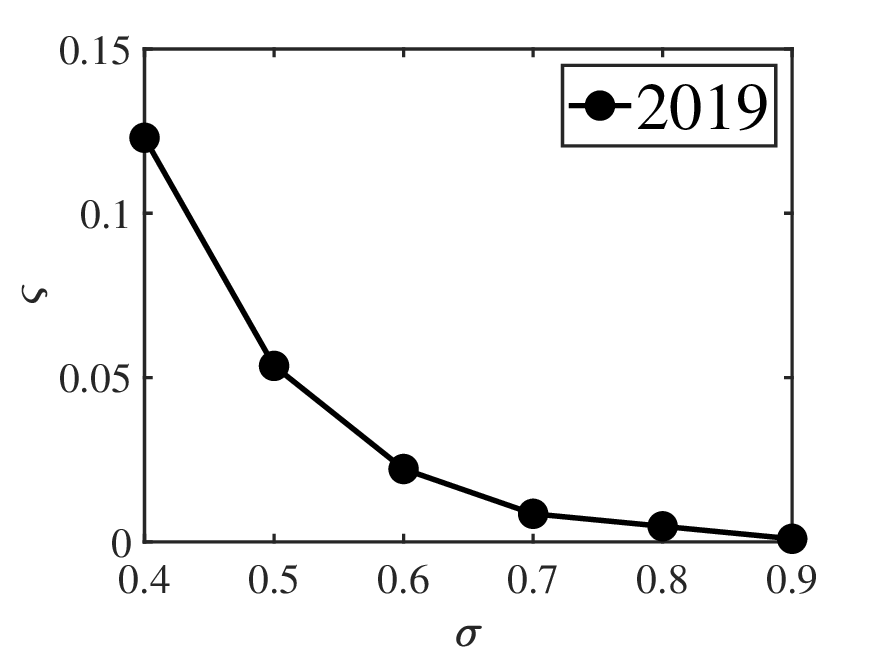}}
{\includegraphics[width=0.3\textwidth]{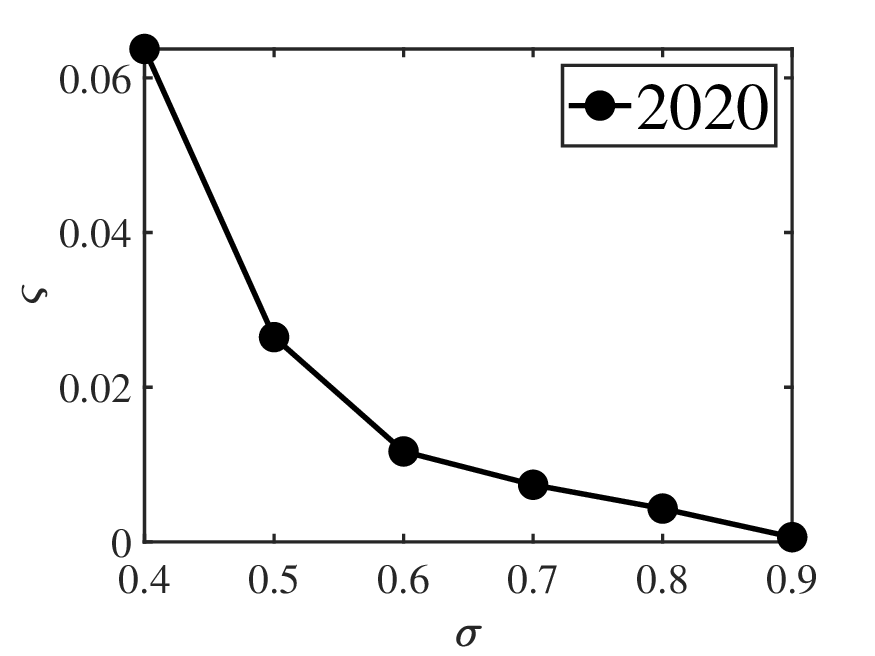}}
{\includegraphics[width=0.3\textwidth]{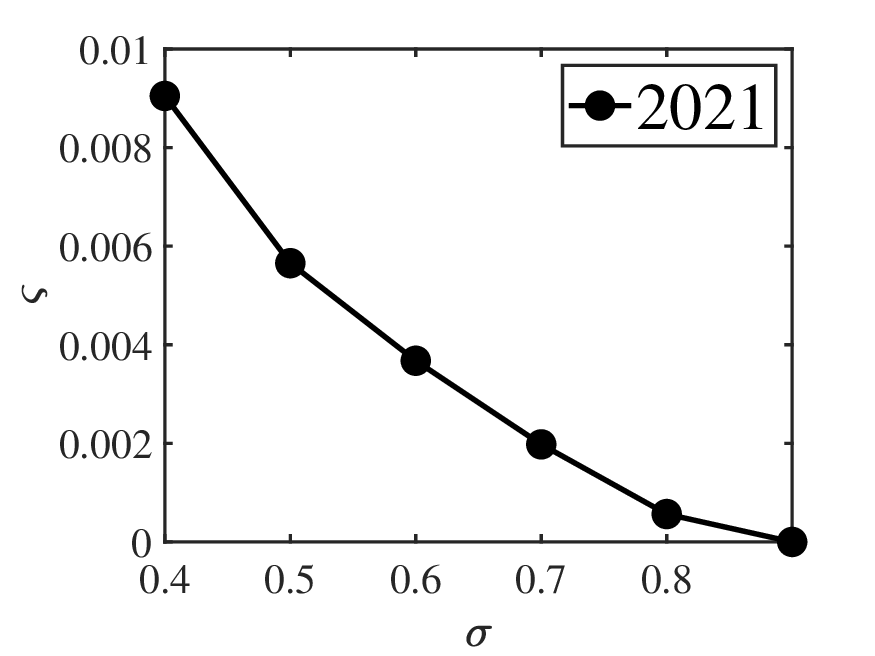}}
{\includegraphics[width=0.3\textwidth]{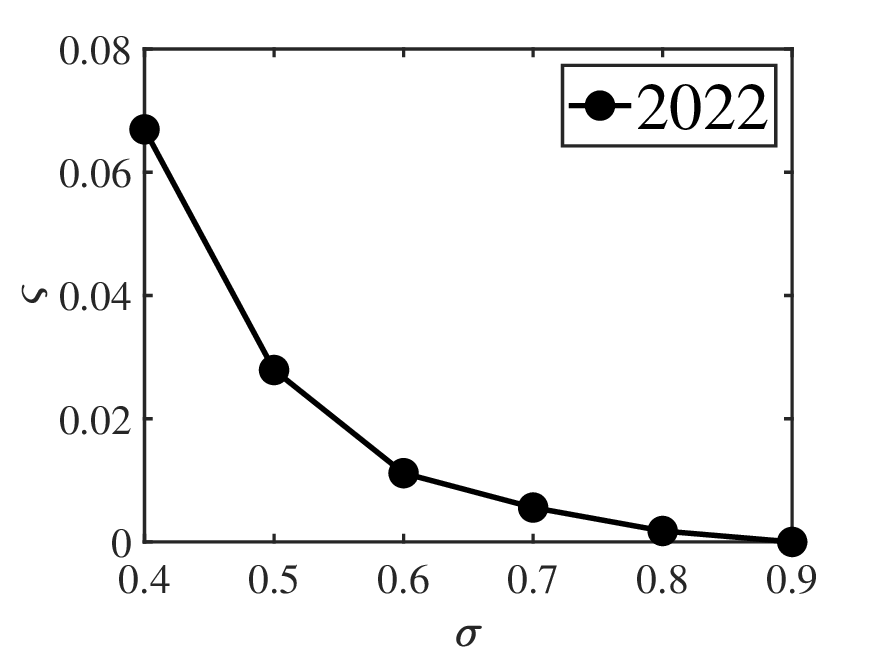}}
{\includegraphics[width=0.3\textwidth]{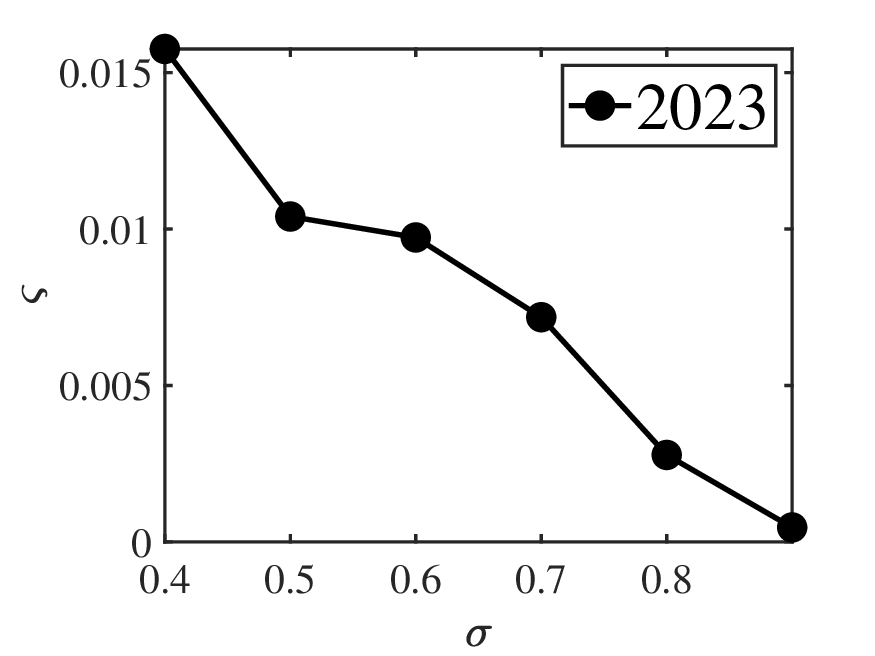}}
{\includegraphics[width=0.3\textwidth]{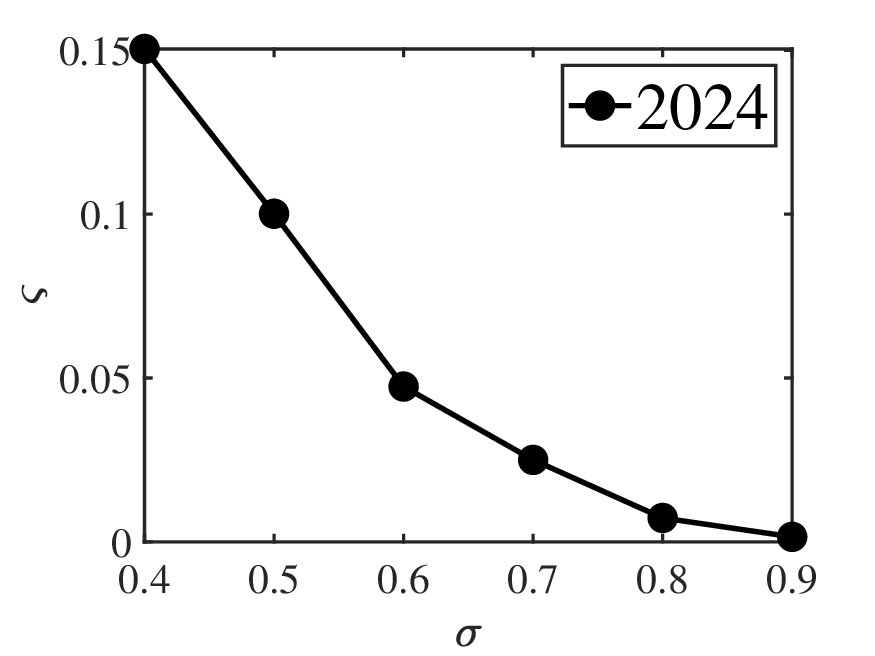}}
}
\caption{Proportion of nodes belonging to LSCBM against $\sigma$ for the twelve annual statistically validated stock networks.}
\label{vsigma} 
\end{figure}

Figure \ref{vsigma} presents the proportion of nodes belonging to LSCBM across the twelve annual statistically validated stock networks (2013-2024) as a function of the correlation strength threshold $\sigma$. Critically, all twelve curves exhibit a consistent, monotonically decreasing trend: $\varsigma$ decreases as $\sigma$ increases from 0.4 to 0.9. This universal pattern underscores the inherent trade-off embedded in LSCBM's definition: increasing $\sigma$ imposes a stricter requirement for pairwise correlation strength within the module, usually reducing its potential size. The curves diverge significantly over time, reflecting the time-varying market structural cohesion. Notably, the 2015 network demonstrates markedly higher $\varsigma$ values across nearly the entire $\sigma$ spectrum compared to other years, aligning with its identification in Table \ref{realdataBasic} as a period of extreme market synchronization (the 2015 stock crash). Conversely, the 2021 network consistently yields the lowest $\varsigma$, confirming its status as the most fragmented year. The sharp decline in $\varsigma$ observed as $\sigma$ exceeds approximately 0.75 across most years highlights the scarcity of very strong ($|\widetilde{\mathbf{C}}_{i,j}|\geq0.8$) statistically validated correlations that meet the structural balance condition within large, stable modules in the Chinese stock market. 
\section{Conclusion}\label{sec5}
This study presents a novel framework for identifying structurally stable core subsystems within financial markets by introducing the concept of the largest strong-correlation balanced module (LSCBM). We establish the LSCBM as the first rigorous integration of statistically validated correlation networks—which objectively filter out spurious relationships—with structural balance theory to uncover market segments characterized by both economically significant correlation strength and relational stability. The core theoretical contribution lies in formally defining the LSCBM and deriving its fundamental asymptotic properties within the random signed graph model, establishing its expected size and multiplicity across diverse network regimes. To enable practical application on large-scale financial networks, we develop the efficient MaxBalanceCore algorithm. Leveraging structural balance theory and network sparsity, MaxBalanceCore identifies LSCBM with quadratic time complexity, making it feasible for real-world stock markets comprising thousands of entities. Empirically, the LSCBM framework reveals that the core structure of the Chinese stock market is dominated by clusters of strongly positively correlated stocks. No instances of the theoretically possible ``enemy of my enemy" motif (balanced negative triangles) were found within any LSCBM across the twelve-year study period (2013-2024) in the Chinese stock market. This absence of significant negative correlations within these core, densely connected modules indicates a critical lack of inherent, statistically robust hedging opportunities within these specific market subsystems. The extreme scarcity of validated negative correlations ($\xi_{-}\leq0.75\%$) in the broader networks support this finding. Instead, these modules act as cohesive risk units reflecting sectoral themes (e.g., Industrials during crises, Financials during pandemics), where stocks move in lockstep, particularly during high-stress events like the 2015 crash. Critically, LSCBMs capture the market’s dynamic reorganization: their composition rotates annually across consecutive years, while their size expands dramatically during systemic crises (e.g., 2015) and contracts in fragmented regimes (e.g., 2021). This sensitivity to economic shifts positions LSCBMs as real-time indicators in China's stock market. For investors, the uniform positivity within LSCBMs implies concentrated exposure to systemic risks, necessitating diversification beyond LSCBMs rather than internal hedging. 

Several promising avenues for future research emerge. Firstly, the application scope of the LSCBM concept warrants exploration beyond financial markets, such as in biological systems (e.g., gene regulatory networks) and social network analysis. For example, in online social networks (e.g., Twitter, Weibo) where relationships can be positive (follow) or negative (block), LSCBM could identify large stable groups characterized by mutual trust or the “enemy of my enemy” principle. In gene regulatory networks, where edges represent activation or inhibition, an LSCBM would correspond to a stable core of genes whose interactions maintain consistent patterns—either all mutually activating or structured with cross-inhibition—relevant to understanding cellular stability. In neuroscience, functional brain networks derived from fMRI data often exhibit positive and negative correlations between regions; applying LSCBM could help identify a stable core of regions with consistent interaction patterns, potentially offering insights into brain organization. In power grids, where positive correlations indicate synchronized flow and negative correlations indicate opposing loads, an LSCBM might reveal a robust, self-stabilizing subset of the grid capable of autonomous operation during disruptions.  Secondly, extending the theoretical foundations is crucial, including investigating LSCBM properties in random signed network models with degree heterogeneity or community structure. Moreover,  generalizing the LSCBM definition to directed signed networks would require redefining structural balance for directed triads and analyzing the resulting module properties. Lastly, within statistically validated financial correlation networks, future work can focus on leveraging the network structure for enhanced community detection for stock markets.  More generally, the LSCBM framework is transferable to any domain where pairwise interactions can be represented as signed, weighted edges. Applying it to a new field simply requires constructing an appropriate signed adjacency matrix—where edge weights reflect interaction strength and signs reflect the nature of the relationship (e.g., trust/distrust, activation/inhibition, synchronization/opposition)—and optionally pre-filtering edges using domain-relevant statistical criteria. The MaxBalanceCore algorithm, being purely structural, then operates on this graph without modification, underscoring the framework's broad applicability.
\section*{CRediT authorship contribution statement}
\textbf{Huan Qing:} Conceptualization, Data curation, Formal analysis, Funding acquisition, Methodology, Software, Visualization, Writing – original draft, Writing - review $\&$ editing.

\textbf{Xiaofei Xu:} Methodology, Funding acquisition,Validation, Visualization, Writing – original draft, Writing – review $\&$ editing.
\section*{Declaration of competing interest}
The authors declare no competing interests.
\section*{Data availability}
Data will be made available on request. 
\section*{Acknowledgements}
Huan Qing was supported by the Scientific Research Foundation of Chongqing University of Technology (Grant No. 2024ZDR003), and the Science and Technology Research Program of Chongqing Municipal Education Commission (Grant No. KJQN202401168). Xiaofei Xu was supported by the  National Natural Science Foundation of China (NSFC) (Grant No. 12301358 and Grant No. 42450275).
\appendix
\section{Proofs of theoretical results}
\subsection{Proof of Lemma \ref{lemmain1}}
\begin{proof}
We prove that the probability of no SCBM existing vanishes asymptotically by analyzing the number of size-3 SCBMs. Let \(Z_3\) denote the number of strong-correlation balanced modules (SCBMs) of size exactly 3. Since the definition of an SCBM requires at least three nodes, the event \(\{\mathcal{S}^* = \emptyset\}\) implies no SCBM of any size exists, which includes size 3. Thus, we have
\[
\{\mathcal{S}^* = \emptyset\} \subseteq \{Z_3 = 0\}.
\]

Consequently, \(\mathbb{P}(\mathcal{S}^* = \emptyset) \leq \mathbb{P}(Z_3 = 0)\). Next, we show \(\mathbb{P}(Z_3 = 0) \to 0\) by using Chebyshev's inequality, which requires to bound the expectation and variance of \(Z_3\).

For any three distinct nodes, the probability of all edges existing with all positive signs is \(\alpha^3\). For the mixed case, there are \(\binom{3}{2} = 3\) choices for which two edges are negative, each with probability \(\alpha \beta^2\). Thus, the probability for a fixed triplet to be an SCBM is:
\[
p = \alpha^3 + 3\alpha\beta^2.
\]

Since \(\alpha > 0\) and \(\beta \geq0\), we have \(p > 0\). The number of triplets is \(\binom{N}{3}\), so we have
\[
\mathbb{E}[Z_3] = \binom{N}{3} p \sim \frac{N^3}{6} (\alpha^3 + 3\alpha\beta^2) = \Theta(N^3).
\]

In particular, \(\mathbb{E}[Z_3] \to \infty\) as \(N \to \infty\). Let \(I_U\) be the indicator that the node set \(U\) (with \(|U| = 3\)) is an SCBM. Then \(Z_3 = \sum_U I_U\), and the variance is
\[
\text{Var}(Z_3) = \sum_U \text{Var}(I_U) + \sum_{U \neq V} \text{Cov}(I_U, I_V).
\]

We bound each term separately. First, since \(I_U\) is a Bernoulli random variable, \(\text{Var}(I_U) \leq \mathbb{E}[I_U]\). Thus, we get
\[
\sum_U \text{Var}(I_U) \leq \sum_U \mathbb{E}[I_U]= \mathbb{E}[Z_3]= \Theta(N^3).
\]

For the covariance terms, partition the sum based on the intersection size \(|U \cap V|\) has the following cases:
\begin{itemize}
  \item Case \(|U \cap V| = 0\): The edge sets are disjoint and independent, so \(\text{Cov}(I_U, I_V) = 0\).
  \item Case \(|U \cap V| = 1\): Suppose \(U = \{a,b,c\}\) and \(V = \{a,d,e\}\) share only node \(a\). The edges within \(U\) (\(ab, ac, bc\)) and within \(V\) (\(ad, ae, de\)) are disjoint and independent. Hence, \(\text{Cov}(I_U, I_V) = 0\).
  \item Case \(|U \cap V| = 2\): Suppose \(U = \{a,b,c\}\) and \(V = \{a,b,d\}\) share nodes \(a\) and \(b\). The edge \(ab\) is shared, while edges \(\{ac, bc\}\) and \(\{ad, bd\}\) are disjoint. The covariance is bounded by \(|\text{Cov}(I_U, I_V)| \leq 2\) (since \(|I_U I_V| \leq 1\) and \(|\mathbb{E}[I_U] \mathbb{E}[I_V]| \leq 1\)). The number of such unordered pairs \((U,V)\) is:
  \[
  \binom{N}{3} \cdot \binom{3}{2} \cdot (N-3) = \Theta(N^4),
  \]
  since we choose \(U\) (\(\binom{N}{3}\) ways), choose two nodes in \(U\) to be shared with \(V\) (\(\binom{3}{2}\) ways), and choose the third node of \(V\) from the remaining \(N-3\) nodes.
Summing over all cases, we have:
\[
\sum_{U \neq V} \text{Cov}(I_U, I_V) \leq 0 + 0 + 2 \cdot \Theta(N^4) = \Theta(N^4).
\]
\end{itemize}

Combining both parts of the variance gives
\[
\text{Var}(Z_3) \leq \Theta(N^3) + \Theta(N^4) = \Theta(N^4).
\]

By Chebyshev's inequality:
\[
\mathbb{P}(Z_3 = 0) \leq \mathbb{P}\left( |Z_3 - \mathbb{E}[Z_3]| \geq \mathbb{E}[Z_3] \right) \leq \frac{\text{Var}(Z_3)}{\mathbb{E}[Z_3]^2} \leq \frac{C N^4}{(c N^3)^2} = \frac{C}{c^2} N^{-2},
\]
where \(C, c > 0\) are constants depending on \(\alpha\) and \(\beta\). As \(N \to \infty\), the right side vanishes:
\[
\mathbb{P}(Z_3 = 0) \to 0.
\]

Finally, since \(\mathbb{P}(\mathcal{S}^* = \emptyset) \leq \mathbb{P}(Z_3 = 0)\), we conclude:
\[
\lim_{N \to \infty} \mathbb{P}(\mathcal{S}^* = \emptyset) = 0.
\]
\end{proof}
\subsection{Proof of Theorem \ref{main1}}
\begin{proof}
For the first part of Theorem \ref{main1}, let \(Z_s\) be the number of SCBMs of size \(s\). Its expectation is:
\[
\mathbb{E}[Z_s] = \binom{N}{s} \sum_{k=0}^s \binom{s}{k} \alpha^{\binom{k}{2} + \binom{s-k}{2}} \beta^{k(s-k)}
\]
where \(k = |A|\), \(s - k = |B|\). The binomial coefficients arise from:
\begin{itemize}
  \item \(\binom{N}{s}\): ways to choose \(s\) nodes from \(N\)
  \item \(\binom{s}{k}\): ways to partition \(s\) nodes into subsets \(A\) with size \(k\) and \(B\) with size \(s-k\).
  \item \(\alpha^{\binom{k}{2} + \binom{s-k}{2}}\): probability all intra-subset edges exist and are positive.
  \item \(\beta^{k(s-k)}\): probability all \(A\)-\(B\) edges exist and are negative.
\end{itemize}

To analyze the asymptotics of \(\mathbb{E}[Z_s]\) as \(N \to \infty\), we apply Stirling's approximation:
\[
\binom{N}{s} \sim \frac{(eN)^s}{s^s \sqrt{2\pi s}}=\frac{e^{s \log N - s \log s + s}}{\sqrt{2 \pi s}}, \quad \binom{s}{k} \sim \frac{e^{s H(k/s)}}{\sqrt{2\pi (k/s)(1 - k/s)s}}=\frac{e^{s H(a)}}{\sqrt{2 \pi s a(1 - a)}},
\]
where \(H(a) = -a \log a - (1-a) \log(1-a)\) is the binary entropy function and we set \(k = as\) for \(a \in [0, 1]\). Here, the binary entropy function captures the combinatorial ``cost" of partitioning nodes into subsets of relative sizes \(a\) and \(1-a\).

The edge probability term \(\alpha^{\binom{k}{2} + \binom{s-k}{2}} \beta^{k(s-k)}\) is exponentiated as follows:
\[
\log \left( \alpha^{\binom{k}{2} + \binom{s-k}{2}} \beta^{k(s-k)} \right) = \underbrace{\left[ \binom{k}{2} + \binom{s-k}{2} \right] \log \alpha}_{\text{Intra-group edges}} + \underbrace{\left[ k(s-k) \log \beta \right]}_{\text{Inter-group edges}}.
\]

Given that \(k = a s\), for intra-group edges, we have
\[
\binom{k}{2} + \binom{s-k}{2} = \frac{k(k-1)}{2} + \frac{(s-k)(s-k-1)}{2} \approx\frac{s^2}{2} \left( a^2 + (1-a)^2 \right) + \mathcal{O}(s),
\]
where the lower-order term \(\mathcal{O}(s)\) vanishes asymptotically. For inter-group edges, we have
   \[
   k(s-k) = (a s)(s - a s) = a(1-a) s^2.
   \]

Thus, the exponent becomes:
\[
s^2 \underbrace{\left[ \frac{a^2 + (1-a)^2}{2} \log \alpha + a(1-a) \log \beta \right]}_{f(a)} + \mathcal{O}(s),
\]
where we define:
\[
f(a) = \frac{a^2 + (1-a)^2}{2} \log \alpha + a(1-a) \log \beta.
\]

Substituting approximations into \(\mathbb{E}[Z_s]\) gives
\[
\mathbb{E}[Z_s] \sim \frac{e^{s \log N - s \log s + s}}{\sqrt{2 \pi s}} \cdot \frac{e^{s H(a)}}{\sqrt{2 \pi s a(1-a)}} \cdot e^{s^2 f(a)}.
\]

Taking logarithms gives
\[
\log \mathbb{E}[Z_s] \approx \underbrace{s \log N - s \log s + s}_{(i)} + \underbrace{s H(a)}_{(ii)} + \underbrace{s^2 f(a)}_{(iii)} + o(s).
\]

We note that terms in \((i)\) (\(-s \log s + s\)) are independent of \(a\). They do not affect the optimal partition \(a^*\) and are absorbed into lower-order terms. Thus, the \(a\)-dependent dominant exponent in \(\log \mathbb{E}[Z_s]\) is:
\[
g(a) = s \log N + s H(a) + s^2 f(a).
\]

To find the most probable partition for SCBMs with size $s$, we maximize \(g(a)\) over \(a \in [0,1]\). Since \(s\log N\) is constant in \(a\), we solve:
   \[
   \max_{a} \left[ s H(a) + s^2 f(a) \right].
   \]

This identifies the partition \(a^*\) that maximizes the likelihood of SCBM formation. Maximizing \(g(a) = s \log N + s H(a) + s^2 f(a)\) over \(a \in [0,1]\) identifies the dominant contribution to the expected number \(\mathbb{E}[Z_s]\) of SCBMs of size \(s\). This maximization determines the most probable partition ratio \(a^* = |A|/s\) that maximizes the exponent in \(\mathbb{E}[Z_s]\), as \(g(a)\) captures the exponential growth rate (via \(s \log N\) and \(s H(a)\)) and edge probability decay (via \(s^2 f(a)\)). The scaling of \(\mathbb{E}[|\mathcal{S}^*|]\) emerges by finding the critical size \(s_c\) given later where \(\mathbb{E}[Z_s]\) transitions from decaying to growing exponentially, which occurs when the maximum of \(g(a)\) (over \(a\)) shifts sign; thus, maximizing \(g(a)\) directly governs the asymptotic behavior of \(\mathbb{E}[|\mathcal{S}^*|]\).

Because \(sH(a)\) is linear in \(s\), \(s^{2}f(a)\) is quadratic in \(s\), and \(H(a)\) is bounded because it ranges in \([0,\log 2]\), maximizing \(g(a)\) reduces to maximize \(f(a)\) for large \(N\). Next, we maximize \(f(a)\) to identify the most probable partition configuration. The derivative of \(f(a)\) is:
\[
f'(a) = (2a - 1)(\log \alpha - \log \beta)
\]
\begin{itemize}
  \item For the case \(\alpha \geq \beta\), we have
  \(f''(a) = 2(\log \alpha - \log \beta) \geq 0\), so \(f(a)\) is convex. The maximum occurs at endpoints:
  \[
  f(0) = f(1) = \frac{1}{2} \log \alpha \implies f^* = \frac{1}{2} \log \alpha
  \]
 \item For the case \(\alpha < \beta\), we have
  \(f''(a) = 2(\log \alpha - \log \beta) < 0\), so \(f(a)\) is concave. The maximum is at \(a = \frac{1}{2}\):
  \[
  f\left(\frac{1}{2}\right) = \frac{1}{4} \log(\alpha \beta) \implies f^* = \frac{1}{4} \log(\alpha \beta)
  \]
\end{itemize}

Since \(\alpha, \beta < 1\), we have \(f^* < 0\). After omitting the low-order term $sH(a)$, we have $g(a^{*})=s\log N+s^{2}f(a^{*})$, which gives $\mathbb{E}[Z_{s}]\approx \mathrm{exp}(g(a^{*}))=\mathrm{exp}(s\log N+s^{2}f(a^{*}))$. For the term $s\log N$, it originally arises for the dominant part of the binomial coefficient $\binom{N}{s}$ which counts the number of ways to choose $s$ nodes from $N$ nodes. Thus, we see that $s\log N$ quantifies the exponential growth in the number of candidate subsets of size $s$, where each subset is a potential module before edge constraints are applied. For the term $s^{2}f(a^{*})$, it originally comes from the joint probability that all edges in a candidate subset satisfy SCBM conditions under the optimal partition $a^{*}=\frac{|A|}{s}$. Because $f(a^{*})$ is negative, $s^{2}f(a^{*})$ represents the exponential decay in the probability that a subset with $s$ nodes forms a module satisfying structural balance condition. Based on the above analysis, we observe that
\begin{itemize}
  \item $s\log N$ increases $\mathbb{E}[Z_{s}]$ by adding more candidate subsets.
  \item $s^{2}f(a^{*})$ decreases $\mathbb{E}[Z_{s}]$ because more edges (growing as $s^{2}$) must satisfy constraints, and each edge has a probability $<1$.
\end{itemize}

Thus, the competition between the two terms $s\log N$ and  $s^{2}f(a^{*})$ determines the phase transition behavior of \(\mathbb{E}[Z_s]\):
\[
\log \mathbb{E}[Z_s]\approx \underbrace{s\log N}_{\text{Combinatorial growth}} + \underbrace{s^2 f(a^*)}_{\text{Probabilistic decay}}.
\]

Given that $s\log N$ grows linearly with $s$ and $s^2 f(a^*)$ decreases quadratically with $s$ since $f(a^{*})<0$, we have:
\begin{itemize}
  \item When \(s\) is small, the linear term $s\log N$ dominates for large $N$ because the quadratic term $|s^2 f(a^*)|$ is small in magnitude. This forces $\mathbb{E}[Z_s]$ to diverge, implying that SCBMs of size small $s$ are abundant in large networks.
  \item Conversely, when \(s\) is large, \(s^2 |f(a^*)|\) dominates. This forces \(\mathbb{E}[Z_s]\) to vanish exponentially, making SCBMs of large size $s$ statistically impossible in large networks.
\end{itemize}

Based on the above analysis, we argue that there must exist a sharp critical size $s_{c}$ that balances the combinatorial growth against edge probability decay. To find \(s_c\), we substitute \(s = c \log N\) and solve the balance equation:
  \[
  c \log N \cdot \log N + c^2 (\log N)^2 f(a^*) = 0 \implies c = -\frac{1}{f(a^*)} = \frac{1}{\lambda(\alpha,\beta)},
  \]
where
\[
\lambda(\alpha, \beta) = -f^* = \begin{cases}
\frac{1}{2} |\log \alpha| & \alpha \geq \beta,\\
\frac{1}{4} (|\log \alpha| + |\log \beta|) & \alpha < \beta.
\end{cases}
\]

Thus, we obtain the threshold size:
\[
s_c = \frac{\log N}{\lambda(\alpha, \beta)}
\]

In fact, $s_{c}$ is the asymptotic scaling of the size of the LSCBM, i.e., $|\mathcal{S}^{*}|$ must concentrate near $s_{c}$. To prove this statement, for any $0<\epsilon<1$, we want to show that
\begin{align*}
\mathbb{P}(|\mathcal{S}^{*}|\in[(1-\epsilon)s_{c},(1+\epsilon)s_{c}])\rightarrow1\qquad\mathrm{as}\qquad N\rightarrow\infty,
\end{align*}
which requires proving the following two distinct behaviors:
\begin{itemize}
  \item Below \(s_c\), SCBMs emerge abundantly.
  \item Above \(s_c\), their probability vanishes  exponentially.
\end{itemize}

For the case that $s$ is smaller than $s_{c}$, we set \(s= (1-\epsilon)s_c\) for any fixed \(0<\epsilon <1\). The exponent is
\[
s\log N + s^{2}f^* = (\log N)^2 \left[ \frac{1-\epsilon}{\lambda} + \frac{(1-\epsilon)^2}{\lambda^2} (-\lambda) \right] = (\log N)^2 \frac{\epsilon - \epsilon^2}{\lambda} > 0,
\]
which gives \(\mathbb{E}[Z_{s}]\approx\mathrm{exp}(s\log N+s^{2}f(a^{*}))\to \infty\) as $N\rightarrow\infty$. By Theorem \ref{main4}, we know that \(\text{Var}(Z_{s}) = o(\mathbb{E}[Z_{s}]^2)\) as $N\rightarrow\infty$ when \(s= (1-\epsilon)s_c\). By Chebyshev’s inequality, we have
\[
\mathbb{P}(Z_{s}= 0) \leq \frac{\text{Var}(Z_{s})}{\mathbb{E}[Z_{s}]^2} \to 0 \qquad\text{as~}N\rightarrow\infty,
\]
which implies that \(\mathbb{P}(Z_{s}> 0) \to 1\) as $N\rightarrow\infty$. Hence, SCBMs of size \(s= (1-\epsilon)s_c\) exist with high probability.

Set \(s = (1+\epsilon) s_c\). The exponent is:
\[
s \log N + s^2 f^* = (\log N)^2 \frac{-\epsilon - \epsilon^2}{\lambda} < 0
\]

Thus \(\mathbb{E}[Z_s]\to 0\) as $N\rightarrow\infty$. By Markov's inequality, we have
\[
\mathbb{P}(Z_s > 0)=\mathbb{P}(Z_s \geq1) \leq \mathbb{E}[Z_s]\to 0\qquad\text{as~}N\rightarrow\infty.
\]

So no SCBMs of size \(> (1+\epsilon)s_c\) exist with high probability. Thus, for any \(\epsilon \in(0,1)\):
 \[
  \lim_{N \to \infty} \mathbb{P}\left( \left| \frac{|\mathcal{S}^*|}{s_c} - 1 \right| < \epsilon \right) = 1.
  \]

Thus \(|\mathcal{S}^*| \sim s_c\) in probability, and
\[
\mathbb{E}[|\mathcal{S}^*|] \sim \frac{\log N}{\lambda(\alpha, \beta)}
\]

For the second part of Theorem \ref{main1}, fix \(\epsilon \in (0, 1/3)\) (e.g., \(\epsilon = 1/4\)). By previous analysis, for any \(\delta \in(0,1)\), there exists \(N_0\) such that for all \(N > N_0\),
\[
\mathbb{P}\left( |\mathcal{S}^{*}| \in \left[(1-\epsilon)s_c, (1+\epsilon)s_c\right] \right) > 1 - \delta,
\]
where \(s_c = \log N / \lambda(\alpha,\beta)\) is the asymptotic scaling of the LSCBM size.

Define \(Q\) as the number of unordered pairs \(\{A, B\}\) of vertex-disjoint SCBMs each of size \(s = \lfloor (1-\epsilon)s_c \rfloor\). The expectation of \(Q\) is
\[
\mathbb{E}[Q] = \frac{1}{2} \binom{N}{s} \binom{N-s}{s} \mu_s^2.
\]

From the proof of the first part of Theorem \ref{main1}, we know that \(\mathbb{E}[Z_s] = \binom{N}{s} \mu_s = \exp\left( s \ln N + s^2 f(a^*) + o(s) \right)\) with \(f(a^*) = -\lambda(\alpha,\beta)\). Substituting \(s = (1-\epsilon)s_c\) gives
\[
\ln \mathbb{E}[Z_s] = \frac{(\epsilon - \epsilon^2)(\ln N)^2}{\lambda} + o((\ln N)^2) \to \infty,
\]
implying \(\mathbb{E}[Z_s] \to \infty\). Rewrite \(\mathbb{E}[Q]\) as
\[
\mathbb{E}[Q] = \frac{1}{2} \mathbb{E}[Z_s] \cdot \binom{N-s}{s} \mu_s.
\]

Using the combinatorial identity \(\binom{N-s}{s} / \binom{N}{s} = \frac{(N-s)!^2}{(N-2s)! N!}\) and Stirling’s approximation, we have
\[
\binom{N-s}{s} \big/ \binom{N}{s} \leq \left(1 - \frac{s}{N}\right)^s \leq e^{-s^2 / N}.
\]

Since \(s = \Theta(\log N)\), \(s^2 / N \to 0\), so \(e^{-s^2 / N} \to 1\). Thus,
\[
\mathbb{E}[Q] \leq \frac{1}{2} \mathbb{E}[Z_s]^2 (1 + o(1)).
\]

Given that \(\mathbb{E}[Z_s] \to \infty\), it follows that \(\mathbb{E}[Q] \to \infty\). To bound \(\operatorname{Var}(Q)\), define the indicator \(I_{A,B} = \mathbf{1}_{\{A \text{ and } B \text{ are disjoint SCBMs}\}}\), so \(Q = \sum_{\{A,B\} \text{ disjoint}} I_{A,B}\). The variance decomposes as
\[
\operatorname{Var}(Q) = \sum_{\{A,B\}} \operatorname{Var}(I_{A,B}) + \sum_{\substack{\{A,B\} \neq \{C,D\} \\ \text{disjoint}}} \operatorname{Cov}(I_{A,B}, I_{C,D}),
\]
where the first term satisfies \(\sum_{\{A,B\}} \operatorname{Var}(I_{A,B}) \leq \mathbb{E}[Q]\), since \(\operatorname{Var}(I_{A,B}) \leq \mathbb{E}[I_{A,B}]\). For the second term, if the vertex sets of \(\{A,B\}\) and \(\{C,D\}\) are disjoint, edge independence implies \(\operatorname{Cov}(I_{A,B}, I_{C,D}) = 0\). When vertex sets overlap, let \(t \geq 1\) be the size of the intersection. Applying techniques from Theorem \ref{main4}, we have
\begin{itemize}
  \item For \(t < 6\) (since SCBMs require at least 3 nodes), \(\operatorname{Cov}(I_{A,B}, I_{C,D}) \leq \mathbb{E}[I_{A,B}]\).
  \item For \(t \geq 6\), conditional expectation and structural balance constraints yield \(\operatorname{Cov}(I_{A,B}, I_{C,D}) \leq \mathbb{E}[I_{A,B}] \mathbb{E}[I_{C,D}] / \mu_t\).
\end{itemize}

Combinatorial counting over intersection sizes confirms that
\[
\sum_{\substack{\{A,B\} \neq \{C,D\} \\ \text{overlapping}}} |\operatorname{Cov}(I_{A,B}, I_{C,D})| \leq c \mathbb{E}[Q]^2 N^{-1}
\]
for some constant \(c > 0\). Since \(\mathbb{E}[Q] \to \infty\), combining these terms obtains
\[
\operatorname{Var}(Q) \leq \mathbb{E}[Q] + c \mathbb{E}[Q]^2 N^{-1} = o(\mathbb{E}[Q]^2).
\]

By Chebyshev’s inequality, we have
\[
\mathbb{P}\left( |Q - \mathbb{E}[Q]| \geq \tfrac{1}{2} \mathbb{E}[Q] \right) \leq \frac{\operatorname{Var}(Q)}{(\mathbb{E}[Q]/2)^2} \to 0,
\]
which implies \(\mathbb{P}(Q \geq \tfrac{1}{2} \mathbb{E}[Q]) \to 1\). As \(\mathbb{E}[Q] \to \infty\), this gives \(\mathbb{P}(Q \geq 1) \to 1\).

Each pair \(\{A,B\}\) counted in \(Q\) must belong to distinct LSCBMs. If they are in the same LSCBM \(\mathcal{S}^*\), then \(|\mathcal{S}^*| \geq 2s\). However,
\[
2s = 2(1-\epsilon)s_c > (1+\epsilon)s_c \geq |\mathcal{S}^*| \quad \text{with high probability},
\]
since \(\epsilon < 1/3\) implies \(2(1-\epsilon) > 1+\epsilon\). This contradicts the maximality of \(|\mathcal{S}^*|\). Therefore,
\[
\mathbb{P}(Z_{|\mathcal{S}^{*}|} \geq 2) \geq \mathbb{P}(Q \geq 1) - \mathbb{P}\left(|\mathcal{S}^{*}| \notin [(1-\epsilon)s_c, (1+\epsilon)s_c]\right) \to 1 - \delta.
\]

As \(\delta > 0\) is arbitrary, \(\lim_{N \to \infty} \mathbb{P}(Z_{|\mathcal{S}^{*}|} \geq 2) = 1\).
\end{proof}
\subsection{Proof of Theorem \ref{main2}}
\begin{proof}
For the first part of this theorem, define \(d = \log b > 0\). We show that the expected number of non-all-positive balanced modules (containing at least one negative edge) vanishes asymptotically. W know that non-all-positive modules fall into the following two categories:
\begin{itemize}
  \item Mixed modules: \(1 \leq |A| \leq s-1\), \(|B| = s - |A| \geq 1\) (contain negative edges).
  \item All-negative modules: \(A = \emptyset\), \(|B| = s\) (fully negative edges but structurally balanced).
\end{itemize}

We analyze the upper bound for mixed modules. For fixed size \(s\) and partition \(a = |A|\) (\(1 \leq a \leq s-1\)), the expected count is:
\[
\mathbb{E}[\text{count}_{\text{mixed}}] \leq \sum_{a=1}^{s-1} \binom{N}{s} \binom{s}{a} \alpha^{\binom{a}{2} + \binom{s-a}{2}} \beta^{a(s-a)}.
\]

To bound the expected count of mixed modules, we apply tight asymptotic inequalities:
\begin{itemize}
  \item \(\binom{s}{a} \leq 2^s\) (since \(\sum_{a=0}^s \binom{s}{a} = 2^s\)),
  \item  \(\binom{N}{s} \leq (eN/s)^s\) (from Stirling's bound \(\binom{N}{s} \leq \frac{(eN)^s}{s^s}\)),
  \item \(\beta \leq 2b/N\) (for large \(N\), as \(\beta = b/N + o(1/N) \implies \beta \leq 2b/N\)),
  \item \(a(s-a) \geq s-1\) (minimized at \(a=1\) or \(a=s-1\) by convexity of \(a(s-a)\)).
\end{itemize}

Substituting these inequalities yields
\[
\mathbb{E}[\text{count}_{\text{mixed}}] \leq \sum_{a=1}^{s-1} \left(\frac{eN}{s}\right)^s 2^s \left(\frac{2b}{N}\right)^{s-1} = (s-1) \left(\frac{eN}{s}\right)^s 2^s \left(\frac{2b}{N}\right)^{s-1}.
\]

Setting \(s = \omega N\) (\(\omega > 0\) constant) and taking logarithms, we get
\[
\log \mathbb{E}[\text{count}_{\text{mixed}}] \leq \log(\omega N) + \omega N \log\left(\frac{e}{\omega}\right) + \omega N \log 2 + (\omega N - 1) \log\left(\frac{2b}{N}\right) + o(1).
\]

We see that the dominant term is \(\omega N \log\left(\frac{2b}{N}\right) = \omega N (\log(2b) - \log N)\). Since \(-\omega \log N \to -\infty\), we have
\[
\log \mathbb{E}[\text{count}_{\text{mixed}}] \leq -\Theta(N \log N) \implies \mathbb{E}[\text{count}_{\text{mixed}}] \leq e^{-\Theta(N \log N)} \to 0 \quad \text{(exponential decay)}.
\]

We now analyze the upper bound for all-negative modules. The expected count of size-\(s\) all-negative modules is
\[
\mathbb{E}[\text{count}_{\text{neg}}] = \binom{N}{s} \beta^{\binom{s}{2}}.
\]

Using \(\beta \leq \frac{2b}{N}\) obtains
\[
\mathbb{E}[\text{count}_{\text{neg}}] \leq \left(\frac{eN}{s}\right)^s \left(\frac{2b}{N}\right)^{\binom{s}{2}}.
\]

Setting \(s = \omega N\) and taking logarithms, we have
\[
\log \mathbb{E}[\text{count}_{\text{neg}}] \leq s \log\left(\frac{eN}{s}\right) + \binom{s}{2} \log\left(\frac{2b}{N}\right) = \omega N \log\left(\frac{e}{\omega}\right) + \frac{\omega N (\omega N - 1)}{2} \log\left(\frac{2b}{N}\right) + o(N).
\]

The dominant term is \(\frac{\omega^2 N^2}{2} \log\left(\frac{2b}{N}\right) = \frac{\omega^2 N^2}{2} (\log(2b) - \log N)\). Since \(-\frac{\omega^2 N^2}{2} \log N \to -\infty\), we get
\[
\mathbb{E}[\text{count}_{\text{neg}}] \to 0 \quad \text{(exponential decay)}.
\]

Therefore, we see that the total expected number of non-all-positive modules:
\[
\mathbb{E}[\text{count}_{\text{non-pos}}] = \mathbb{E}[\text{count}_{\text{mixed}}] + \mathbb{E}[\text{count}_{\text{neg}}] \to 0.
\]

Summing over \(s \geq 3\) gives
\[
\sum_{s=3}^N \mathbb{E}[\text{count}_{\text{non-pos}, s}] \leq O(N) e^{-\Theta(N \log N)} \to 0.
\]

Then, by Markov’s inequality, we have
\[
\mathbb{P}(\exists \text{non-all-positive balanced module}) \leq \sum_{s=3}^N \mathbb{E}[\text{count}_{\text{non-pos}, s}] \to 0.
\]

Therefore, w.h.p. LSCBM is an all-positive module (\(B = \emptyset\), automatically structurally balanced). Now, let \(Z_k\) denote the number of all-positive modules of size \(k\) (all edges present and positive), we get
\[
\mathbb{E}[Z_k] = \binom{N}{k} \alpha^{\binom{k}{2}}.
\]

Set \(k = c \mu = c \frac{N d}{b}\), where \(\mu = \frac{N \log b}{b}\), \(d = \log b\). We then provide asymptotic analysis of \(\mathbb{E}[Z_k]\). Using Stirling’s approximation gives
\[
\log \binom{N}{k} \leq k \log\left(\frac{N}{k}\right) + k.
\]

Given \(\alpha = 1 - \frac{b}{N} + o(1/N)\), by Taylor expansion, we have
\[
\log \alpha = -\frac{b}{N} - \frac{b^2}{2N^2} + o(1/N^2) \leq -\frac{b}{2N} \quad \text{(for large } N\text{)}.
\]

Then:
\[
\binom{k}{2} \log \alpha = \frac{k(k-1)}{2} \left(-\frac{b}{N} + O(1/N^2)\right) = -\frac{b k^2}{2N} + \frac{b k}{2N} + O(k^2/N^2).
\]

Substituting \(k = c \frac{N d}{b}\):
\[
\begin{aligned}
\log \mathbb{E}[Z_k] &= k \log\left(\frac{N}{k}\right) + k - \frac{b k^2}{2N} + O(\log N) \\
&= c \frac{N d}{b} \log\left(\frac{b}{c d}\right) + c \frac{N d}{b} - \frac{b}{2N} \left(c^2 \frac{N^2 d^2}{b^2}\right) + O(\log N) \\
&= \frac{N d}{b} \left[ c \log\left(\frac{b}{c d}\right) + c - \frac{c^2 d}{2} \right] + O(\log N) \\
&= \frac{N d}{b} h(c) + O(\log N),
\end{aligned}
\]
where \(h(c) = c \log\left( \frac{b}{c d} \right) + c - \frac{c^2 d}{2}\) and \(d = \log b\). When \(c=1\), we have
\[
h(1) = \log\left( \frac{b}{d} \right) + 1 - \frac{d}{2} = (\log b) + 1 - \frac{\log b}{2} - \log(\log b) = \frac{\log b}{2} + 1 - \log(\log b).
\]

Define \(f(d) = 1 + \frac{d}{2} - \log d\) (\(d = \log b > 0\)). Its derivative is
\[
f'(d) = \frac{1}{2} - \frac{1}{d}, \quad f''(d) = \frac{1}{d^2} > 0.
\]

The minimum occurs at \(d=2\) with \(f(2) = 1 + 1 - \log 2 \approx 1.307 > 0\). Since \(\lim_{d \to 0^+} f(d) = \infty\) and \(\lim_{d \to \infty} f(d) = \infty\), \(f(d) > 0\) for all \(d > 0\) (i.e., \(b > 1\)).
For \(k = \mu = \frac{N \log b}{b}\) (\(c=1\)), we have
\[
\log \mathbb{E}[Z_k] \sim \frac{N d}{b} h(1) \to \infty \implies \mathbb{E}[Z_k] \to \infty.
\]

Since \(h(c) \to -\infty\) as \(c \to \infty\) and \(h(c)\) is continuous, there must exist a \(c_0(b) > 0\) such that \(h(c_0) = 0\):
\begin{itemize}
  \item If \(b \geq e^2 \approx 7.389\), setting \(c_0=2\), we have
  \[
  h(2) = 2 \log\left(\frac{b}{2 \log b}\right) + 2 - 2 \log b = 2(\log b - \log 2 - \log(\log b)) + 2 - 2 \log b = -2 \log 2 + 2 - 2 \log(\log b) \leq -0.772 < 0.
  \]
  Thus, set \(C_{0}(b) = 2\).
\item If \(1 < b < e^2\), solve \(h(c_0) = 0\) numerically (e.g., bisection) and set \(C_0(b) = c_0 + 1 > 1\).
For \(k' = C_0(b) \mu = C_0(b) \frac{N \log b}{b}\), we have
\[
\log \mathbb{E}[Z_{k'}] \sim \frac{N d}{b} h(C_0(b)) \to -\infty \implies \mathbb{E}[Z_{k'}] \to 0 \quad \text{(exponentially)}.
\]
\end{itemize}

Given that w.h.p. \(|\mathcal{S}^{*}|\) is the size of an all-positive module, so w.h.p. \(|\mathcal{S}^{*}| = \max \{ k \mid Z_k > 0 \}\). Next, we prove that w.h.p. \(\mu \leq |\mathcal{S}^{*}| \leq k'\). For \(k = \mu\), by previous analysis, we know that \(\mathbb{E}[Z_k] \to \infty\). We now show \(\frac{\text{Var}(Z_k)}{(\mathbb{E}[Z_k])^2} \to 0\). According to variance decomposition, we have
\[
\text{Var}(Z_k) = \sum_{U} \text{Var}(I_U) + \sum_{U \neq V} \text{Cov}(I_U, I_V) \leq \mathbb{E}[Z_k] + \sum_{U \neq V} \mathbb{E}[I_U I_V],
\]
where \(I_U\) is the indicator that subset \(U\) forms an all-positive module. For the second term, we have
\[
\sum_{U \neq V} \mathbb{E}[I_U I_V] = \sum_{t=1}^{k-1} \sum_{\substack{U,V \\ |U \cap V|=t}} \mathbb{E}[I_U I_V] = \sum_{t=1}^{k-1} \binom{N}{k} \binom{k}{t} \binom{N-k}{k-t} \alpha^{2\binom{k}{2} - \binom{t}{2}}.
\]

The variance ratio is
\[
\frac{\text{Var}(Z_k)}{(\mathbb{E}[Z_k])^2} \leq \frac{1}{\mathbb{E}[Z_k]} + \sum_{t=1}^{k-1} \Gamma_t \alpha^{-\binom{t}{2}} \mathrm{~with~} \Gamma_t = \frac{\binom{k}{t} \binom{N-k}{k-t}}{\binom{N}{k}}.
\]

For $\Gamma_t$ and $\alpha^{-\binom{t}{2}}$, we have \(\Gamma_t \leq \left( \frac{e k^2}{t N} \right)^t\) and \(\alpha^{-\binom{t}{2}} = \exp\left( -\binom{t}{2} \log \alpha \right) \leq \exp\left( \binom{t}{2} \frac{2b}{N} \right) \leq \exp\left( \frac{b t^2}{N} \right)\) (since \(\log \frac{1}{\alpha} \leq \frac{2b}{N}\)). Substituting \(k = \mu = \frac{N d}{b}\) (so \(\frac{k^2}{N} = \frac{N d^2}{b^2}\)) and defining
\[
g(t) = t \log \left( \frac{e k^2}{t N} \right) + \frac{b t^2}{N} = t (\log N + \log(e d^2 / b^2) - \log t) + \frac{b t^2}{N}.
\]

For fixed \(b\), \(\log(e d^2 / b^2) = O(1)\). We split the summation by considering the following two cases:
\begin{itemize}
  \item Case 1: when \(1 \leq t \leq \sqrt{N}\), we have \(g(t) \leq t (\log N + C_1)\) (\(C_1 = \log(e d^2 / b^2)\) constant). So,
  \[
  \Gamma_t \alpha^{-\binom{t}{2}} \leq \exp(g(t)) \leq \exp(t (\log N + C_1)) = (e^{C_1} N)^t.
  \]
  Then we have
  \[
  \sum_{t=1}^{\lfloor \sqrt{N} \rfloor} \Gamma_t \alpha^{-\binom{t}{2}} \leq \sum_{t=1}^{\lfloor \sqrt{N} \rfloor} (e^{C_1} N)^t \leq \sqrt{N} (e^{C_1} N)^{\sqrt{N}} = \exp\left( \Theta(\sqrt{N} \log N) \right).
  \]
  By previous analysis, we know that \(\mathbb{E}[Z_k] = e^{\Theta(N)}\) when $k=\mu$, so \((\mathbb{E}[Z_k])^2 = e^{\Theta(N)}\). Since \(\sqrt{N} \log N = o(N)\),
  \[
  \sum_{t=1}^{\lfloor \sqrt{N} \rfloor} \Gamma_t \alpha^{-\binom{t}{2}} = e^{o(N)} = o\left( (\mathbb{E}[Z_k])^2 \right).
  \]
\item Case 2: When \(\sqrt{N} < t \leq k-1\), by simple analysis, we have $g(t)\leq g(N)=D N$ with (\(D\) being a constant. Thus, we get
  \[
  \Gamma_t \alpha^{-\binom{t}{2}} \leq \exp(g(t)) \leq e^{D N},
  \]
  which gives
  \[
  \sum_{t=\lceil \sqrt{N} \rceil}^{k-1} \Gamma_t \alpha^{-\binom{t}{2}} \leq k e^{D N} = \Theta(N) e^{D N} = e^{D N + \log N}.
  \]
  Given that \(\log \mathbb{E}[Z_k] = \frac{N d}{b} h(1) + O(\log N)\) with \(h(1) > 0\) when $k=\mu$, so \((\mathbb{E}[Z_k])^2 = \exp\left( \frac{2N d}{b} h(1) + O(\log N) \right)\).
  Since \(h(1) > 0\), for large \(N\), \(\frac{2Nd}{b} h(1)=\Theta (N)\), implying:
  \[
  \sum_{t=\lceil \sqrt{N} \rceil}^{k-1} \Gamma_t \alpha^{-\binom{t}{2}} = o\left( (\mathbb{E}[Z_k])^2 \right).
  \]
\end{itemize}

Combining both cases gives
\[
\sum_{t=1}^{k-1} \Gamma_t \alpha^{-\binom{t}{2}} = o\left( (\mathbb{E}[Z_k])^2 \right), \quad \frac{1}{\mathbb{E}[Z_k]} \to 0 \implies \frac{\text{Var}(Z_k)}{(\mathbb{E}[Z_k])^2} \to 0.
\]

Then, by Chebyshev’s inequality, we have
\[
\mathbb{P}(Z_k = 0) \leq \frac{\text{Var}(Z_k)}{(\mathbb{E}[Z_k])^2} \to 0 \implies \mathbb{P}(Z_k > 0) \to 1.
\]

Since w.h.p. LSCBM is all-positive, w.h.p. \(|\mathcal{S}^{*}| \geq \mu\). For \(k' = C_0(b) \mu\), by previous analysis, we know that \(\mathbb{E}[Z_{k'}] \to 0\) exponentially. By Markov’s inequality:
\[
\mathbb{P}(Z_{k'} > 0) \leq \mathbb{E}[Z_{k'}] \to 0.
\]

For \(s > k'\), since \(h(c)\) is continuous and \(h(C_0(b)) < 0\), we have \(h(c) \leq -\varsigma(b) < 0\), where \(\varsigma(b) > 0\) depends on \(b\). Thus, we have
\[
\mathbb{E}[Z_s] \leq \exp\left( \frac{N d}{b} h(c) \right) \leq \exp\left( -\frac{N d}{b} \varsigma(b) \right) \quad \text{for} \quad s \geq k'.
\]

Then, we get
\[
\mathbb{P}(\exists \text{ all-positive module of size } \geq k') \leq \sum_{s=k'}^N \mathbb{E}[Z_s] \leq (N-k')\exp\left( -\frac{N d}{b} \varsigma(b) \right)\to 0.
\]

Since w.h.p. LSCBM is all-positive, w.h.p. \(|\mathcal{S}^{*}| \leq k'\). In conclusion, we have w.h.p. \(\mu \leq |\mathcal{S}^{*}| \leq k'\). By previous analysis, we see that there exist \(\delta > 0\) and \(N_0 > 0\) such that for \(N > N_0\),
\[
\mathbb{P}\left( |\mathcal{S}^{*}| \notin [\mu, k'] \right) \leq e^{-\delta N}.
\]

Decomposing the expectation gets
\[
\mathbb{E}[|\mathcal{S}^{*}|] = \sum_{m=0}^{N} \mathbb{P}(|\mathcal{S}^{*}| > m) = \sum_{m=0}^{\lfloor \mu \rfloor - 1} \mathbb{P}(|\mathcal{S}^{*}| > m) + \sum_{m=\lfloor \mu \rfloor}^{\lfloor k' \rfloor} \mathbb{P}(|\mathcal{S}^{*}| > m) + \sum_{m=\lfloor k' \rfloor + 1}^{N} \mathbb{P}(|\mathcal{S}^{*}| > m).
\]

For the lower bound, we have
  \[
  \mathbb{E}[|\mathcal{S}^{*}|] \geq \sum_{m=0}^{\lfloor \mu \rfloor - 1} \mathbb{P}(|\mathcal{S}^{*}| > \lfloor \mu \rfloor) \geq \sum_{m=0}^{\lfloor \mu \rfloor - 1} \mathbb{P}(|\mathcal{S}^{*}| \geq \mu) \geq \lfloor \mu \rfloor (1 - e^{-\delta N}).
  \]

  Since \(\mu = \Theta(N)\), \(\lfloor \mu \rfloor = \mu (1 + o(1))\), so:
  \[
  \mathbb{E}[|\mathcal{S}^{*}|] \geq \mu (1 - e^{-\delta N}) (1 + o(1)) \sim \mu = \frac{N \log b}{b}.
  \]

For the upper bound, we have
  \[
  \mathbb{E}[|\mathcal{S}^{*}|] \leq \sum_{m=0}^{\lfloor \mu \rfloor - 1} 1 + \sum_{m=\lfloor \mu \rfloor}^{\lfloor k' \rfloor} 1 + \sum_{m=\lfloor k' \rfloor + 1}^{N} \mathbb{P}(|\mathcal{S}^{*}| > k') \leq \lfloor \mu \rfloor + \lfloor k' \rfloor - \lfloor \mu \rfloor + N \mathbb{P}(|\mathcal{S}^{*}| > k').
  \]

Since \(\mathbb{P}(|\mathcal{S}^{*}| > k') \leq e^{-\delta N}\),  we have
  \[
  N \mathbb{P}(|\mathcal{S}^{*}| > k') \leq N e^{-\delta N} \to 0.
  \]

  Since \(\lfloor k' \rfloor = k' (1 + o(1)) = C_0(b) \frac{N \log b}{b} (1 + o(1))\):
  \[
  \mathbb{E}[|\mathcal{S}^{*}|] \leq k' + o(1) \sim C_0(b) \frac{N \log b}{b}.
  \]

As \(C_0(b)\) is a constant (depending on \(b\)), we have
\[
\frac{N \log b}{b} (1 - o(1)) \leq \mathbb{E}[|\mathcal{S}^{*}|] \leq C_0(b) \frac{N \log b}{b} (1 + o(1)) \implies \mathbb{E}[|\mathcal{S}^{*}|] = \Theta\left(\frac{N \log b}{b}\right).
\]

For the second part of this theorem, the strategy hinges on constructing sufficiently large, disjoint all-positive modules that cannot coexist within a single LSCBM due to size constraints. Select a constant \(\delta > 0\) satisfying: (1)  \(g(\delta) = \delta \ln(1/\delta) + \delta - \frac{b \delta^2}{2} > 0\) (ensured by choosing \(\delta < \delta_{\max}\), where \(\delta_{\max}\) is the largest root of \(g(\delta) = 0\)); (2) \(\delta > \frac{C_0(b) \log b}{2b}\) (guaranteeing \(2\delta N > C_0(b)\mu\) w.h.p.). Such \(\delta\) exists for \(b > 1\): \(g(\delta) > 0\) holds for small \(\delta > 0\) due to the \(\delta \ln(1/\delta)\) term dominating, while \(\frac{C_0(b) \log b}{2b}\) is a fixed positive constant. Set \(s = \lfloor \delta N \rfloor\).

We know that the expected number of size-\(s\) all-positive cliques (trivially balanced SCBMs) is:
\[
\mathbb{E}[Z_s] = \binom{N}{s} \alpha^{\binom{s}{2}}.
\]

Using \(\binom{N}{s} \leq \left( \frac{eN}{s} \right)^s\) and \(\ln \alpha \leq -\frac{b}{2N}\) for large \(N\) gives
\[
\ln \mathbb{E}[Z_s] \leq s \ln\left(\frac{eN}{s}\right) - \frac{b s(s-1)}{4N}.
\]

Substituting \(s = \delta N\) gets
\[
\ln \mathbb{E}[Z_s] \leq N \left[ \delta \ln(1/\delta) + \delta - \frac{b \delta^2}{4} + o(1) \right] = N \left[ g(\delta) + \frac{b \delta^2}{4} \right] + o(N).
\]

Since \(g(\delta) > 0\), the expression in brackets is positive, implying \(\mathbb{E}[Z_s] \to \infty\). Thus, size-\(s\) all-positive modules are abundant. Now we define \(Q\) as the number of unordered pairs \(\{A, B\}\) of disjoint size-\(s\) all-positive modules. Its expectation is:
\[
\mathbb{E}[Q] = \frac{1}{2} \binom{N}{s} \binom{N-s}{s} \left( \alpha^{\binom{s}{2}} \right)^2 = \frac{1}{2} \mathbb{E}[Z_s] \cdot \binom{N-s}{s} \alpha^{\binom{s}{2}}.
\]

Bounding the second factor gives
\[
\binom{N-s}{s} \alpha^{\binom{s}{2}} \leq \exp\left( N \left[ \delta \ln(1/\delta) + \delta - \delta^2 - \frac{b \delta^2}{4} + o(1) \right] \right).
\]

Combined with \(\mathbb{E}[Z_s] = \exp\left( N \left[ \delta \ln(1/\delta) + \delta - \frac{b \delta^2}{4} + o(1) \right] \right)\), we get:
\[
\mathbb{E}[Q] \leq \frac{1}{2} \exp\left( N \left[ 2\delta \ln(1/\delta) + 2\delta - \delta^2 - \frac{b \delta^2}{2} + o(1) \right] \right) = \frac{1}{2} \exp\left( N \left[ 2g(\delta) - \delta^2 + o(1) \right] \right).
\]

As \(g(\delta) > 0\) and dominates \(\delta^2\) for small \(\delta\), the exponent is positive, so \(\mathbb{E}[Q] \to \infty\). Variance analysis (similar to Theorem \ref{main4}) shows \(\text{Var}(Q) = o(\mathbb{E}[Q]^2)\). By Chebyshev's inequality, we have
\[
\mathbb{P}\left(Q < \tfrac{1}{2}\mathbb{E}[Q]\right) \leq \frac{4\text{Var}(Q)}{\mathbb{E}[Q]^2} \to 0,
\]
which implyies \(\mathbb{P}(Q \geq 1) \to 1\). Thus, w.h.p. there exists a pair \(\{A, B\}\) of disjoint size-\(s\) all-positive modules. Since each is an SCBM, if they belonged to the same LSCBM \(\mathcal{S}^*\), then \(|\mathcal{S}^*| \geq 2s\). However, by the proof of the first part of Theorem \ref{main2} and our choice of \(\delta\), we have
\[
2s \approx 2\delta N > 2 \cdot \tfrac{C_0(b) \log b}{2b} N = C_0(b)\mu \geq S_{\max} \quad \text{w.h.p.},
\]
a contradiction. Therefore, \(A\) and \(B\) must reside in distinct LSCBMs. Combining these results:
\[
\mathbb{P}(\text{Multiple LSCBMs}) \geq \mathbb{P}(Q \geq 1) - \mathbb{P}(S_{\max} < \mu) - \mathbb{P}(S_{\max} > C_(b)\mu) \to 1,
\]
completing the proof.
\end{proof}
\subsection{Proof of Theorem \ref{main3}}
\begin{proof}
For the first part of Theorem \ref{main3}, let \(Z_s\) denote the number of strong-correlation balanced modules (SCBM) of size \(s \geq 3\) in \(\mathcal{G}(N,\alpha,\beta)\). The size of the LSCBM is \(|\mathcal{S}^{*}| = \max \{ s \mid Z_s > 0 \}\). We bound \(\mathbb{E}[|\mathcal{S}^{*}|]\) by analyzing \(\mathbb{E}[Z_s]\) and summing over \(s\). For a fixed vertex set \(U\) of size \(s\), the probability that \(U\) forms an SCBM is bounded by summing over all possible partitions \(A \cup B = U\). There are \(2^s\) partitions (each vertex assigned independently to \(A\) or \(B\)), and for a partition with \(|A| = k\), \(|B| = s - k\), the probability is \(\alpha^{\binom{k}{2} + \binom{s-k}{2}} \beta^{k(s-k)}\). Since \(\beta \leq 1\), we have
\[
\mathbb{P}(U \text{ is SCBM})=\sum_{k=0}^{s} \binom{s}{k} \alpha^{\binom{k}{2} + \binom{s-k}{2}} \beta^{k(s-k)} \leq 2^s \max_{k} \left( \alpha^{\binom{k}{2} + \binom{s-k}{2}} \right).
\]

The maximum is attained at partitions minimizing the exponent of \(\alpha\). As \(\binom{k}{2} + \binom{s-k}{2} \geq \frac{s(s-2)}{4}\) for all \(k\), we have
\[
\max_{k} \left( \alpha^{\binom{k}{2} + \binom{s-k}{2}}\right) \leq \alpha^{\frac{s(s-2)}{4}},
\]
which gives
\[
\mathbb{P}(U \text{ is SCBM}) \leq 2^s \alpha^{\frac{s(s-2)}{4}}.
\]

The expected number of SCBMs of size \(s\) is
\[
\mathbb{E}[Z_s] = \binom{N}{s} \mathbb{P}(U \text{ is SCBM}) \leq \binom{N}{s} 2^s \alpha^{\frac{s(s-2)}{4}} .
\]

Using the bound \(\binom{N}{s} \leq \left( \frac{eN}{s} \right)^s\) gives
\[
\mathbb{E}[Z_s] \leq \left( \frac{eN}{s} \right)^s 2^s \alpha^{\frac{s(s-2)}{4}} = \left( \frac{2eN}{s} \right)^s \alpha^{\frac{s(s-2)}{4}}.
\]

Taking the natural logarithm gets
\[
\log \mathbb{E}[Z_s] \leq s \log \left( \frac{2eN}{s} \right) + \frac{s(s-2)}{4}\log \alpha.
\]

For \(s \geq 4\), \(\frac{s(s-2)}{4}\geq \frac{s^2}{8}\). Thus, we have
\[
\log \mathbb{E}[Z_s] \leq s \log \left( \frac{2eN}{s} \right) - \frac{s^2}{8} |\log \alpha|, \quad \text{for } s \geq 4.
\]

We express the expectation as
\[
\mathbb{E}[|\mathcal{S}^{*}|] = \sum_{m=3}^{N} \mathbb{P}(|\mathcal{S}^{*}| \geq m).
\]

Set \(k = c \frac{\log N}{|\log \alpha|}\) with constant \(c > 8\) (to be determined). We have
\[
\mathbb{E}[|\mathcal{S}^{*}|] = \underbrace{\sum_{m=3}^{\lfloor k \rfloor} \mathbb{P}(|\mathcal{S}^{*}| \geq m)}_{\text{Sum I}} + \underbrace{\sum_{m=\lfloor k \rfloor + 1}^{N} \mathbb{P}(|\mathcal{S}^{*}| \geq m)}_{\text{Sum II}}.
\]

For Sum I, since \(\mathbb{P}(|\mathcal{S}^{*}| \geq m) \leq 1\), we get
  \[
  \text{Sum I} \leq \lfloor k \rfloor - 2 \leq k.
  \]

For Sum II, by the union bound, \(\mathbb{P}(|\mathcal{S}^{*}| \geq m) \leq \sum_{s=m}^{N} \mathbb{P}(Z_s \geq 1) \leq \sum_{s=m}^{N} \mathbb{E}[Z_s]\). Thus, we have
  \[
  \text{Sum II} \leq \sum_{m=\lfloor k \rfloor + 1}^{N} \sum_{s=m}^{N} \mathbb{E}[Z_s] = \sum_{s=\lfloor k \rfloor + 1}^{N} \mathbb{E}[Z_s] (s - \lfloor k \rfloor) \leq \sum_{s=\lfloor k \rfloor + 1}^{N} s \mathbb{E}[Z_s].
  \]

For \(s \geq \lfloor k \rfloor + 1 \geq k\) and large \(N\), by previous analysis, we know that
\[
\log \mathbb{E}[Z_s] \leq s \log \left( \frac{2eN}{s} \right) - \frac{s^2}{8} |\log \alpha|.
\]

We show that for large \(N\),
\[
s \log \left( \frac{2eN}{s} \right) \leq \frac{s^2}{16} |\log \alpha|.
\]

Rearranging it gives
\[
\log \left( \frac{2eN}{s} \right) \leq \frac{s}{16} |\log \alpha|.
\]

Substitute \(s > k = c \frac{\log N}{|\log \alpha|}\):
\[
\frac{s}{16} |\log \alpha| > \frac{c \log N}{16}.
\]

The left side is
\[
\log \left( \frac{2eN}{s} \right) = \log(2e) + \log N - \log s \leq \log(2e) + \log N \quad \text{(since \(\log s > 0\))}.
\]

For \(c > 16\), \(\frac{c}{16} > 1\), so \(\frac{c \log N}{16} > \log(2e) + \log N\) for large \(N\) as \(\log N\) dominates. Thus, we have
\[
\log \left( \frac{2eN}{s} \right) \leq \frac{s}{16} |\log \alpha|,
\]
which implies
\[
s \log \left( \frac{2eN}{s} \right) - \frac{s^2}{8} |\log \alpha| \leq -\frac{s^2}{16} |\log \alpha|.
\]

So we have
\[
\mathbb{E}[Z_s] \leq \exp \left( -\frac{s^2}{16} |\log \alpha| \right),
\]
which gives
\[
\text{Sum II} \leq \sum_{s=\lfloor k \rfloor + 1}^{N} s \exp \left( -\frac{s^2}{16} |\log \alpha| \right).
\]

For large \(N\), the function \(f(x) = x \exp \left( -\frac{x^2}{16} |\log \alpha| \right)\) is decreasing for \(x \geq k\) (as \(k \to \infty\)). Thus,
\[
\sum_{s=\lfloor k \rfloor + 1}^{N} s \exp \left( -\frac{s^2}{16} |\log \alpha| \right) \leq \int_{k}^{\infty} x \exp \left( -\frac{x^2}{16} |\log \alpha| \right) dx.
\]

Compute the integral:
\[
\int_{k}^{\infty} x e^{-a x^2} dx = \frac{1}{2a} e^{-a k^2}, \quad \text{where } a = \frac{|\log \alpha|}{16}.
\]

Substituting \(a\) gives
\[
\frac{1}{2a} e^{-a k^2} = \frac{8}{|\log \alpha|} \exp \left( -\frac{k^2}{16} |\log \alpha| \right).
\]

Now substitute \(k = c \frac{\log N}{|\log \alpha|}\):
\[
\exp \left( -\frac{k^2}{16} |\log \alpha| \right) = \exp \left( -\frac{c^2 (\log N)^2}{16 |\log \alpha|} \right).
\]

As \(N \to \infty\), \(|\log \alpha| \to \infty\), so we have
 \(\frac{8}{|\log \alpha|} \to 0\), \(\exp \left( -\frac{c^2 (\log N)^2}{16 |\log \alpha|} \right) \leq 1\),
and 
\[
\frac{8}{|\log \alpha|} \exp \left( -\frac{c^2 (\log N)^2}{16 |\log \alpha|} \right) \to 0.
\]

Thus, we get \(\text{Sum II} \to 0\). Combining the sums obtains
\[
\mathbb{E}[|\mathcal{S}^{*}|] \leq k + o(1) = c \frac{\log N}{|\log \alpha|} + o(1).
\]

Since \(c > 8\) is a constant, we have
\[
\mathbb{E}[|\mathcal{S}^{*}|] = O\left( \frac{\log N}{|\log \alpha|} \right).
\]

For the second part of Theorem \ref{main3}, we define the event \(\mathcal{A} := \{ s_{\text{low}} \leq |\mathcal{S}^*| < s_{\text{high}} \}\) with \(s_{\text{low}} = \lfloor 0.1k \rfloor\), \(s_{\text{high}} = \lceil (8 - 0.1)k \rceil\), and \(k = \log N / |\log \alpha|\), where \(|\log \alpha| = o(\sqrt{\log N})\). For \(s = s_{\text{high}}\), the expectation \(\mathbb{E}[Z_s]\) vanishes asymptotically. From previous analysis, we know that
  \[
  \mathbb{E}[Z_s]\leq \left( \frac{2eN}{s} \right)^s \alpha^{\binom{s-1}{2}}.
  \]

  Taking logarithms and using \(\binom{s-1}{2} \geq s^2/8\) for \(s \geq 4\), we substitute \(s = (8 - 0.1)k\) and get
  \[
  \log \mathbb{E}[Z_s]\leq -\Theta\left( \frac{(\log N)^2}{|\log \alpha|} \right) \to -\infty.
  \]

  By Markov’s inequality, we have \(\mathbb{P}(|\mathcal{S}^*| \geq s_{\text{high}}) \leq \mathbb{E}[Z_s]\to 0\). For \(s_0 = s_{\text{low}}\), consider SCBMs composed solely of positive edges (i.e., \(B = \emptyset\)). The expectation \(\mathbb{E}[Z_{s_0}^{\text{pos}}]\) diverges:
  \[
  \mathbb{E}[Z_{s_0}^{\text{pos}}]\geq \left( \frac{N}{s_0} \right)^{s_0} e^{-s_0} \alpha^{\binom{s_0}{2}}.
  \]

  Substituting \(s_0 = 0.1k\), the dominant term in \(\log \mathbb{E}[Z_{s_0}^{\text{pos}}]\) is \(0.095(\log N)^2 / |\log \alpha| \to \infty\). Since \(\mathbb{E}[Z_{s_0}] \geq \mathbb{E}[Z_{s_0}^{\text{pos}}]\), we have \(\mathbb{E}[Z_{s_0}] \to \infty\). Critically, Lemma \ref{main31} given later establishes \(\text{Var}(Z_{s_0}) = o(\mathbb{E}[Z_{s_0}]^2)\). By Chebyshev’s inequality, we have
  \[
  \mathbb{P}(Z_{s_0} = 0) \leq \frac{\text{Var}(Z_{s_0})}{\mathbb{E}[Z_{s_0}]^2} \to 0 \implies \mathbb{P}(|\mathcal{S}^*| \geq s_0) \to 1,
  \]
  which gives \(\mathbb{P}(A) \to 1\).

For any fixed \(s \in [s_{\text{low}}, s_{\text{high}})\), the same analysis as above shows \(\inf_{s} \mathbb{E}[Z_s] \to \infty\). Moreover, Lemma \ref{main31} guarantees that\(\text{Var}(Z_s) = o(\mathbb{E}[Z_s]^2)\) uniformly over \(s\) in this interval. By Chebyshev's inequality, we have
\[
\mathbb{P}\left( |Z_s - \mathbb{E}[Z_s]| \geq \tfrac{1}{2}\mathbb{E}[Z_s] \right) \leq \frac{4\text{Var}(Z_s)}{\mathbb{E}[Z_s]^2} \to 0 \quad \text{(uniformly in } s\text{)}.
\]

Hence, \(\mathbb{P}(Z_s \geq \tfrac{1}{2}\mathbb{E}[Z_s]) \to 1\) uniformly. Since \(\inf_s \mathbb{E}[Z_s] \to \infty\), there exists a \(N_0\) such that \(\tfrac{1}{2}\mathbb{E}[Z_s] \geq 2\) for all \(s \in [s_{\text{low}}, s_{\text{high}})\) and \(N > N_0\). Consequently, \(\mathbb{P}(Z_s \geq 2) \to 1\) uniformly in \(s\). By the law of total probability, we have
\[
\mathbb{P}(Z_{|\mathcal{S}^*|} \geq 2) \geq \mathbb{P}(Z_{|\mathcal{S}^*|} \geq 2 \mid \mathcal{A}) \mathbb{P}(\mathcal{A}).
\]

Since \(\mathbb{P}(\mathcal{A}) \to 1\), it suffices to show \(\mathbb{P}(Z_{|\mathcal{S}^*|} \geq 2 \mid \mathcal{A}) \to 1\). Conditioned on \(\mathcal{A}\), \(|\mathcal{S}^*| = s\) for some \(s \in [s_{\text{low}}, s_{\text{high}})\). Define \(h(s) := \mathbb{P}(Z_s < 2)\). We require:
\[
\mathbb{E}\left[ h(|\mathcal{S}^*|) \mid \mathcal{A} \right] \to 0.
\]
By Lemma \ref{main32} provided later, \(\sup_{s \in [s_{\text{low}}, s_{\text{high}})} h(s) \to 0\). Thus, we have
\[
\mathbb{E}\left[ h(|\mathcal{S}^*|) \mid \mathcal{A} \right] \leq \sup_{s} h(s) \to 0,
\]
implying \(\mathbb{P}(Z_{|\mathcal{S}^*|} \geq 2 \mid \mathcal{A}) = 1 - \mathbb{E}\left[ h(|\mathcal{S}^*|) \mid \mathcal{A} \right] \to 1\). Finally, we have
\[
\lim_{N \to \infty} \mathbb{P}(Z_{|\mathcal{S}^*|} \geq 2) = 1.
\]

\begin{lem}\label{main31}
(Variance control) For \(s = \lfloor c k \rfloor\) with \(c < 8\) and \(k = \log N / |\log \alpha|\), assume \(|\log \alpha| = o(\sqrt{\log N})\). We have \(\text{Var}(Z_s) = o(\mathbb{E}[Z_s]^2)\).
\end{lem}
\begin{proof}
Using the variance decomposition from proofs of Theorem \ref{main4} gives
\[
\frac{\text{Var}(Z_s)}{\mathbb{E}[Z_s]^2} \leq \frac{1}{\mathbb{E}[Z_s]} + \sum_{t=1}^{2} \Gamma_t \frac{1}{\mu_s} + \sum_{t=3}^{s-1} \Gamma_t \frac{1}{\mu_t},
\]
where \(\Gamma_t = \binom{s}{t} \frac{\binom{N-s}{s-t}}{\binom{N}{s}} \leq \left( \frac{e s^2}{t N} \right)^t\), and \(\mu_s = \mathbb{E}[I_U]\) for \(|U| = s\). From previous analysis, we know that
\[
\mu_s \leq \exp\left( s \log(2eN/s) - \tfrac{s^2}{8} |\log \alpha| \right), \quad s \geq 4.
\]

Substituting \(s = c k = c \log N / |\log \alpha|\) gives
\[
\log \mu_s \leq \left( c - \tfrac{c^2}{8} \right) \tfrac{(\log N)^2}{|\log \alpha|} + O\left( \tfrac{\log N \log \log N}{|\log \alpha|} \right).
\]

Since \(c < 8\), the coefficient \(c - c^2/8 > 0\), so \(\mathbb{E}[Z_s] \to \infty\), i.e., \(\frac{1}{\mathbb{E}[Z_s]} \to 0\). For(\(t=1,2\)), we have
  \[
  \Gamma_t \frac{1}{\mu_s} \leq \exp\left( -\Theta\left( \tfrac{(\log N)^2}{|\log \alpha|} \right) \right) \to 0.
  \]

For (\(t \geq 3\)), we have
  \[
  \Gamma_t \frac{1}{\mu_t} \leq \exp\left( -2t \log N + \tfrac{t^2}{8} |\log \alpha| + O(t \log \log N) \right).
  \]

  The exponent is dominated by \(-2t \log N + \tfrac{t^2}{8} |\log \alpha|\), which is maximized at \(t=3\) and strictly negative for large \(N\) due to \(|\log \alpha| = o(\sqrt{\log N})\). Summing over \(t\) gets
  \[
  \sum_{t=3}^{s-1} \Gamma_t \frac{1}{\mu_t} \leq s \exp\left( K \tfrac{(\log N)^2}{|\log \alpha|} \right) \to 0, \quad K < 0.
  \]

Thus, we have \(\frac{\text{Var}(Z_s)}{\mathbb{E}[Z_s]^2} \to 0\).
\end{proof}
\begin{lem}\label{main32}
(Uniform convergence) For \(h(s) = \mathbb{P}(Z_s < 2)\), we have \(\sup_{s \in [s_{\text{low}}, s_{\text{high}}]} h(s) \to 0\) as \(N \to \infty\), where \(s_{\text{low}} = \lfloor 0.1k \rfloor\), \(s_{\text{high}} = \lfloor (8 - 0.1)k \rfloor\), and \(k = \log N / |\log \alpha|\) with \(|\log \alpha| = o(\sqrt{\log N})\).
\end{lem}
\begin{proof}
Fix an arbitrary \(s \in [s_{\text{low}}, s_{\text{high}}]\). By Lemma \ref{main31}, the variance of \(Z_s\) satisfies \(\text{Var}(Z_s) = o(\mathbb{E}[Z_s]^2)\) uniformly over \(s\) in this interval. From the earlier analysis of Theorem \ref{main3}, the expectation \(\mathbb{E}[Z_s] \to \infty\) uniformly for \(s \in [s_{\text{low}}, s_{\text{high}}]\). Consequently, there exists \(N_0 > 0\) such that for all \(N > N_0\) and all \(s \in [s_{\text{low}}, s_{\text{high}}]\),
\[
\mathbb{E}[Z_s] > 2.
\]

Now consider the event \(Z_s < 2\). Since \(Z_s\) is non-negative integer-valued, \(Z_s < 2\) implies \(Z_s \leq 1\). Therefore,
\[
\{Z_s < 2\} \subseteq \{ |Z_s - \mathbb{E}[Z_s]| \geq \mathbb{E}[Z_s] - 1.5 \},
\]
where the inclusion holds because \(\mathbb{E}[Z_s] > 2\) ensures \(\mathbb{E}[Z_s] - 1.5 > 0.5 > 0\), and if \(Z_s \leq 1\), then
\[
|Z_s - \mathbb{E}[Z_s]| \geq \mathbb{E}[Z_s] - 1 \geq \mathbb{E}[Z_s] - 1.5 + 0.5 > \mathbb{E}[Z_s] - 1.5.
\]

Applying Chebyshev’s inequality to the right-hand side event gives
\[
h(s) = \mathbb{P}(Z_s < 2) \leq \mathbb{P}\left( |Z_s - \mathbb{E}[Z_s]| \geq \mathbb{E}[Z_s] - 1.5 \right) \leq \frac{\text{Var}(Z_s)}{(\mathbb{E}[Z_s] - 1.5)^2}.
\]

We now analyze this upper bound uniformly in \(s\). First, from the denominator, we have
\[
(\mathbb{E}[Z_s] - 1.5)^2 = \mathbb{E}[Z_s]^2 \left(1 - \frac{1.5}{\mathbb{E}[Z_s]}\right)^2,
\]
which gives
\[
\frac{\text{Var}(Z_s)}{(\mathbb{E}[Z_s] - 1.5)^2} = \frac{\text{Var}(Z_s)}{\mathbb{E}[Z_s]^2} \cdot \frac{1}{\left(1 - \frac{1.5}{\mathbb{E}[Z_s]}\right)^2}.
\]

Since \(\mathbb{E}[Z_s] \to \infty\) and
 \(\frac{\text{Var}(Z_s)}{\mathbb{E}[Z_s]^2} \to 0\) uniformly in \(s\), we have
\[
\sup_{s \in [s_{\text{low}}, s_{\text{high}}]} h(s) \leq \sup_{s \in [s_{\text{low}}, s_{\text{high}}]} \frac{\text{Var}(Z_s)}{(\mathbb{E}[Z_s] - 1.5)^2} \to 0.
\]
\end{proof}
\end{proof}
\subsection{Variance bound for subcritical modules}
The following theorem is used to prove Theorem \ref{main1}.
\begin{thm}\label{main4}
Under \(\mathcal{G}(N, \alpha, \beta)\), for any $\epsilon\in(0,1)$, let \(s=(1-\epsilon)s_{c}=(1-\epsilon) \frac{\log N}{\lambda(\alpha,\beta)}\), where \(\lambda(\alpha,\beta) > 0\) is defined as:
\[
\lambda(\alpha,\beta) =
\begin{cases}
\frac{1}{2} |\log \alpha| & \alpha \geq \beta \\
\frac{1}{4} (|\log \alpha| + |\log \beta|) & \alpha < \beta
\end{cases}.
\]
We have
\[
\text{Var}(Z_{s})=o(\mathbb{E}[Z_{s}]^2)\qquad \mathrm{as~}N\rightarrow\infty.
\]
\end{thm}
\begin{proof}
Let \(I_U\) be the indicator random variable for the event that the subset \(U \subseteq \{1,2,\ldots,N\}\) is a strong-correlation balanced module (SCBM) of size \(s\). Then \(Z_s = \sum_{U: |U|=s} I_U\). The variance decomposes of  $Z_{s}$ can be written as
\[
\text{Var}(Z_s) = \sum_{U} \text{Var}(I_U) + \sum_{U \neq V} \text{Cov}(I_U, I_V).
\]

For the term \(\sum_{U}\text{Var}(I_U)\), since \(\text{Var}(I_U) = \mathbb{E}[I_U^2] - (\mathbb{E}[I_U)]^2 \leq \mathbb{E}[I_U^2] = \mathbb{E}[I_U]\) (as \(I_U^2 = I_U\)), we have:
  \[
  \sum_{U} \text{Var}(I_U) \leq \sum_{U} \mathbb{E}[I_U]= \mathbb{E}[Z_s].
  \]

Next, we focus on the term\(\sum_{U \neq V} \text{Cov}(I_U, I_V)\). By the covariance definition and \(I_U, I_V \geq 0\), \(\text{Cov}(I_U, I_V) \leq \mathbb{E}[I_U I_V]\). When \(U \cap V = \emptyset\), \(I_U\) and \(I_V\) are independent (due to edge independence), so \(\text{Cov}(I_U, I_V) = 0\). Thus, we only consider pairs with \(|U \cap V| = t \geq 1\), where $t\in\{1,2,\ldots,s-1\}$ and $t\neq s$ because $U=V$ if $t=s$.

Fix \(t \in \{1, 2, \dots, s-1\}\).  We define \(\mathcal{P}_t = \{(U, V) : |U| = |V| = s, |U \cap V| = t\}\), where \(|\mathcal{P}_t| = \binom{N}{s} \binom{s}{t} \binom{N-s}{s-t}\) (ways to choose \(U\), intersection \(U \cap V\), and \(V \setminus U\)).
Then we have
\[
\sum_{U \neq V} \text{Cov}(I_U, I_V)\leq\sum_{U \neq V} \mathbb{E}[I_U I_V]= \sum_{t=1}^{s-1} \sum_{(U,V) \in \mathcal{P}_t} \mathbb{E}[I_U I_V].
\]

Let \(W = U \cap V\) with \(|W| = t\). For the case \(t < 3\), recall that every SCBM requires a minimum size of 3, since \(I_U I_V \leq I_U\), we have
  \[
  \mathbb{E}[I_U I_V]\leq \mathbb{E}[I_U]= \mu_s, \quad \text{where } \mu_s = \mathbb{E}[I_U].
  \]

For the case \(t \geq 3\), the following two lemmas hold.
\begin{lem}\label{lem1}
If \(I_U = 1\) and \(I_V = 1\), then \(I_W = 1\).
\end{lem}
\begin{proof}
\(U\) and $V$ are SCBMs, so they are modules and structurally balanced. Since \(W \subseteq U\) and \(|W| = t \geq 3\), \(W\) is a SCBM, i.e., \(I_W = 1\).
\end{proof}
\begin{lem}\label{lem2}
Given \(I_W = 1\), the events \(\{I_U = 1\}\) and \(\{I_V = 1\}\) are conditionally independent.
\end{lem}
\begin{proof}
Fix the edge set \(\mathcal{E}_W\) of \(W\) that satisfy SCBM conditions. Since \((U \setminus W) \cap (V \setminus W) = \emptyset\), the edge sets \(\mathcal{E}_{U \setminus W}\) (edges within \(U \setminus W\) and between \((U \setminus W)\) and \(W\)) and \(\mathcal{E}_{V \setminus W}\) (edges within \(V \setminus W\) and between \((V \setminus W)\) and \(W\)) satisfy: \(\mathcal{E}_{U \setminus W} \cap \mathcal{E}_{V \setminus W} = \emptyset\). All edges are generated independently, so \(\mathcal{E}_{U \setminus W}\) and \(\mathcal{E}_{V \setminus W}\) are independent. \(I_U = 1\) if and only if \(\mathcal{E}_{U \setminus W}\) satisfies SCBM conditions given \(\mathcal{E}_W\), and similarly for \(I_V\). Thus, given \(\mathcal{E}_W\) (i.e., \(I_W = 1\)), \(I_U\) and \(I_V\) depend on independent edge sets and are conditionally independent.
\end{proof}

By the law of total expectation and Lemma \ref{lem1}, we have
  \[
  \mathbb{E}[I_U I_V]= \mathbb{E}[\mathbb{E}[I_U I_V \mid I_W]]\leq\mathbb{E}[I_W \cdot \mathbb{E}[I_U I_V \mid I_W]].
  \]

Given \(I_W = 1\), conditional independence in Lemma \ref{lem2} implies
  \[
  \mathbb{E}[I_U I_V \mid I_W = 1]= \mathbb{P}(I_U = 1 \mid I_W = 1) \mathbb{P}(I_V = 1 \mid I_W = 1).
  \]

By Lemma \ref{lem1}, we have \(\{I_U = 1\} \subseteq \{I_W = 1\}\), which gives
  \[
  \mathbb{P}(I_U = 1 \mid I_W = 1) = \frac{\mathbb{P}(I_U = 1)}{\mathbb{P}(I_W = 1)} = \frac{\mu_s}{\mu_t}, \quad \text{similarly } \mathbb{P}(I_V = 1 \mid I_W = 1) = \frac{\mu_s}{\mu_t}.
  \]

Thus, we have
  \[
  \mathbb{E}[I_U I_V]\leq \mathbb{E}[I_W]\cdot \left( \frac{\mu_s}{\mu_t} \right)^2 = \mu_t \cdot \frac{\mu_s^2}{\mu_t^2} = \frac{\mu_s^2}{\mu_t}.
  \]

Define:
\[
c_t =
\begin{cases}
\mu_s & \text{if } t < 3, \\
\frac{\mu_s^2}{\mu_t} & \text{if } t \geq 3.
\end{cases}
\]

Then we have \(\mathbb{E}[I_U I_V]\leq c_t\) for all \(t\in\{1,2,\ldots,s-1\}\). Substituting \(\mathbb{E}[I_U I_V]\leq c_t\) into the variance decomposition gives
\[
\text{Var}(Z_s) \leq \mathbb{E}[Z_s]+ \sum_{t=1}^{s-1} |\mathcal{P}_t| c_t.
\]

Given that $Z_{s}=\sum_{U:|U|=s}I_{U}$, we have $\mathbb{E}[Z_s]=\sum_{U:|U|=s}\mathbb{E}[I_{U}]=\binom{N}{s} \mu_s$. Combing \(\mathbb{E}[Z_s]= \binom{N}{s} \mu_s\) with \(|\mathcal{P}_t| = \binom{N}{s} \binom{s}{t} \binom{N-s}{s-t}\) gives
\[
\frac{\text{Var}(Z_s)}{\mathbb{E}[Z_s]^2} \leq \frac{1}{\mathbb{E}[Z_s]} + \sum_{t=1}^{s-1} \frac{\binom{s}{t} \binom{N-s}{s-t} c_t}{\binom{N}{s} \mu_s^2}.
\]

For \(t < 3\),\(c_t = \mu_s\), so we have
  \[
  \frac{\binom{s}{t} \binom{N-s}{s-t} \mu_s}{\binom{N}{s} \mu_s^2} = \frac{\binom{s}{t} \binom{N-s}{s-t}}{\binom{N}{s}} \cdot \frac{1}{\mu_s}.
  \]

For \(t \geq 3\), \(c_t = \frac{\mu_s^2}{\mu_t}\), so we have
  \[
  \frac{\binom{s}{t} \binom{N-s}{s-t} \frac{\mu_s^2}{\mu_t}}{\binom{N}{s} \mu_s^2} = \frac{\binom{s}{t} \binom{N-s}{s-t}}{\binom{N}{s}} \cdot \frac{1}{\mu_t}.
  \]

Thus we have
\[
\frac{\text{Var}(Z_s)}{\mathbb{E}[Z_s]^2} \leq \underbrace{\frac{1}{\mathbb{E}[Z_s]}}_{(I)} + \sum_{t=1}^{2} \underbrace{\frac{\binom{s}{t} \binom{N-s}{s-t}}{\binom{N}{s}} \cdot \frac{1}{\mu_s}}_{(II)} + \sum_{t=3}^{s-1} \underbrace{\frac{\binom{s}{t} \binom{N-s}{s-t}}{\binom{N}{s}} \cdot \frac{1}{\mu_t}}_{(III)}.
\]

Let \(s = (1-\epsilon) \frac{\log N}{\lambda(\alpha,\beta)}\). From former analysis, we know that
\[
\mathbb{E}[Z_s]= \exp\left( s \log N + s^2 f(a^*) + o(s) \right), \quad f(a^*) = -\lambda(\alpha,\beta) < 0,
\]
where \(a^*\) is the partition ratio maximizing \(f(a)\). Since \(\mathbb{E}[Z_s]= \binom{N}{s} \mu_s\), we have \(\mu_{s}=\exp\left(s^{2}f(a^{*})+s\log s-s+o(s))\right)\).  Similarly, \(\mu_t = \exp\left( t^2 f(a^*)+t\log t-t + o(t) \right)\). Thus we have
\[
\frac{1}{\mu_s} =\exp\left( \lambda(\alpha, \beta) s^2-s\log s+s + o(s) \right), \frac{1}{\mu_t} = \exp\left( \lambda(\alpha, \beta) t^2-t\log t+t + o(t) \right).
\]

Define the combinatorial ratio:
\[
\Gamma_t = \frac{\binom{s}{t} \binom{N-s}{s-t}}{\binom{N}{s}}.
\]

Using standard combinatorial bounds gives
\[
\Gamma_t \leq \binom{s}{t} \left( \frac{s}{N} \right)^t \leq \left( \frac{e s}{t} \right)^t \left( \frac{s}{N} \right)^t = \left( \frac{e s^2}{t N} \right)^t.
\]

We now analyze the three terms (I), (II), and (III), respectively (note \(s =(1-\epsilon)\frac{\log N}{\lambda(\alpha,\beta)}=\Theta(\log N)\)):
\begin{itemize}
  \item For term (I), since \(\mathbb{E}[Z_s] \to \infty\) when \(s = (1-\epsilon)s_c\), we have
  \[
  (I) = \frac{1}{\mathbb{E}[Z_s]} \to 0 \quad \text{as } N \to \infty.
  \]
  \item For term (II) with \(t = 1, 2\), use the tighter bound:
\[
\Gamma_t \frac{1}{\mu_s} \leq s^t \left( \frac{s}{N} \right)^{s-t} \exp\left( \lambda(\alpha, \beta)s^2 - s \log s + s + o(s) \right).
\]

Simplify the exponent:
\[
\lambda(\alpha, \beta) s^2 - s \log s + s + (s-t) \log s + t \log s - (s-t) \log N = \lambda(\alpha, \beta) s^2 - (s-t) \log N + s .
\]

Substitute \(s = (1 - \epsilon) \frac{\log N}{\lambda(\alpha, \beta)}\):
\[
\lambda(\alpha, \beta) s^2 - (s-t) \log N + s = -\epsilon(1-\epsilon) \frac{(\log N)^2}{\lambda(\alpha, \beta)} + t \log N + \frac{1-\epsilon}{\lambda(\alpha, \beta)} \log N + o(\log N).
\]

The dominant term is \(-\epsilon(1-\epsilon) \frac{(\log N)^2}{\lambda(\alpha, \beta)} < 0\), so:
\[
\Gamma_t \frac{1}{\mu_s} \leq \exp\left( -\Theta((\log N)^2) \right) \to 0.
\]

Thus \((II) \to 0\).
\item For term (III) with \(t \geq 3\), we have
\[
\Gamma_t \frac{1}{\mu_t} \leq \exp\left( t \log\left( \frac{e s^2}{t N} \right) + \lambda(\alpha, \beta) t^2 - t \log t + t + o(t) \right).
\]

The exponent is
\[
\lambda(\alpha, \beta) t^2 + t(1 + 2 \log s - \log t - \log N) - t \log t + t + o(t).
\]

Then substitute \(\log s = \log \log N + \Theta(1)\), we have
\[
\lambda(\alpha, \beta) t^2 - t \log N + 2t \log \log N - 2t \log t + \Theta(t) + o(t).
\]

Given that \(t \leq s = \Theta(\log N)\), and \(\lambda(\alpha, \beta) t \leq (1-\epsilon) \log N\), we get
\[
\lambda(\alpha, \beta) t^2 \leq (1-\epsilon) t \log N \implies \lambda(\alpha, \beta) t^2 - t \log N \leq -\epsilon t \log N.
\]

For any \(\delta > 0\), for large \(N\), we have
\[
2t \log \log N - 2t \log t + \Theta(t) \leq \delta t \log N.
\]

Choosing \(\delta = \epsilon/2\) gives
\[
\lambda(\alpha, \beta) t^2 - t \log N + 2t \log \log N - 2t \log t + \Theta(t) \leq -\frac{\epsilon}{2} t \log N.
\]

Then we get
\[
\Gamma_t \frac{1}{\mu_t} \leq \exp\left( -\frac{\epsilon}{2} t \log N \right) = N^{-\frac{\epsilon}{2} t},
\]
which gives
\[
(III) \leq \sum_{t=3}^{s-1} N^{-\frac{\epsilon}{2} t} \leq \sum_{t=3}^{\infty} N^{-\frac{\epsilon}{2} t} = \frac{N^{-\frac{3\epsilon}{2}}}{1 - N^{-\frac{\epsilon}{2}}} \to 0.
\]
\end{itemize}

Finally, since
\[
(I) \to 0, \quad (II) \to 0, \quad (III) \to 0 \quad \text{as} \quad N \to \infty,
\]
we conclude that
\[
\frac{\text{Var}(Z_s)}{\mathbb{E}[Z_s]^2} \to 0\quad \text{as} \quad N \to \infty,
\]
i.e., \(\text{Var}(Z_s) = o(\mathbb{E}[Z_s]^2)\).
\end{proof}
\section{MATLAB codes of MaxBalanceCore}
The MATLAB codes of MaxBalanceCore are provided below:
\begin{lstlisting}
%% MaxBalanceCore
% LSCBM = MaxBalanceCore(C_tilde, sigma) returns the largest
% strong-correlation balanced module.
% Inputs:
%   C_tilde : N x N statistically validated correlation matrix.
%   sigma   : strength threshold (default 0.7).
% Output:
%   LSCBM   : row vector of node indices (sorted).

function LSCBM = MaxBalanceCore(C_tilde, sigma)
    % Construct filtered signed adjacency matrix
    n = size(C_tilde, 1);
    S = sign(C_tilde) .* (abs(C_tilde) >= sigma);
    S(1:n+1:end) = 0;  % Remove self-loops
    % Compute node impact (degree centrality)
    node_impact = sum(S ~= 0, 2);
    
    % Find maximum balanced clique
    best_clique = [];
    best_size = 0;
    [~, order] = sort(node_impact, 'descend');
    
    for i = 1:min(100, n)
        seed = order(i);
        if node_impact(seed) == 0, continue; end
        
        % Build balanced clique from seed
        [clique, A, B] = build_balanced_clique(S, seed);
        
        % Expand clique and maintain partitions
        [clique, A, B] = expand_clique(S, clique, A, B, sigma);
        
        % Update best solution
        if length(clique) > best_size
            best_clique = clique;
            best_size = length(clique);
        end
    end
    
    LSCBM = sort(best_clique);
end

%% build_balanced_clique
% [clique, A, B] = build_balanced_clique(S, seed) builds an initial
% balanced clique from a seed node.
% Inputs:
%   S    : signed adjacency matrix (-1,0,1).
%   seed : starting node index.
% Outputs:
%   clique : vector of nodes in the initial clique.
%   A, B   : partition sets (A: positive to seed, B: negative to seed).

function [clique, A, B] = build_balanced_clique(S, seed)
    % Initialize partitions
    neighbors = find(S(seed, :) ~= 0);
    A = [seed, neighbors(S(seed, neighbors) > 0)];
    B = neighbors(S(seed, neighbors) < 0);
    
    % Vectorized intra-group filtering
    A = filter_group(S, A, 1);
    B = filter_group(S, B, 1);
    
    % Check inter-group connections (matrix ops instead of loops)
    if ~isempty(A) && ~isempty(B)
        [conflictA, conflictB] = find(S(A, B) >= 0);
        A(unique(conflictA)) = [];
        B(unique(conflictB)) = [];
    end
    
    clique = [A, B];
end

%% filter_group
% group = filter_group(S, group, req_sign) removes nodes in 'group'
% that lack the required sign with all other members.
% Inputs:
%   S        : signed adjacency matrix.
%   group    : current group vector.
%   req_sign : required sign (+1 or -1) for intra-group edges.
% Output:
%   group    : filtered group.

function group = filter_group(S, group, req_sign)
    % Vectorized intra-group filtering
    if numel(group) < 2, return; end
    
    subS = S(group, group);
    mask = ~eye(numel(group));  % Off-diagonal mask
    invalid = any((subS ~= req_sign) & mask, 2);
    group(invalid) = [];
end

%% expand_clique
% [clique, A, B] = expand_clique(S, clique, A, B, sigma) expands an
% existing balanced clique by adding compatible nodes.
% Inputs:
%   S      : signed adjacency matrix.
%   clique : current clique (union of A and B).
%   A, B   : current partitions.
%   sigma  : strength threshold.
% Outputs:
%   clique : expanded clique.
%   A, B   : updated partitions.

function [clique, A, B] = expand_clique(S, clique, A, B, sigma)
    n = size(S, 1);
    candidates = setdiff(1:n, clique);
    if isempty(candidates), return; end
    
    % Precompute connection strength of candidates to current clique
    candidate_strength = all(abs(S(candidates, clique)) >= sigma, 2);
    strong_candidates = candidates(candidate_strength);
    
    for node = strong_candidates
        % Check if can join A: same sign as A, opposite to B
        joinA = (isempty(A) || all(S(node, A) == 1)) && ...
                (isempty(B) || all(S(node, B) == -1));
        
        % Check if can join B: opposite to A, same as B
        joinB = (isempty(A) || all(S(node, A) == -1)) && ...
                (isempty(B) || all(S(node, B) == 1));
        
        if joinA
            A = [A, node];
            clique = [clique, node];
        elseif joinB
            B = [B, node];
            clique = [clique, node];
        end
    end
end
\end{lstlisting}
\bibliographystyle{model5-names}\biboptions{authoryear}
\bibliography{reference}

@article{dong2025social,
	title={A social balance theory-based modeling framework for group-to-empirical decision-making transition with cognitive inertia and trust propagation},
	author={Dong, Jianglin and Zhao, Yiyi and Mu, Shangqun and Mao, Haixia and Hu, Jiangping},
	journal={Expert Systems with Applications},
	pages={129705},
	year={2025},
	publisher={Elsevier}
}

@article{garcia2025improving,
	title={Improving community detection algorithms in directed graphs with fuzzy measures. An application to mobility networks},
	author={Garcia-Pardo, Inmaculada Gutierrez and Perez, Maria Barroso and Gonzalez, Daniel Gomez and Cantalejo, Javier Castro},
	journal={Expert Systems with Applications},
	volume={269},
	pages={126305},
	year={2025},
	publisher={Elsevier}
}

@article{ye2025effect,
	title={The effect of uncertainty index based on sparse method on volatility prediction of stock market},
	author={Ye, Cheng and Ou, HongJing and Basile, Vincenzo and Bhuiyan, Miraj Ahmed},
	journal={Expert Systems with Applications},
	volume={290},
	pages={128208},
	year={2025},
	publisher={Elsevier}
}

@article{sammon2026index,
	title={Index rebalancing and stock market composition: Do indexes time the market?},
	author={Sammon, Marco and Shim, John J},
	journal={Journal of Financial Economics},
	volume={177},
	pages={104229},
	year={2026},
	publisher={Elsevier}
}

@article{NewmanWeighted,
  title = {Analysis of weighted networks},
  author = {Newman, M. E. J.},
  journal = {Phys. Rev. E},
  volume = {70},
  pages = {056131},
  year = {2004}
}

@article{qing2023community,
  title = {Regularized spectral clustering under the mixed membership stochastic block model},
  author = {Huan Qing and Jingli Wang},
  journal = {Neurocomputing},
  volume = {550},
  pages = {126490},
  year = {2023}
}

@article{SC,
  title = {{Spectral clustering and the high-dimensional stochastic blockmodel}},
  author = {Karl Rohe and Sourav Chatterjee and Bin Yu},
  journal = {Annals of Statistics},
  volume = {39},
  pages = {1878 -- 1915},
  year = {2011}
}

@article{RSC,
  title = {Regularized spectral clustering under the degree-corrected stochastic blockmodel},
  author = {Qin, Tai and Rohe, Karl},
  journal = {Advances in Neural Information Processing Systems},
  volume = {26},
  year = {2013}
}

@article{SCORE,
  title = {{Fast community detection by SCORE}},
  author = {Jiashun Jin},
  journal = {Annals of Statistics},
  volume = {43},
  pages = {57 -- 89},
  year = {2015}
}

@article{CMM,
  title = {{Convexified modularity maximization for degree-corrected stochastic block models}},
  author = {Yudong Chen and Xiaodong Li and Jiaming Xu},
  journal = {Annals of Statistics},
  volume = {46},
  pages = {1573 -- 1602},
  year = {2018}
}

@article{qing2025community,
  title={Community detection in multi-layer networks by regularized debiased spectral clustering},
  author={Qing, Huan},
  journal={Engineering Applications of Artificial Intelligence},
  volume={152},
  pages={110627},
  year={2025},
  publisher={Elsevier}
}

@article{zhang2025assessing,
  title={Assessing systemic importance using multilayer dynamic networks: Evidence from China's stock market},
  author={Zhang, Yue and Chen, Haozhi and He, Xiaolei},
  journal={International Review of Economics \& Finance},
  pages={104279},
  year={2025},
  publisher={Elsevier}
}

@article{zhu2025brown,
  title={Are brown stocks valuable to green stocks? Evidence from China},
  author={Zhu, Sha and Fu, Hai and Wei, Yu and Shang, Yue and Chen, Xiaodan},
  journal={Finance Research Letters},
  volume={76},
  pages={106983},
  year={2025},
  publisher={Elsevier}
}

@article{chen2025contagion,
  title={Contagion risk prediction with Chart Graph Convolutional Network: Evidence from Chinese stock market},
  author={Chen, Zhensong and Zhang, Wenjun and Yao, Yinhong},
  journal={Emerging Markets Review},
  pages={101426},
  year={2025},
  publisher={Elsevier}
}

@article{mantegna1999hierarchical,
  title={Hierarchical structure in financial markets},
  author={Mantegna, Rosario N},
  journal={The European Physical Journal B-Condensed Matter and Complex Systems},
  volume={11},
  number={1},
  pages={193--197},
  year={1999},
  publisher={Springer}
}

@article{zhao2025can,
  title={Can we better predict financial crisis? The role of Laplacian-energy-like measure},
  author={Zhao, Xian and Huang, Chuangxia and Yang, Xiaoguang and Cao, Jie and Yang, Xin},
  journal={International Review of Economics \& Finance},
  pages={104396},
  year={2025},
  publisher={Elsevier}
}

@article{ma2015memetic,
  title = {A memetic algorithm for computing and transforming structural balance in signed networks},
  author = {Ma, Lijia and Gong, Maoguo and Du, Haifeng and Shen, Bo and Jiao, Licheng},
  journal = {Knowledge-Based Systems},
  volume = {85},
  pages = {196--209},
  year = {2015}
}

@article{facchetti2011computing,
  title = {Computing global structural balance in large-scale signed social networks},
  author = {Facchetti, Giuseppe and Iacono, Giovanni and Altafini, Claudio},
  journal = {Proceedings of the National Academy of Sciences},
  volume = {108},
  pages = {20953--20958},
  year = {2011}
}

@article{zheng2015social,
  title = {Social balance in signed networks},
  author = {Zheng, Xiaolong and Zeng, Daniel and Wang, Fei-Yue},
  journal = {Information Systems Frontiers},
  volume = {17},
  pages = {1077--1095},
  year = {2015}
}

@article{heider1946attitudes,
  title = {Attitudes and cognitive organization},
  author = {Heider, Fritz},
  journal = {Journal of Psychology},
  volume = {21},
  pages = {107--112},
  year = {1946}
}

@article{cartwright1956structural,
  title = {Structural balance: a generalization of Heider's theory.},
  author = {Cartwright, Dorwin and Harary, Frank},
  journal = {Psychological Review},
  volume = {63},
  pages = {277},
  year = {1956}
}

@article{wang2016optimizing,
  title = {Optimizing dynamical changes of structural balance in signed network based on memetic algorithm},
  author = {Wang, Shanfeng and Gong, Maoguo and Du, Haifeng and Ma, Lijia and Miao, Qiguang and Du, Wei},
  journal = {Social Networks},
  volume = {44},
  pages = {64--73},
  year = {2016}
}

@article{cai2022structure,
  title = {Structure information learning for neutral links in signed network embedding},
  author = {Cai, Shensheng and Shan, Wei and Zhang, Mingli},
  journal = {Information Processing \& Management},
  volume = {59},
  pages = {102917},
  year = {2022}
}

@article{song2022evolutionary,
  title = {Evolutionary prisoner’s dilemma game on signed networks based on structural balance theory},
  author = {Song, Shenpeng and Feng, Yuhao and Xu, Wenzhe and Li, Hui-Jia and Wang, Zhen},
  journal = {Chaos, Solitons \& Fractals},
  volume = {164},
  pages = {112706},
  year = {2022}
}

@article{harary1953notion,
  title = {On the notion of balance of a signed graph.},
  author = {Harary, Frank},
  journal = {Michigan Mathematical Journal},
  volume = {2},
  pages = {143--146},
  year = {1953}
}

@article{chi2010network,
  title = {A network perspective of the stock market},
  author = {Chi, K Tse and Liu, Jing and Lau, Francis CM},
  journal = {Journal of Empirical Finance},
  volume = {17},
  pages = {659--667},
  year = {2010}
}

@article{obilor2018test,
  title = {{Test for significance of Pearson’s correlation coefficient}},
  author = {Obilor, Esezi Isaac and Amadi, Eric Chikweru},
  journal = {International Journal of Innovative Mathematics, Statistics \& Energy Policies},
  volume = {6},
  pages = {11--23},
  year = {2018}
}

@article{edgell1984effect,
  title = {Effect of violation of normality on the t test of the correlation coefficient},
  author = {Edgell, Stephen E and Noon, Sheila M},
  journal = {Psychological Bulletin},
  volume = {95},
  pages = {576},
  year = {1984}
}

@article{de1985does,
  title = {Does the stock market overreact?},
  author = {De Bondt, Werner FM and Thaler, Richard},
  journal = {Journal of finance},
  volume = {40},
  pages = {793--805},
  year = {1985}
}

@article{fama1965behavior,
  title = {The behavior of stock-market prices},
  author = {Fama, Eugene F.}, 
  journal = {Journal of Business},
  volume = {38},
  pages = {34--105},
  year = {1965}
}

@article{chen1986economic,
  title = {Economic forces and the stock market},
  author = {Chen, Nai-Fu and Roll, Richard and Ross, Stephen A},
  journal = {Journal of Business},
  pages = {383--403},
  year = {1986}
}

@article{barsky1993does,
  title = {Why does the stock market fluctuate?},
  author = {Barsky, Robert B and De Long, J Bradford},
  journal = {Quarterly Journal of Economics},
  volume = {108},
  pages = {291--311},
  year = {1993}
}

@article{gordon1959dividends,
  title = {Dividends, earnings, and stock prices},
  author = {Gordon, Myron J},
  journal = {Review of Economics and Statistics},
  volume = {41},
  pages = {99--105},
  year = {1959}
}

@article{jiang2021applications,
  title = {Applications of deep learning in stock market prediction: recent progress},
  author = {Jiang, Weiwei},
  journal = {Expert Systems with Applications},
  volume = {184},
  pages = {115537},
  year = {2021}
}

@article{paramati2017effects,
  title = {{The effects of stock market growth and renewable energy use on CO2 emissions: evidence from G20 countries}},
  author = {Paramati, Sudharshan Reddy and Mo, Di and Gupta, Rakesh},
  journal = {Energy Economics},
  volume = {66},
  pages = {360--371},
  year = {2017}
}

@article{boungou2022impact,
  title = {{The impact of the Ukraine--Russia war on world stock market returns}},
  author = {Boungou, Whelsy and Yati{\'e}, Alhonita},
  journal = {Economics Letters},
  volume = {215},
  pages = {110516},
  year = {2022}
}

@article{arouri2016economic,
  title = {{Economic policy uncertainty and stock markets: Long-run evidence from the US}},
  author = {Arouri, Mohamed and Estay, Christophe and Rault, Christophe and Roubaud, David},
  journal = {Finance Research Letters},
  volume = {18},
  pages = {136--141},
  year = {2016}
}

@article{antonakakis2013dynamic,
  title = {Dynamic co-movements of stock market returns, implied volatility and policy uncertainty},
  author = {Antonakakis, Nikolaos and Chatziantoniou, Ioannis and Filis, George},
  journal = {Economics Letters},
  volume = {120},
  pages = {87--92},
  year = {2013}
}

@article{engle2013stock,
  title = {Stock market volatility and macroeconomic fundamentals},
  author = {Engle, Robert F and Ghysels, Eric and Sohn, Bumjean},
  journal = {Review of Economics and Statistics},
  volume = {95},
  pages = {776--797},
  year = {2013}
}

@article{shah2019stock,
  title = {{Stock market analysis: A review and taxonomy of prediction techniques}},
  author = {Shah, Dev and Isah, Haruna and Zulkernine, Farhana},
  journal = {International Journal of Financial Studies},
  volume = {7},
  pages = {26},
  year = {2019}
}

@article{newman2003structure,
  title = {The structure and function of complex networks},
  author = {Newman, Mark EJ},
  journal = {SIAM Review},
  volume = {45},
  pages = {167--256},
  year = {2003}
}

@article{dorogovtsev2002evolution,
  title = {Evolution of networks},
  author = {Dorogovtsev, Sergey N and Mendes, Jose FF},
  journal = {Advances in Physics},
  volume = {51},
  pages = {1079--1187},
  year = {2002}
}

@article{albert2002statistical,
  title = {Statistical mechanics of complex networks},
  author = {Albert, R{\'e}ka and Barab{\'a}si, Albert-L{\'a}szl{\'o}},
  journal = {Reviews of Modern Physics},
  volume = {74},
  pages = {47},
  year = {2002}
}

@article{samitas2022covid,
  title = {{Covid-19 pandemic and spillover effects in stock markets: A financial network approach}},
  author = {Samitas, Aristeidis and Kampouris, Elias and Polyzos, Stathis},
  journal = {International Review of Financial Analysis},
  volume = {80},
  pages = {102005},
  year = {2022}
}

@article{venturini2022climate,
  title = {{Climate change, risk factors and stock returns: A review of the literature}},
  author = {Venturini, Alessio},
  journal = {International Review of Financial Analysis},
  volume = {79},
  pages = {101934},
  year = {2022}
}

@article{habib2018stock,
  title = {Stock price crash risk: review of the empirical literature},
  author = {Habib, Ahsan and Hasan, Mostafa Monzur and Jiang, Haiyan},
  journal = {Accounting \& Finance},
  volume = {58},
  pages = {211--251},
  year = {2018}
}

@article{kwapien2012physical,
  title = {Physical approach to complex systems},
  author = {Kwapie{\'n}, Jaros{\l}aw and Dro{\.z}d{\.z}, Stanis{\l}aw},
  journal = {Physics Reports},
  volume = {515},
  pages = {115--226},
  year = {2012}
}

@article{Frank2009,
  title = {{Economic Networks: The New Challenges}},
  author = {Frank Schweitzer  and Giorgio Fagiolo  and Didier Sornette  and Fernando Vega-Redondo  and Alessandro Vespignani  and Douglas R. White},
  journal = {Science},
  volume = {325},
  pages = {422-425},
  year = {2009}
}

@article{VIDALTOMAS2021101981,
  title = {{Transitions in the cryptocurrency market during the COVID-19 pandemic: A network analysis}},
  author = {David Vidal-Tomás},
  journal = {Finance Research Letters},
  volume = {43},
  pages = {101981},
  year = {2021}
}

@article{XIA2018222,
  title = {{Comparison between global financial crisis and local stock disaster on top of Chinese stock network}},
  author = {Lisi Xia and Daming You and Xin Jiang and Quantong Guo},
  journal = {Physica A: Statistical Mechanics and its Applications},
  volume = {490},
  pages = {222-230},
  year = {2018}
}

@article{HE2022121732,
  title = {{Sudden shock and stock market network structure characteristics: A comparison of past crisis events}},
  author = {Chengying He and Zhang Wen and Ke Huang and Xiaoqin Ji},
  journal = {Technological Forecasting and Social Change},
  volume = {180},
  pages = {121732},
  year = {2022}
}

@article{Acemoglu2015,
  title = {{Systemic Risk and Stability in Financial Networks}},
  author = {Acemoglu, Daron and Ozdaglar, Asuman and Tahbaz-Salehi, Alireza},
  journal = {American Economic Review},
  volume = {105},
  pages = {564--608},
  year = {2015}
}

@article{HEIBERGER2014376,
  title = {Stock network stability in times of crisis},
  author = {Raphael H. Heiberger},
  journal = {Physica A: Statistical Mechanics and its Applications},
  volume = {393},
  pages = {376-381},
  year = {2014}
}

@article{Heiberger2018,
  title = {Predicting economic growth with stock networks},
  author = {Heiberger, Raphael H.},
  journal = {Physica A: Statistical Mechanics and its Applications},
  volume = {489},
  pages = {102--111},
  year = {2018}
}

@article{Memon2019,
  title = {{Structural Change and Dynamics of Pakistan Stock Market During Crisis: A Complex Network Perspective}},
  author = {Memon, Bilal Ahmed and Yao, Hongxing},
  journal = {Entropy},
  volume = {21},
  pages = {248},
  year = {2019}
}

@article{EOM20171,
  title = {Effects of common factors on stock correlation networks and portfolio diversification},
  author = {Cheoljun Eom and Jong Won Park},
  journal = {International Review of Financial Analysis},
  volume = {49},
  pages = {1-11},
  year = {2017}
}

@article{ESMAEILPOURMOGHADAM2019121800,
  title = {{Complex networks analysis in Iran stock market: The application of centrality}},
  author = {Hadi Esmaeilpour Moghadam and Teymour Mohammadi and Mohammad Feghhi Kashani and Abbas Shakeri},
  journal = {Physica A: Statistical Mechanics and its Applications},
  volume = {531},
  pages = {121800},
  year = {2019}
}

@article{ZHANG2019748,
  title = {The stability of Chinese stock network and its mechanism},
  author = {Weiping Zhang and Xintian Zhuang},
  journal = {Physica A: Statistical Mechanics and its Applications},
  volume = {515},
  pages = {748-761},
  year = {2019}
}

@article{YANG2022103180,
  title = {An empirical study of risk diffusion in the cryptocurrency market based on the network analysis},
  author = {Ming-Yuan Yang and Zhen-Guo Wu and Xin Wu},
  journal = {Finance Research Letters},
  volume = {50},
  pages = {103180},
  year = {2022}
}

@article{MAJAPA201635,
  title = {{Topology of the South African stock market network across the 2008 financial crisis}},
  author = {Mohamed Majapa and Sean Joss Gossel},
  journal = {Physica A: Statistical Mechanics and its Applications},
  volume = {445},
  pages = {35-47},
  year = {2016}
}

@article{NOBI2014135,
  title = {Effects of global financial crisis on network structure in a local stock market},
  author = {Ashadun Nobi and Seong Eun Maeng and Gyeong Gyun Ha and Jae Woo Lee},
  journal = {Physica A: Statistical Mechanics and its Applications},
  volume = {407},
  pages = {135-143},
  year = {2014}
}

@article{CHEN2022100836,
  title = {{Identifying systemically important financial institutions in complex network: A case study of Chinese stock market}},
  author = {Wei Chen and Xiaoli Hou and Manrui Jiang and Cheng Jiang},
  journal = {Emerging Markets Review},
  volume = {50},
  pages = {100836},
  year = {2022}
}

@article{BOGINSKI20063171,
  title = {{Mining market data: A network approach}},
  author = {Vladimir Boginski and Sergiy Butenko and Panos M. Pardalos},
  journal = {Computers \& Operations Research},
  volume = {33},
  pages = {3171-3184},
  year = {2006}
}

@article{YANG2024101138,
  title = {{Influential risk spreaders and systemic risk in Chinese financial networks}},
  author = {Ming-Yuan Yang and Zhen-Guo Wu and Xin Wu and Sai-Ping Li},
  journal = {Emerging Markets Review},
  volume = {60},
  pages = {101138},
  year = {2024}
}

@article{QU2022111939,
  title = {Identification of the most influential stocks in financial networks},
  author = {Junyi Qu and Ying Liu and Ming Tang and Shuguang Guan},
  journal = {Chaos, Solitons \& Fractals},
  volume = {158},
  pages = {111939},
  year = {2022}
}

@article{Liang2024,
  title = {{Evolution of Complex Network Topology for Chinese Listed Companies Under the COVID-19 Pandemic}},
  author = {Kaihao Liang and Shuliang Li and Wenfeng Zhang and Zhuokui Wu and Jiaying He and Mengmeng Li and Yuling Wang},
  journal = {Computational Economics},
  volume = {63},
  pages = {1121--1136},
  year = {2024}
}

@article{Wang2018,
  title = {{Correlation Structure and Evolution of World Stock Markets: Evidence from Pearson and Partial Correlation-Based Networks}},
  author = {Wang, Gang-Jin and Xie, Chi and Stanley, H. Eugene},
  journal = {Computational Economics},
  volume = {51},
  pages = {607--635},
  year = {2018}
}

@article{LI2020101185,
  title = {{Analysis of the impact of Sino-US trade friction on China’s stock market based on complex networks}},
  author = {Yanshuang Li and Xintian Zhuang and Jian Wang and Weiping Zhang},
  journal = {North American Journal of Economics and Finance},
  volume = {52},
  pages = {101185},
  year = {2020}
}

@article{QING2025112769,
  title = {Community detection by spectral methods in multi-layer networks},
  author = {Huan Qing},
  journal = {Applied Soft Computing},
  volume = {171},
  pages = {112769},
  year = {2025}
}

@article{MASUDA20251,
  title = {{Introduction to correlation networks: Interdisciplinary approaches beyond thresholding}},
  author = {Naoki Masuda and Zachary M. Boyd and Diego Garlaschelli and Peter J. Mucha},
  journal = {Physics Reports},
  volume = {1136},
  pages = {1-39},
  year = {2025}
}

@article{LU20161,
  title = {Vital nodes identification in complex networks},
  author = {Linyuan Lü and Duanbing Chen and Xiao-Long Ren and Qian-Ming Zhang and Yi-Cheng Zhang and Tao Zhou},
  journal = {Physics Reports},
  volume = {650},
  pages = {1-63},
  year = {2016}
}

@article{tabassum2018social,
  title = {{Social network analysis: An overview}},
  author = {Tabassum, Shazia and Pereira, Fabiola SF and Fernandes, Sofia and Gama, Jo{\~a}o},
  journal = {Wiley Interdisciplinary Reviews: Data Mining and Knowledge Discovery},
  volume = {8},
  pages = {e1256},
  year = {2018}
}

@article{peng2018influence,
  title = {{Influence analysis in social networks: A survey}},
  author = {Peng, Sancheng and Zhou, Yongmei and Cao, Lihong and Yu, Shui and Niu, Jianwei and Jia, Weijia},
  journal = {Journal of Network and Computer Applications},
  volume = {106},
  pages = {17--32},
  year = {2018}
}

@article{boccaletti2006complex,
  title = {Complex networks: Structure and dynamics},
  author = {Boccaletti, Stefano and Latora, Vito and Moreno, Yamir and Chavez, Martin and Hwang, D-U},
  journal = {Physics Reports},
  volume = {424},
  pages = {175--308},
  year = {2006}
}

@article{huang2009network,
  title = {{A network analysis of the Chinese stock market}},
  author = {Huang, Wei-Qiang and Zhuang, Xin-Tian and Yao, Shuang},
  journal = {Physica A: Statistical Mechanics and its Applications},
  volume = {388},
  pages = {2956--2964},
  year = {2009}
}

@article{LIU2024122529,
  title = {{Statistical analysis of the regional air quality index of Yangtze River Delta based on complex network theory}},
  author = {Jia-Bao Liu and Ya-Qian Zheng and Chien-Chiang Lee},
  journal = {Applied Energy},
  volume = {357},
  pages = {122529},
  year = {2024}
}

@article{Yan2023,
  title = {{Community detection for New York stock market by SCORE-CCD}},
  author = {Yan, Yanan and Yang, Yuehan},
  journal = {Computational Statistics},
  volume = {38},
  pages = {1255--1282},
  year = {2023}
}

@article{Li2022,
  title = {{Undirected and Directed Network Analysis of the Chinese Stock Market}},
  author = {Li, Binghui and Yang, Yuehan},
  journal = {Computational Economics},
  volume = {60},
  pages = {1155--1173},
  year = {2022}
}

@article{ZHOU2023118944,
  title = {Dynamic analysis and community recognition of stock price based on a complex network perspective},
  author = {Yingrui Zhou and Zengqiang Chen and Zhongxin Liu},
  journal = {Expert Systems with Applications},
  volume = {213},
  pages = {118944},
  year = {2023}
}

@article{xing2023community,
  title = {Community detection and clustering characteristics analysis of the stock market},
  author = {Xing, Jia and Li, Binghui and Yang, Yuehan},
  journal = {Managerial and Decision Economics},
  volume = {44},
  pages = {3893--3906},
  year = {2023}
}

@article{chen2025systemic,
  title={Systemic risk and network effects in RCEP financial markets: Evidence from the TEDNQR model},
  author={Chen, Yan and Luo, Qiong and Zhang, Feipeng},
  journal={The North American Journal of Economics and Finance},
  volume={76},
  pages={102317},
  year={2025},
  publisher={Elsevier}
}

\end{document}